%% file: main.tex
\documentclass[11pt,final]{article}

\usepackage{xcolor}

\usepackage{fullpage}
\usepackage{amsmath,amssymb,amsthm}
\usepackage{lmodern}

\usepackage[numbers,sort&compress]{natbib} 
\setcitestyle{numbers}

\usepackage[colorlinks=true]{hyperref}
\hypersetup{
     colorlinks  = true,
     urlcolor    = teal,
	 citecolor   = teal,
	 linkcolor   = red
}

\usepackage[ruled,noend,linesnumbered]{algorithm2e} 

\SetAlFnt{\small}
\SetAlCapFnt{\small}
\SetAlCapNameFnt{\small}
\SetAlCapHSkip{0pt}
\IncMargin{-\parindent}

\usepackage[utf8]{inputenc} 
\usepackage[T1]{fontenc}    
\usepackage{lmodern}
\usepackage{url}            
\usepackage{amsfonts}       
\usepackage{nicefrac}       
\usepackage{microtype}      
\usepackage{booktabs}       

\newtheorem{theorem}{Theorem}
\newtheorem{lemma}{Lemma}
\newtheorem{proposition}{Proposition}
\newtheorem{corollary}{Corollary}
\newtheorem{conjecture}{Conjecture}
\newtheorem{definition}{Definition}
\newtheorem{remark}{Remark}
\newtheorem{property}{Property}
\newtheorem{claim}{Claim}
\newtheorem{observation}{Observation}
\newtheorem{assumption}{Assumption}

\theoremstyle{definition} 
\newtheorem{example}{Example}

\newcount\Includeappendix 
\newcommand{\appnx}[1]{{\ifnum\Includeappendix=1{#1}\else{the appendix}\fi}}

\input{def}
\begin{document}

\title{Envious Explore and Exploit}

\author{
Omer Ben{-}Porat%
\thanks{%
    {Technion---Israel Institute of Technology (\url{omerbp@technion.ac.il})}, corresponding author}
\and Yotam Ganfi%
\thanks{%
    {Weizmann Institute (\url{yotam.gafni@gmail.com})}}
\and Or Markovetzki%
\thanks{%
    {Technion---Israel Institute of Technology (\url{ormar@campus.technion.ac.il})}}
}
\maketitle
\begin{abstract}
Explore-and-exploit tradeoffs play a key role in recommendation systems (RSs), aiming at serving users better by learning from previous interactions. Despite their commercial success, the societal effects of explore-and-exploit mechanisms are not well understood, especially regarding the utility discrepancy they generate between different users.
In this work, we measure such discrepancy using the economic notion of envy. We present a multi-armed bandit-like model in which every round consists of several sessions, and rewards are realized once per round. We call the latter property \emph{reward consistency}, and show that the RS can leverage this property for better societal outcomes. On the downside, doing so also generates envy, as late-to-arrive users enjoy the information gathered by early-to-arrive users. We examine the generated envy under several arrival order mechanisms and virtually any \emph{anonymous} algorithm, i.e., any algorithm that treats all similar users similarly without leveraging their identities. We provide tight envy bounds on uniform arrival and upper bound the envy for nudged arrival, in which the RS can affect the order of arrival by nudging its users. Furthermore, we study the efficiency-fairness trade-off by devising an algorithm that allows constant envy and approximates the optimal welfare in restricted settings. Finally, we validate our theoretical results empirically using simulations.
\end{abstract}

\input{sections/intro} 
\input{sections/model}

\input{sections/uniform}

\input{sections/nudge}

\input{sections/efc}

\input{sections/discussion}

\section*{Acknowledgments}\label{sec:Acknowledgments}
The work of O. Ben-Porat was supported by the Israel Science Foundation (ISF; Grant No. 3079/24). 

\bibliographystyle{plainnat}

\input{main.bbl}
\ifnum\Includeappendix=1{
\appendix
\input{sections/appendix-uniform}

\input{sections/appendix-nudge}

\input{sections/appendix-nudge-models}
\input{sections/appendix-efc}
\input{sections/appendix-average}

\input{sections/appendix-socially-opt}
\input{sections/appendix-simulations}

}%
\fi

\end{document}

%% file: def.tex
\usepackage{xcolor}
\usepackage{algpseudocode}
\usepackage{longtable}
\usepackage{enumitem}
\usepackage{mathtools}
\usepackage{amsmath}
\usepackage{nicefrac}

\newcommand{\ifapp}[2]{\ifnum\Includeappendix=1{#1}\else{#2}\fi}

\usepackage{subcaption}

\algtext*{EndFor}
\algtext*{EndIf}
\algrenewcommand\algorithmicrequire{\textbf{Input:}}

\newcommand\abs[1]{\left\lvert#1\right\rvert}
\DeclareMathOperator*{\argmax}{arg\,max}

\newcommand{\ind}[1]{\mathbb{I}_{#1} }

\newcommand{\E}[1]{\mathbb{E}\left[ #1 \right]}
\newcommand{\var}[1]{\mathrm{Var}\left( #1 \right)}

\newcommand{\uni}[1]{\text{Uni}\left(#1\right)}
\newcommand{\prb}[1]{\Pr\left( #1 \right)}
\newcommand{\ber}[1]{\text{Bernoulli}\!\left(#1\right)}

\newcommand{\Del}[2]{\Delta^{#1}_{#2}}

\newcommand{\rt}[1]{r_{#1}^t}
\newcommand{\Rt}[1]{R_{#1}^t}

\newcommand{\dif}{\Delta}
\newcommand{\tdif}{\tilde\Delta}
\newcommand{\dift}{{\dif^t}}
\newcommand{\adif}[2]{\dif_{#1,#2}}
\newcommand{\adift}[2]{\dif^t_{#1,#2}}
\newcommand{\sdif}[2]{\dif_{\left(#1\right),\left(#2\right)}}
\newcommand{\sdift}[2]{\dif^t_{\left(#1\right),\left(#2\right)}}

\newcommand{\env}{\mathcal{E}}
\newcommand{\sw}{\textit{SW}}

\newcommand{\envavg}{\env_{\textnormal{avg}}}

\newcommand{\ordv}{\boldsymbol{\eta}}
\newcommand{\sigv}{\boldsymbol{\sigma}}

\newcommand{\ordr}{\boldsymbol{\eta}}
\newcommand{\ord}{\eta}

\newcommand{\ordname}{\mathcal O}
\newcommand{\uniord}{\ordname_{\textnormal{Uni}}}
\newcommand{\advord}{\ordname_{\textnormal{Adv}}} 

\newcommand{\sugord}{\ordname_{\textnormal{Ndg}}}

\newcommand{\sopt}{2OPT}
\newcommand{\is}{{i^\star}}
\newcommand{\js}{{j^\star}}

\newcommand{\ef}{EF1}
\newcommand{\efc}{EFC}

%% file: sections/intro.tex
\section{Introduction}
\label{chap: intro}
Multi-armed bandits (MABs) are a cornerstone of machine learning.
In this framework, a learner makes decisions sequentially, choosing an arm (i.e., an action) and subsequently receiving a reward in each round.
Central to MABs is the dilemma of exploiting known arms with historically high rewards and exploring lesser-known arms that could yield even higher rewards.
The literature includes diverse variants of the problem, including contextual bandits~\cite{chu2011contextual, slivkins2011contextual}, dueling bandits~\cite{yue2012k, cohen2021dueling}, non-stationary bandits~\cite{besbes2014stochastic, levine2017rotting}, and Lipschitz bandits~\cite{kleinberg2008multi}.
The main challenge in these works is to overcome reward uncertainty, compared to ideal scenarios with complete information.

However, with the widespread commercial deployment of explore-and-exploit systems, concerns have expanded beyond mere reward uncertainty to encompass broader societal and fairness issues.
Recent research considers facets of fairness, such as regulatory requirements \cite{patil2021achieving, bahar2020fiduciary,ron2021corporate}, fairness of exposure~\cite{ wang2021fairness, liu2017calibrated}, and meritocratic fairness \cite{joseph2016fairness}. While much of this work has tackled fairness from the perspective of the ``arms,'' user-centric fairness remains relatively underexplored. In particular, the concept of \textit{envy}, a well-established theme in fair division literature~\cite{moulin2004fair,procaccia2013cake}, has been somewhat overlooked in this context (see notable exceptions in Subsection~\ref{sec: rw}). This is especially important because perceived unfairness, as captured by envy, can significantly affect user satisfaction and the overall trustworthiness of the system.

In this work, we contribute to the research on envy within explore-and-exploit systems by modeling two key behaviors that characterize real-world, user-centric applications:
\begin{enumerate}[wide, labelwidth=0pt, labelindent=0pt]
    \item \textit{Recurring users}: Rather than being one-time visitors, users often return to the system repeatedly. For example, navigation apps experience periodic engagement as users rely on them regularly.
    \item \textit{Reward consistency}: Although rewards may fluctuate stochastically, they tend to remain stable over specific periods. For instance, the quality of a restaurant’s daily special is generally consistent throughout the day, and travel times on a frequently used route remain largely unchanged over short intervals.
\end{enumerate}
These assumptions introduce nuances in the user experience that foster envy. Due to the dynamics of exploration, some users may explore more frequently and consequently receive lower rewards. This discrepancy in user experience can lead to reluctance to engage with the system or follow its recommendations.

\subsection{Our Contribution}
\label{sec: contribution}
Our contribution is two-fold. Conceptually, we advance the understanding of envy dynamics in explore-and-exploit systems using a stylized model. Our framework encompasses $N$ homogeneous agents, $K$ arms and $T$ \emph{rounds}. Each round comprises multiple \emph{sessions}, wherein distinct agents interact with the system. During each session, the algorithm selects an arm on behalf of the corresponding agent. We assume that arm rewards are independently sampled in every round but are \emph{consistent} across all sessions within a given round. Within this framework, algorithms typically aim to optimize an aggregate performance metric---such as maximizing social welfare, enhancing risk-adjusted outcomes, or minimizing cumulative regret. Consequently, each round involves both exploration---unveiling rewards for the current session---and exploitation, wherein past session information is leveraged to benefit subsequent agents. 

In this context, agents participating in earlier sessions explore more frequently, while agents arriving later benefit from the information gathered by their predecessors.
This gives rise \emph{envy}, defined as the maximal disparity in cumulative rewards among agents. As expected, the manner in which agents enter a round plays a pivotal role in influencing envy.
We formalize this arrival process as an \emph{arrival function}, which dictates the sequence of agent arrivals within a round. We propose three arrival functions: Uniform, denoted as $\uniord$, which uniformly determines the order of agent arrival for each round; nudged arrival, denoted as $\sugord$, allowing the system to influence the order of arrival, potentially mitigating envy across successive rounds; and adversarial arrival, denoted as $\advord$, establishing a worst-case scenario for agent arrival order from an envy perspective.

Our second contribution is technical: We characterize the envy of \emph{anonymous} algorithms, i.e., invariant under any permutation of the agent order. Focusing on such algorithms is justified since agents have homogeneous preferences, and their identities are irrelevant for optimizing any aggregated metric. For the uniform arrival $\uniord$, Section~\ref{sec:uniform} establishes an upper bound of $O\Big(\sqrt{\ln{(N)} \sum^{T}_{t=1}{\var{\dift}} }\Big)$, where $\Delta^t$ is the average reward discrepancy in round $t$ and $N$ is the number of agents. Crucially, $\Delta^t$ depends on both the algorithm at hand and the reward distributions. Under mild assumptions, this bound becomes $O( \sqrt{\nicefrac{TK\ln(N)}{N}})$. We also prove an (almost) matching lower bound of $\Omega\Big(\sqrt{\sum^{T}_{t=1}{\var{\dift}} }\Big)$. 

For the nudged arrival $\sugord$, Section~\ref{sec:nudge} derives an upper bound of $O(\nicefrac{N}{\delta})$, where $\delta \in (0,1)$ measures the degree to which the system can influence the arrival order. Surprisingly, this bound only depends on the number of agents but not the number of rounds $T$. We further establish a tight linear bound of $\Theta(T)$ for the adversarial arrival, with constants hidden in the bound depending on the specific instance. Our core techniques heavily rely on martingale and excursion theory, enabling the analysis of (possibly dependent) reward discrepancies. 

While our primary aim is to understand and characterize the envy generated by any arbitrary algorithm, we also take initial steps toward balancing social welfare and envy. In Section~\ref{sec:extensions}, we focus on uniform rewards and introduce an algorithm that can attain any point on the efficiency-fairness Pareto frontier. Formally, it receives a positive number $C \in \mathbb R_{+}$, and ensures that envy would not grow beyond $C$ almost surely. We show that for $C=1$, namely, an envy of at most 1,  our algorithm is able to gain almost half of the additional welfare compared to a no-envy algorithm.

\subsection{Related Work}
\label{sec: rw}
From the economic perspective, there has been extensive research in the realm of fairness~\cite{conitzer2017fair, Kahneman86, shapley_shubik_1954, dubey1975uniqueness,shapley1974cores, shapley1971cores}. Envy-freeness is a central concept in the fair division and allocation literature, for instance, in cake-cutting~\cite{Steven95,aziz2016discrete,cohler2011optimal}. Envy-freeness refers to a situation where no individual envies the allocation of another. 
Building on this foundation, a growing body of research on \emph{dynamic fair allocation} \cite{hassanzadeh2023sequential, benade2018make,kash2014no,zeng2020tradeoffs,sinclair2022sequential}, exemplified by applications such as food bank operations, explores ways for allocating resources dynamically. 

Within this literature, \citet{benade2018make} is the work most related to ours. The authors model the dynamic fair division of $T$ indivisible goods that arrive online and must be immediately allocated to exactly $n$ heterogeneous agents who assign a value in $[0,1]$ to each item. Their goal is to design allocation algorithms that minimize the maximum envy at time $T$, where envy is defined as the largest difference between any agent's overall value for items allocated to another agent and her own. They design an efficient algorithm with an envy of $\tilde O(\sqrt{T})$, and show that this guarantee is asymptotically optimal. Despite the many alignments in the underlying mathematical framework and results, there are significant distinctions between the work of \citet{benade2018make} and our work.  

First and foremost, while envy in \citet{benade2018make} arises from heterogeneous valuations among agents, our model involves homogeneous agents. To emphasize this point, in the model of  \citet{benade2018make}, if agents were homogeneous, then simply allocating each item to the agent with the lowest cumulative reward would yield an envy of at most one and not $\tilde O(\sqrt{T})$. Second, in our work, envy is generated due to arrival order rather than solely by allocation decisions like \citet{benade2018make}. In our setting, the same arm (or resource) can be assigned to several agents, with its utility being revealed only after it is pulled for the first agent; thus, envy emerges from the inherent explore-and-exploit tradeoffs. Third, while \citet{benade2018make} aim \emph{to design} fair algorithms, we focus on \emph{measuring} envy under virtually any algorithm and different arrival order mechanisms. Indeed, in our work, the decision-making algorithm could operate optimistically like UCB~\cite{auer2002finite}, follow Thompson sampling~\cite{agrawal2012analysis}, or leverage deep learning techniques to learn reward distribution over time. 

From the machine learning perspective, fairness has attracted increasing interest over the past decade (see, e.g., \cite{Mehrabi19,caton2020fairness,pessach2022review} for recent surveys).
Within the domain of fair classification~\cite{Kleinberg16}, several works apply a fair division lens to evaluate the equity of predicted classes relative to protected attributes~\cite{Hossain20-2, ben2021protecting}. Some works consider fairness in MABs but overlook multi-agent complexities~\cite{wang2021fairness,liu2017calibrated}, and other works consider multi-agent MABs but sidestep fairness considerations~\cite{chakraborty2017coordinated,bargiacchi2018learning}. 

A work that is similar to ours is \citet{Hossain20} and its extensions~\cite{jones2023efficient}.
The authors introduce a multi-agent MABs model, in which a pull of an arm simultaneously serves multiple agents.
The objective of \cite{Hossain20} is minimizing regret with respect to the \textit{Nash social welfare}, a celebrated concept in the economic literature~\cite{kaneko1979nash}. Despite apparent parallels, our work and \cite{Hossain20} differ fundamentally.
For instance, while they assume a single arm is pulled each round to serve all agents, we permit different agents to pull different arms within a round.
Moreover, our model captures intra-round information exploitation due to the consistency of rewards, an angle that they do not model.
Their emphasis is on the Nash product as an objective, known for its envy-freeness properties (EF1)~\cite{caragiannis2019unreasonable}. We, on the other hand, aim to measure the envy induced by any anonymous algorithm.

Another pertinent line of work involves incentivizing exploration~\cite{Kremer14,bahar2020fiduciary,frazier2014incentivizing, Mansour15,simchowitz2024exploration}.
In the framework proposed by~\cite{Kremer14}, algorithms can suggest which arms agents should pull. However, agents are strategic and might opt not to follow the algorithm's recommendation, particularly when it comes to exploration.
This necessitates the design of incentive-compatible mechanisms~\cite{Roughgarden10}. Interestingly, this reluctance to embrace exploration is also central to our work because exploration inherently induces envy. Nevertheless, our approach diverges; we assume agents are non-strategic and the algorithm is in charge of arm-pulling, treating envy as a quantifiable undesired outcome.

Finally, while this work focuses on analyzing envy dynamics for arbitrary algorithms, designing socially optimal algorithms under our model is nontrivial. Because rewards remain constant for extended periods, the problem shares many characteristics with the Pandora’s Box framework \cite{weitzman1978optimal,boodaghians2020pandora,berger2023pandora,amanatidis2024pandora} and prophet inequalities \cite{agrawal2020optimal,brustle2024competition}. We expand on these connections in Section~\ref{sec:discussion}.




%% file: sections/model.tex
\section{Model}
\label{sec:model}
Let $[N] = \{1, 2, \dots, N \}$ be a set of $N$ agents.
We examine an environment in which a system interacts with the agents over $T$ rounds.
Every round $t\leq T$ comprises $N$ \emph{sessions}, each session represents an encounter of the system with exactly one agent, and each agent interacts exactly once with the system every round.
I.e., in each round $t$ the agents arrive sequentially.

\paragraph{Arrival order} The \emph{arrival order} of round $t$, denoted as $\ordv_t=(\ord_t(1),\dots, \ord_t(N))$, is an element from set of all permutations of $[N]$. Each entry $q$ in $\ordv_t$ is the index of the agent that arrives in the $q^{\text{th}}$ session of round $t$.
For example, if $\ord_t(1) = 2$, then agent $2$ arrives in the first session of round $t$.
Correspondingly, $\ord_t^{-1}(i)=q$ implies that agent $i$ arrives in the $q^{\text{th}}$ session of round $t$. 

As we demonstrate later, the arrival order has an immediate impact on agent rewards. We call the mechanism by which the arrival order is set \emph{arrival function} and denote it by $\ordname$. Throughout the paper, we consider several arrival functions such as the \emph{uniform arrival} function, denoted by $\uniord$, and the \emph{nudged arrival} $\sugord$; we introduce those formally in Sections~\ref{sec:uniform} and~\ref{sec:nudge}, respectively.


\paragraph{Arms} We consider a set of $K \geq 2$ arms, $A = \{a_1, \ldots, a_K\}$. The reward of arm $a_i$ in round $t$ is a random variable $X_i^t \sim \mathcal{D}^t_i$, where the rewards $(X_i^t)_{i,t}$ are mutually independent and bounded within the interval $[0,1]$. The reward distribution $\mathcal{D}^t_i$ of arm $a_i$, $i\in [K]$ at round $t\in T$ is assumed to be non-stationary but independent across arms and rounds. We denote the realized reward of arm $a_i$ in round $t$ by $x_i^t$. We assume \emph{reward consistency}, meaning that rewards may vary between rounds but remain constant within the sessions of a single round. Specifically, if an arm $a_i$ is selected multiple times during round~$t$, each selection yields the same reward $x_i^t$, where the superscript $t$ indicates its dependence on the round rather than the session. This consistency enables the system to leverage information obtained from earlier sessions to make more informed decisions in later sessions within the same round. We provide further details on this principle in Subsection~\ref{subsec:information}.

\paragraph{Algorithms} An algorithm is a mapping from histories to actions. We typically expect algorithms to maximize some aggregated agent metric like social welfare. Let $\mathcal H^{t,q}$ denote the information observed during all sessions of rounds $1$ to $t-1$ and sessions $1$ to $q-1$ in round $t$.  The history $\mathcal H^{t,q}$ is an element from $(A \times [0,1])^{(t-1) \cdot N +q-1}$, consisting of pairs of the form (pulled arm, realized reward). Notice that we restrict our attention to \emph{anonymous} algorithms, i.e., algorithms that do not distinguish between agents based on their identities. Instead, they only respond to the history of arms pulled and rewards observed, without conditioning on which specific agent performed each action.
Furthermore, we sometimes assume that algorithms have Bayesian information, meaning they are aware of the reward distributions $(\mathcal{D}^t_i)_{i,t}$. If such an assumption is required to derive a result, we make it explicit. 

\paragraph{Rewards} Let $\rt{i}$ denote the reward received by agent $i \in [N]$ at round $t$, and let $\Rt{i}$ denote her cumulative reward at the end of round $t$, i.e., $\Rt{i} = \sum_{\tau=1}^{t}{r^{\tau}_{i}}$. We further denote the \emph{social welfare} as the sum of the rewards all agents receive after $T$ rounds. Formally, $\sw=\sum^{N}_{i=1}{R^T_i}$. We emphasize that social welfare is independent of the arrival order.

\paragraph{Envy}
We denote by $\adift{i}{j}$ the reward discrepancy of agents $i$ and $j$ in round $t$; namely, $\adift{i}{j}= \rt{i} - \rt{j}$. 
The (cumulative) \emph{envy} between two agents at round $t$ is the difference in their cumulative rewards. Formally, $\env_{i,j}^t= \Rt{i} - \Rt{j}$ is the envy after $t$ rounds between agent $i$ and $j$. We can also formulate envy as the sum of reward discrepancies, $\env_{i,j}^t= \sum^{t}_{\tau=1}{\adif{i}{j}^\tau}$. Notice that envy is a signed quantity and can be either positive or negative. Specifically, if $\env_{i,j}^t < 0$, we say that agent $i$ envies agent $j$, and if $\env_{i,j}^t > 0$, agent $j$ envies agent $i$. The main goal of this paper is to investigate the behavior of the \emph{maximal envy}, defined as
\[
\env^t = \max_{i,j \in [N]} \env^t_{i,j}.
\]
For clarity, the term \emph{envy} will refer to the maximal envy.\footnote{ We address alternative definitions of envy in Section~\ref{sec:discussion}.} 
Note that $\env_{i,j}^t$ are random variables that depend on the decision-making algorithm, realized rewards, and the arrival order, and therefore, so is $\env^t$. If a result we obtain regarding envy depends on the arrival order $\ordname$, we write $\env^t(\ordname)$. Similarly, to ease notation, if $\ordname$ can be understood from the context, it is omitted.

\paragraph{Further Notation} We use the subscript $(q)$ to address elements of the $q^{\text{th}}$ session, for $q\in [N]$.
That is, we use the notation $\rt{(q)}$ to denote the reward granted to the agent that arrives in the $q^{\text{th}}$ session of round $t$ and $\Rt{(q)}$ to denote her cumulative reward. 
Correspondingly, $\sdift{q}{w} = \rt{(q)} - \rt{(w)}$ is the reward discrepancy of the agents arriving in the $q^{\text{th}}$ and $w^{\text{th}}$ sessions of round $t$, respectively. 
To distinguish agents, arms, sessions and rounds, we use the letters $i,j$ to mark agents and arms, $q,w$ for sessions, and $t,\tau$ for rounds.

\subsection{Example}
\label{sec: example}
To illustrate the proposed setting and notation, we present the following example, which serves as a running example throughout the paper.

\begin{table}[t]
\centering
\begin{tabular}{|c|c|c|c|}
\hline
$t$ (round) & $\ordv_t$ (arrival order) & $x_1^t$ & $x_2^t$ \\ \hline
1           & 2, 1                     & 0.6     & 0.92    \\ \hline
2           & 1, 2                     & 0.48    & 0.1     \\ \hline
3           & 2, 1                     & 0.15    & 0.8     \\ \hline
\end{tabular}
\caption{
    Data for Example~\ref{example 1}.
}
\label{tbl: example}
\end{table}

\begin{algorithm}[t]
\caption{Algorithm for Example~\ref{example 1}}
\label{alguni}
\DontPrintSemicolon 
\For{round $t = 1$ to $T$}{
    pull $a_{1}$ in the first session\label{alguniexample: first}\\
    \lIf{$x^t_1 \geq \frac{1}{2}$}{pull $a_{1}$ again in second session \label{alguniexample: pulling a again}}
    \lElse{pull $a_{2}$ in second session \label{alguniexample: sopt else}}
}
\end{algorithm}

\begin{example}\label{example 1}
We consider $K=2$ uniform arms, $X_1,X_2 \sim \uni{0,1}$, and $N=2$ for some $T\geq 3$. We shall assume arm decision are made by Algorithm~\ref{alguni}: In the first session, the algorithm pulls $a_1$; if it yields a reward greater than $\nicefrac{1}{2}$, the algorithm pulls it again in the second session (the ``if'' clause). Otherwise, it pulls $a_2$.

We further assume that the arrival orders and rewards are as specified in Table~\ref{tbl: example}. Specifically, agent 2 arrives in the first session of round $t=1$, and pulling arm $a_2$ in this round would yield a reward of $x^1_2 = 0.92$. Importantly, \emph{this information is not available to the decision-making algorithm in advance} and is only revealed when or if the corresponding arms are pulled.

In the first round, $\boldsymbol{\eta}^1 = \left(2,1\right)$; thus, agent 2 arrives in the first session.
The algorithm pulls arm $a_1$, which means, $a^1_{(1)} = a_1$, and the agent receives $r_{2}^1=r_{(1)}^1=x_1^1=0.6$.
Later that round, in the second session, agent 1 arrives, and the algorithm pulls the same arm again since $x^1_1 = 0.6 \geq \nicefrac{1}{2}$ due to the ``if'' clause.
I.e., $a^1_{(2)} = a_1$ and $r_{1}^1 = r_{(2)}^1 = x_1^1 = 0.6$.
Even though the realized reward of arm $a_2$ in that round is higher ($0.92$), the algorithm is not aware of that value.
At the end of the first round, $R^1_1 = R^1_{(2)} = R^1_2 = R^1_{(1)} = 0.6$. The reward discrepancy is thus $\adif{1}{2}^1 = \adif{2}{1}^1= \sdif{2}{1}^1 = 0.6 - 0.6 =0$.

In the second round, agent 1 arrives first, followed by agent 2.
Firstly, the algorithm pulls arm $a_1$ and agent 1 receives a reward of $r_{1}^2 = r_{(1)}^2 = x_1^2 = 0.48$.
Because the reward is lower than $\nicefrac{1}{2}$, in the second session the algorithm pulls the other arm ($a^2_{(2)} = a_2$), granting agent 2 a reward of $r_{2}^2 = r_{(2)}^2 = x_2^2 = 0.1$.
At the end of the second round, $R^2_1 = R^2_{(1)} = 0.6 + 0.48 = 1.08$ and $R^2_2 = R^2_{(2)} = 0.6 + 0.1 = 0.7$. Furthermore, $\sdif{2}{1}^2 = \adif{2}{1}^2 = r^2_{2} - r^2_{1} = 0.1 - 0.48 = -0.38$.

In the third and final round, agent 2 arrives first again, and receives a reward  of $0.15$ from $a_1$. When agent 1 arrives in the second session, the algorithm pulls arm $a_2$, and she receives a reward of $0.8$. As for the reward discrepancy, $\sdif{2}{1}^3 = \adif{2}{1}^3 = r^3_{2} - r^3_{1} = 0.15 - 0.8 = -0.75$. 

Finally, agent 1 has a cumulative reward of $R^3_1 = R^3_{(2)} = 0.6 + 0.48 + 0.8 = 1.88$, whereas agent~2 has a cumulative reward of $R^3_2 = R^3_{(1)} = 0.6 + 0.1 + 0.15 = 0.85$. In terms of envy, $\env^1_{1,2}= \adif{1}{2}^1 =0$, $\env^2_{1,2}=\adif{1}{2}^1+\adif{1}{2}^2= 0.38$, and $\env^3_{1,2} = -\env^3_{2,1} = R^3_1-R^3_2 = 1.88-0.85 = 1.03$, and consequently the envy in round 3 is $\env^3 = 1.03$.
\end{example}

\subsection{Information Exploitation}
\label{subsec:information}

In this subsection, we explain how algorithms can exploit intra-round information.
Since rewards are consistent in the sessions of each round, acquiring information in each session can be used to increase the reward of the following sessions.
In other words, the earlier sessions can be used for exploration, and we generally expect agents arriving in later sessions to receive higher rewards.
Taken to the extreme, an agent that arrives after all arms have been pulled could potentially obtain the highest reward of that round, depending on how the algorithm operates.

To further demonstrate the advantage of late arrival, we reconsider Example~\ref{example 1} and Algorithm~\ref{alguni}. 
The expected reward for the agent in the first session of round $t$ is $\E{\rt{(1)}}=\mu_1=\frac{1}{2}$, yet the expected reward of the agent in the second session is
\begin{align*}
\E{\rt{(2)}}=\E{\rt{(2)}\mid X^t_1 \geq \frac{1}{2} }\prb{X^t_1 \geq \frac{1}{2}} + \E{\rt{(2)}\mid X^t_1 < \frac{1}{2} }\prb{X^t_1 < \frac{1}{2}};
\end{align*}
thus, $\E{\rt{(2)}} =\E{X^t_1\mid X^t_1 \geq \frac{1}{2} }\cdot \frac{1}{2} + \mu_2\cdot\frac{1}{2} = \frac{5}{8}$.
Consequently, the expected welfare per round is $\E{\rt{(1)}+\rt{(2)}}=1+\frac{1}{8}$, and the benefit of arriving in the second session of any round $t$ is $\E{\rt{(2)} - \rt{(1)}} = \frac{1}{8}$. This gap creates envy over time, which we aim to measure and understand.
\subsection{Socially Optimal Algorithms}
\label{sec: sw}
Since our model is novel, particularly in its focus on the reward consistency element, studying social welfare maximizing algorithms represents an important extension of our work. While the primary focus of this paper is to analyze envy under minimal assumptions about algorithmic operations, we also make progress in the direction of social welfare optimization. See more details in Section~\ref{sec:discussion}.



%% file: sections/uniform.tex
\section{Uniform Arrival}
\label{sec:uniform}
In this section, we assume agents arrive uniformly; that is, the uniform arrival function, $\uniord$, picks in every round $t$ a permutation $\ordv_t$ uniformly at random from the set of all distributions $\textnormal{Perm}([N])$. We start with an insight into reward discrepancies for uniform arrival. Then, in Subsections~\ref{subsec:uni upper} and~\ref{subsec:uni lower}, we provide upper and lower bounds on the expected envy, respectively. 

Recall that $\dif^t_{i,j}$ is the signed reward discrepancy between agents~$i$ and~$j$ at round~$t$. This quantity depends on both the algorithm, the reward distribution, and the arrival function. 
For example, if an algorithm always pulls $a_1$, then $\dif^t_{i,j} = 0$ almost surely for every $i,j,t$, since all agents receive the same reward. By contrast, for more general algorithms, $\dif^t_{i,j}$ can vary almost arbitrarily, reflecting different approaches to exploration and exploitation over time. Under uniform arrival, each round's agent order is chosen uniformly at random from all permutations of $[N]$, ensuring identical treatment of all agents in expectation.  Consequently, the reward discrepancies exhibit a symmetry:
\begin{remark}\label{remark: symmetric dif}
Under $\uniord$ and any algorithm, the random variables $\dif^t_{i,j}$ in round $t$ are identically distributed for all pairs $(i,j)\in [N]^2$.
\end{remark}

Due to Remark~\ref{remark: symmetric dif}, we simplify the notation in this section and use $\dif^t$ to denote the reward discrepancy distribution between any two agents. Clearly, $\E{\dif^t}=0$. Note that the random variables $\adift{i_1}{j_1}$ and $\adift{i_2}{j_2}$ are still correlated for different pairs of agents $(i_1,j_1)$ and $(i_2,j_2)$. For example, $\adift{i}{j}= -\adift{j}{i}$ for all $i,j$. 
Moreover, $\dif^t$ might depend on the rewards obtained in previous rounds, reflecting temporal correlations that are typical in multi-armed bandit settings (e.g., stationary rewards).  As a result, a high reward discrepancy in earlier rounds might suggest a more thorough exploration, leading to a low reward discrepancy in later rounds. Consequently, we cannot assume that $(\dif^t)_{t=1}^T$ are independent.

\subsection{Upper Bound}\label{subsec:uni upper}
Our first result is an upper bound on the expected envy.
\begin{theorem}\label{thm: uni upper-bound}
When executing any algorithm, it holds that
\[\E{\max_{1\leq t \leq T} \env^t (\uniord)} \leq 2\sqrt{\ln{(N)} \sum^{T}_{t=1}{\var{\dift}} }.\]
\end{theorem}
\begin{proof}[Proof of Theorem~\ref{thm: uni upper-bound}]
The proof leverages properties of martingales and subgaussian random variables. Since the rewards are bounded, the resulting envy is also bounded and therefore subgaussian with some parameter. However, because envy is a sum of \emph{possibly dependent} random variables, additional arguments are needed to obtain a sharper subgaussian parameter. The following proposition is the primary technical ingredient in achieving this refinement.
\begin{proposition}\label{prop:envy is good SG}
For any two arbitrary agents $i,j\in [N]$, $\max_{1\leq t \leq T} \env^t_{i,j}$ is $\sqrt{\sum^T_{t=1}{\var{\dift}}}$-subgaussian.
\end{proposition}
We prove this proposition in \ifapp{Section~\ref{appendix:uni}}{the appendix}. Equipped with Proposition~\ref{prop:envy is good SG}, we can use the following well-known property of the maximum of subgaussian random variables.
\begin{claim}\label{claim: sg max}
    Let $Y_1, \ldots, Y_n$ be (possibly dependent) $\sigma$-subgaussian random variables. It holds that
    \[
    \E{\max_{i \in [n]}{\{Y_i\}} } \leq \sqrt{2\sigma^2 \ln{(n)} }.
    \]
\end{claim}
Claim~\ref{claim: sg max} is folklore, but we include its proof in \ifapp{Section~\ref{appendix:uni}}{the appendix} for completeness. Using this claim along with Proposition~\ref{prop:envy is good SG}, we obtain 
\[
\E{\max_{1\leq t \leq T} \env^t} =
\E{\max_{i,j \in [N]}{\left\{ {\max_{1\leq t \leq T} \env^t_{i,j}} \right\}}} \leq
\sqrt{2 \ln{(N^2)} \sum^{T}_{t=1}{\var{\dift}} } =
2\sqrt{\ln{(N)} \sum^{T}_{t=1}{\var{\dift}} }.
\]
This concludes the proof of Theorem~\ref{thm: uni upper-bound}.
\end{proof}
\subsubsection{Refining the Variance of Reward Discrepancy}
Although Theorem~\ref{thm: uni upper-bound} contains a factor of $\ln N$, it does not explicitly capture the influence of $N$ or $K$, as these parameters are embedded within $\var{\dift}$. To clarify this further, we refine our analysis by focusing on the subclass of algorithms known as \emph{explore-first} algorithms. 

We say that an algorithm satisfies the explore-first property if, for every round $t$, once it repeats an arm selection (i.e., chooses $a_i$ in two sessions), it commits to $a_i$ for all subsequent sessions in that round. More formally, 
\begin{definition}[Explore-first]\label{def:explore-first}
An algorithm is \emph{explore-first} if for every $t$ there exists a session $q(t)$ such that:
\begin{itemize}
    \item Before session $q(t)$, the algorithm pulls only unobserved arms.
    \item From session $q(t)$ onward, the algorithm exclusively pulls a single already-observed arm.
\end{itemize} 
\end{definition}
Observe that an explore-first algorithm commits to a single arm in each round, but this arm may vary across rounds and depend on the observed rewards. The explore-first property is natural and applies, e.g., to algorithms that use a threshold to start exploitation (like Algorithm~\ref{alguni}). Relying on Definition~\ref{def:explore-first}, we can express $\var{\dift}$ in terms of $N$ and $K$.
\begin{proposition}\label{thm: threshold var}
    When executing any explore-first algorithm with $K$ arms and $N$ agents, it holds that $\var{\dift} \leq \min\left\{1, \frac{(2N-K)(K-1)}{N(N-1)} \right\}$.  
\end{proposition}
The intuition behind this proposition is as follows. Since at most $K$ agents are exploring, the remaining $N-K$ agents are exploiting and receive identical rewards, so no discrepancy occurs between any two exploiting agents. As $N$ grows relative to $K$, most agent pairs consist of two exploiting agents, leaving only a small fraction of pairs---those involving at least one exploring agent---that can contribute to the variance. This decrease in the number of discrepancy-generating pairs results in a tighter overall bound on $\var{\dift}$.

Combining Proposition~\ref{thm: threshold var} with Theorem~\ref{thm: uni upper-bound}, we get the following corollary.
\begin{corollary} \label{thm: sqrt TK N}
    When executing any explore-first algorithm with $K$ arms and $N$ agents, the expected envy is $\E{\max_{1\leq t \leq T} \env^t  (\uniord)}=O\left( \sqrt{\frac{TK\ln(N)}{N}} \right)$.
\end{corollary}

\subsection{Lower Bound}\label{subsec:uni lower}
\label{sec: lb}

Next, we move to develop a lower bound on the expected envy. As is apparent, not all algorithms generate envy. To demonstrate, consider the following complementary cases:
\begin{enumerate}
    \item An instance in which all arms have the same \emph{deterministic} reward. In this case, any algorithm produces zero envy among all agents as all receive the same reward.
    \item An algorithm that pulls the same arm in all sessions. Even if rewards are stochastic, the algorithm does not generate any envy due to reward consistency.
\end{enumerate}
These examples show that envy will not accumulate if the instance or the algorithm are degenerate, motivating to focus on \emph{executions}: A pair of algorithm and reward distributions. In what follows, we characterize a class of non-degenerate executions for which envy always accumulates. we call such compositions \emph{sufficiently random executions}.
Formally,
\begin{definition}[Sufficiently Random]\label{def: sufficiently random}
    An execution is called \emph{sufficiently random} if it holds that
\begin{equation}\label{eq:def suff}
\sum_{t=1}^T{\var{\dift}} \geq \sqrt{T}.   
\end{equation}
\end{definition}
In other words, an execution is sufficiently random if the average variance of the reward discrepancy between two agents in a round is greater or equal to $\frac{1}{\sqrt{T}}$. To further illustrate, we provide the following example. 
\begin{example}\label{example: uni suff}
Recall the execution in Example~\ref{example 1}, with $K=2$ arms with  $\uni{0,1}$ rewards, $N=2$ agents, and Algorithm~\ref{alguni}. As we formally show in \ifnum\Includeappendix=0{the appendix}\else{Claim~\ref{claim: uni suff}}\fi, it holds that $\var{\dift} = \frac{1}{12}$; thus, as long as $T \geq 144$, we have $\sum_{t=1}^T{\var{\dift}} = \frac{T}{12}  {\geq } \sqrt{T}
$; hence, Inequality~\eqref{eq:def suff} holds and this execution is sufficiently random. 
\end{example}
Intuitively speaking, any ``reasonable'' algorithm and reward distributions form a sufficiently random execution. As long as the reward distributions are constant w.r.t. the horizon $T$ and the algorithm conducts enough exploration, the execution will be sufficiently random. To demonstrate cases that are insufficiently random, consider the following example.
\begin{example}\label{example: ber suff}
Assume $K$ i.i.d. arms with rewards distributed $Bernoulli(p)$ for some $p \in (0,1)$, and $N$ agents for $N\geq 2$. We focus on the socially optimal algorithm: It picks fresh arms until a reward of $1$ is obtained and exploits it afterward for all the remaining agents. We show in \ifnum\Includeappendix=0{the appendix}\else{Claim~\ref{claim:example ber suff}}\fi~ that $\var{\Delta^t} \geq \frac{2p(1 - p)}{N}$ and $\var{\Delta^t} \leq  2pK$.

Consequently, as long as  $\frac{2p(1 - p)}{N} \geq \frac{1}{\sqrt T}$ holds, i.e., for $T \geq \frac{N^2}{4p^2(1-p)^2}$, the execution is sufficiently random. However, if $2pK <\frac{1}{\sqrt T}$, which is the case if $p<\frac{1}{2K\sqrt T}$, the execution is insufficiently random.
\end{example}
Next, we derive a lower bound on the envy for sufficiently random executions.
\begin{theorem}\label{thm: uni lower-bound}
Any algorithm as part of sufficiently random execution generates an envy of
\[\E{\max_{1\leq t \leq T} \env^t  (\uniord)} \geq c\sqrt{ \sum^{T}_{t=1}{\var{\dif^t}}},\]
where $c>0$ is a global constant that does not depend on the instance. 
\end{theorem}

\begin{proof}[Proof of Theorem~\ref{thm: uni lower-bound}]
    Noticeably, we can bound $\env^t$ using the envy between two specific agents.
    That is,
    \begin{align}\label{thm: mp uni lower-bound 1}
        \E{\max_{1\leq t \leq T} \env^t} =
        \E{\max_{i,j \in [N]}{\left\{ \max_{1\leq t \leq T} \env^t_{i,j} \right\}}} \geq
        \E{\max_{1\leq t \leq T} \abs{\env^t_{1,N}}} =
        \E{\max_{1\leq t \leq T}\abs{\sum^T_{t=1}{\dift}} },
    \end{align}
    where the last inequality is due to  Remark~\ref{remark: symmetric dif}. To express $\E{\max_{1\leq t \leq T}\abs{\sum^T_{t=1}{\dift}} }$ as a function of $\sum_t\var{\dift}$, we can lower bound its quadratic variation as we have done in the proof of the upper bound. Next, 
    we use the Burkholder-Davis-Gundy inequality~\cite{davis1970intergrability}. We present a simplified version of the theorem, as we are only interested in the special case of a discrete-time martingale with $L_1$ norm. 
    \begin{theorem}\label{thm:BDG}[Burkholder-Davis-Gundy inequality]
    Let $\{M^t\}_{t \geq 0}$ be a discrete-time martingale with $M_0 = 0$. There exist positive constants $A_1$ and $B_1$ such that 
    \[ 
    A_1 \, \E{\sqrt{\sum_{t=1}^T (M^t - M^{t-1})^2}}\leq \E{\max_{0 \leq t \leq T} \abs{M^t}}\leq B_1 \, \E{\sqrt{\sum_{t=1}^T (M^t - M^{t-1})^2}}.
    \]    
    \end{theorem}
    As we formally show in \ifnum\Includeappendix=0{the appendix}\else{Observation~\ref{envy is martingale}}\fi, the sequence $(\env^t_{1,N})_t$ forms a martingale; hence, an immediate corollary from Theorem~\ref{thm:BDG} and Inequality~\eqref{thm: mp uni lower-bound 1} is that\footnote{The reader might be tempted to use this theorem for the upper bound as well, obtaining a straightforward bound without additional intricate arguments. However, the maximal envy $\env^T$ is not a martingale, so we can only apply it to the envy between two agents. Replacing the $\max$ operator with a summation results in an upper bound of $O\left(N^2 \sqrt{\sum_{l=1}^{T} \var{\dift}}\right)$, which includes an  additional multiplicative factor of $\nicefrac{N^2}{\ln^2 N}$ to the bound of Theorem~\ref{thm: uni upper-bound}.}
    \begin{align}\label{eq: gdgggsds}
    \E{\max_{1\leq t \leq T} \env^t} \geq A_1 \E{ \sqrt{\sum^T_{t=1}{ (\dift)^2}} }.
    \end{align}
    
    Next, we use the following auxiliary proposition, which provides a reverse Jensen's-like inequality of the square root function.
    \begin{proposition}\label{thm: board}
        Let $Y$ be a non-negative random variable with a finite second moment. It holds
        \begin{align*}
            \E{\sqrt{Y}} \geq \sqrt{\E{Y}}\left(1- \frac{\var{Y}}{2\E{Y}^2}\right).
            \end{align*}
    \end{proposition}
    We prove Proposition~\ref{thm: board} in \ifnum\Includeappendix=0{the appendix}\else{Appendix~\ref{appendix:uni}}\fi. Applying Proposition~\ref{thm: board} with $Y = \sum^T_{t=1}{ \dift^2}$ to Inequality~\eqref{eq: gdgggsds},
    \begin{align}\label{thm: mp uni lower-bound 2}
        \E{\max_{1\leq t \leq T} \env^t} \geq
        A_1 \E{ \sqrt{\sum^T_{t=1}{ \dift^2}} }\geq
        A_1 \sqrt{\E{ \sum^T_{t=1}{ \dift^2}} }
        \left(
        1 - \frac{\var{ \sum^T_{t=1}{ \dift^2} } }{2 \E{ \sum^T_{t=1}{ \dift^2} }^2 }
        \right) .
    \end{align}
    Next, due to our assumption that the execution is sufficiently random, we get  
    \begin{align}\label{thm: mp uni lower-bound 3}
        2 \left(\E{ \sum^T_{t=1}{ \dift^2} }\right)^2 =2 \left( \sum^T_{t=1}{ \E{ \dift^2} - \E{ \dift}^2}\right)^2 = 2 \left( \sum^T_{t=1}{\var{  \dif^t } }\right)^2 \geq 2T.
    \end{align}
    Furthermore, notice that $\var{ \sum^T_{t=1}{ \dift^2} } \leq T$ as the discrepancies $(\dift)_t$ are supported in the $[-1, 1]$ segment. Combining this fact with Inequality~\eqref{thm: mp uni lower-bound 3}, 
    \begin{align}\label{eq:m,bnhjikw}
    \frac{\var{ \sum^T_{t=1}{ \dift^2} } }{2 \E{ \sum^T_{t=1}{ \dift^2} }^2 }\leq \frac{T}{2T}=\frac{1}{2}.
    \end{align}
    Plugging Inequality~\eqref{eq:m,bnhjikw} to Inequality~\eqref{thm: mp uni lower-bound 2}, we conclude that 
    
    \[
    \E{\max_{1\leq t \leq T} \env^t} \geq
    A_1 \sqrt{\E{ \sum^T_{t=1}{ \dift^2}} }
    \left(1 -\frac{1}{2}\right)
    =
    \frac{A_1}{2}\sqrt{\sum^T_{t=1}{ \var{ \dift }}},
    \]
    where the last inequality is again due to $\E{\dift}=0$. This concludes the proof of Theorem~\ref{thm: uni lower-bound}.
\end{proof}
Theorems~\ref{thm: uni lower-bound} and~\ref{thm: uni upper-bound} imply the following corollary.
\begin{corollary}\label{cor: uni-envy}
Any algorithm as part of sufficiently random execution generates an envy of $\E{\max_{1\leq t \leq T} \env^t  (\uniord)} = \Theta\left(\sqrt{ \sum^{T}_{t=1}{\var{\dif^t}} }\right)$.  
\end{corollary}
Before we complete the section, we remark that high envy may still arise under insufficiently random executions. Indeed, although Definition~\ref{def: sufficiently random} gives a sufficient condition, it is not a necessary one. From a technical perspective, Theorem~\ref{thm: uni lower-bound} applies whenever the left-hand side of Inequality~\eqref{eq:m,bnhjikw} is constant. In certain cases, this allows us to relax the requirement for sufficiently random executions. The following proposition shows that this relaxation holds in the instance described in Example~\ref{example: ber suff}. Recall that the execution in this example is insufficiently random for $p < \frac{1}{2K\sqrt T}$. However, the next proposition shows that the tight bound in Corollary~\ref{cor: uni-envy} can still hold if $p$ is significantly smaller.
\begin{proposition}\label{prop:insufficient}
    For in the execution in Example~\ref{example: ber suff}, as long as    
    $p\in \left[\frac{N}{cT}, 1-\frac{N}{cT}\right]$ for a constant $c \geq 2$ and $T \geq N$, the envy satisfies $\E{\max_{1\leq t \leq T} \env^t  (\uniord)} = \Theta\left(\sqrt{ \sum^{T}_{t=1}{\var{\dif^t}} }\right)$.
\end{proposition}

%% file: sections/nudge.tex
\section{Nudged and Adversarial Arrival}
\label{sec:nudge}
In this section, we address nudged arrival: We assume an exogenous arrival mechanism can influence the order in which agents arrive. Practically, this captures scenarios where the system sends push notifications or otherwise encourages some agents to arrive earlier or later. Our goal is to analyze the envy of arbitrary algorithms without changing the way they select arms. We stress that our analysis still assumes anonymous algorithms: The decision-making process is not affected by agent identities.

Subsection~\ref{subsec:nudged prot} introduces the nudged arrival protocol, $\sugord$, and the accompanying assumptions. Later,  Subsection~\ref{Envy Analysis} presents our main result of the section: An upper bound of $\env^T({\sugord})$ for a broad class of algorithms. Interestingly, we show that the bound depends on the instance parameters but not on the horizon $T$. Finally, we complete this section by adopting a complementary approach, where an adversary can pick the worst arrival order in terms of envy, and show that an $\Omega(T)$ regret is inevitable. For clarity, we remind the reader that for $q,i\in [N]$, $\ordv_t (q)=i$ implies that agent $i$ has arrive in the $q$'th session in round $t$ (similarly for $\ordv_t^{-1}(i)=q$). 
\subsection{Nudged Arrival Protocol}\label{subsec:nudged prot}
\begin{algorithm}[t]
\caption{Nudged Arrival}
\label{alg: sugg arr}
\SetAlgoLined
\LinesNumbered
\KwIn{horizon $T$, nudge parameter $\delta$} \label{line:input}
\For{round $t = 1$ to $T$}{ \label{line:for_loop} 
    let $\sigv_t$ be an $[N]\rightarrow [N]$ mapping such that 
    \[
    R^{t-1}_{\sigma_t(N)} \leq R^{t-1}_{\sigma_t(N-1)} \leq \dots \leq R^{t-1}_{\sigma_t(2)} \leq R^{t-1}_{\sigma_t(1)}.    
    \]
    \label{line:mapping}\\
    sample an arrival order $\ordv_t \sim \sugord(\sigma_t,\delta)$ \label{line:request}
}
\end{algorithm}
We describe the arrival protocol in Algorithm~\ref{alg: sugg arr}. It receives the horizon $T$ and a nudge parameter~$\delta$ as input, and interacts with any recommender algorithm through the horizon. In each round in the for loop of Line~\ref{line:for_loop}, we pick an \textit{ideal permutation} $\sigma_t$ that orders the agents according to their cumulative rewards.  The mapping $\sigma_t$ prioritizes agents according to their cumulative rewards in previous rounds, from the most rewarded one to the least rewarded one. In our notation, $\sigma_t : \{1,\dots,n\} \to [N]$, so that $\sigma_t(i)$ is the name of the agent in position $i$. Hence, $\sigma_t(1)$ is the agent with the highest cumulative reward so far, and $\sigma_t(N)$ is the agent with the lowest. While we assume the agents can be nudged toward this ordering $\sigma_t$, we do not claim that it is implemented or enforced in its exact form. Instead, in Line~\ref{line:request}, we sample an arrival order $\ordv$ from   $\sugord(\sigma_t,\delta)$, which is a distribution over permutations of $N$. Particularly, we assume that $\sugord$ satisfies the \emph{nudged arrival} property.

\begin{property}[Nudged Arrival]\label{prop:nudge}
Given a scalar \( \delta \in (0,1) \) and a mapping \( \sigv : [N] \rightarrow [N] \),  for every two agents \( i,j \) with \( \sigma^{-1}(i) < \sigma^{-1}(j) \) the distribution $\sugord(\sigv,\delta)$ satisfies
\begin{equation}\label{def: nudged} 
\Pr_{\ordv \sim \sugord(\sigma,\delta)}\left(\ordv^{-1}(i) < \ordv^{-1}(j)\right) \geq \frac{1+\delta}{2}. 
\end{equation}  
\end{property}
In other words, Algorithm~\ref{alg: sugg arr} introduces a structured form of randomness in the sequence of agent arrivals. Each round allows for the selection of a permutation $\sigma_t$, representing a preferred (but not guaranteed) ordering of agents. For any two agents $i$ and $j$, if the permutation prioritizes agent $i$'s arrival over agent $j$'s arrival, i.e., $\sigma_t^{-1}(i) < \sigma_t^{-1}(j)$, then agent $i$ is more likely to arrive before agent $j$. Particularly, the probability of agent $i$ preceding agent $j$ in the realized arrival order $\ordv_t$ is at least $\frac{1}{2} + \frac{\delta}{2}$, ensuring a consistent bias $\delta$ toward the preferred ordering. Property~\ref{prop:nudge} is inspired by several well-known models of stochastic ranking, like Mallows model~\cite{mallows1957non}, Plackett-Luce~\cite{marden1996analyzing} and also noisy comparison models~\cite{braverman2007noisy}. In \ifnum\Includeappendix=0{the appendix}\else{Section~\ref{appendix:nudge-models}}\fi, we describe how to derive $\delta$ from each of these models. By specifying Property~\ref{prop:nudge} rather than the full underlying ordering model, we preserve flexibility in how global orderings are derived.  
Next, we illustrate nudged arrival and the role it plays in envy dynamics.
\begin{example}[Nudged Arrival and Envy Dynamics]  \label{example: envy with sugg}
We reconsider the setting of Example~\ref{example 1}, with $K=2$ arms with  $\uni{0,1}$ rewards, $N=2$ agents, Algorithm~\ref{alguni}, but with the nudged arrival ${\sugord}$. To ease readability, we keep using $\env^T$ to denote $\env^T({\sugord})$, omitting the dependence on  ${\sugord}$. Suppose that after $t-1$ rounds, for some arbitrary $t \in [T]$, agent rewards satisfy $\env^{t-1} =\env^{t-1}_{1,2} = R^{t-1}_1 - R^{t-1}_2 > 1$, indicating that agent 2 envies agent 1.

In this scenario, the expected envy after round $t$ is given by:  
{  
\thinmuskip=2mu  
\medmuskip=3mu plus 2mu minus 3mu  
\thickmuskip=4mu plus 5mu minus 2mu  
\begin{align}\label{eq:env_update}  
\mathbb{E}[\env^t \mid \env^{t-1}_{1,2} > 1]  
&= \mathbb{E}\left[\left| \env^{t-1}_{1,2} + \Delta^t_{1,2} \right| \mid \env^{t-1}_{1,2} > 1 \right] = \env^{t-1}_{1,2} + \mathbb{E}\left[\Delta^t_{1,2} \mid\env^{t-1}_{1,2} > 1 \right]. 
\end{align}  
}

The term $\mathbb{E}[\Delta^t_{1,2} \mid \env^{t-1}_{1,2} > 1 ]$ can be expressed as:  
{  
\thinmuskip=2mu  
\medmuskip=3mu plus 2mu minus 3mu  
\thickmuskip=4mu plus 5mu minus 2mu  
\begin{align}\label{eq:delta_expression}  
\E{\Delta^t_{1,2} \mid \env^{t-1}_{1,2} > 1} &= \E{r_{(2)}^t - r_{(1)}^t \mid \env^{t-1}_{1,2} > 1, \ordv_t = (2,1)}\prb{\ordv_t = (2,1) \mid \env^{t-1}_{1,2} > 1} \nonumber \\  
& \qquad + \E{r_{(1)}^t - r_{(2)}^t \mid \env^{t-1}_{1,2} > 1, \ordv_t  = (1,2)}\prb{\ordv_t = (1,2) \mid \env^{t-1}_{1,2} > 1} \nonumber\\
&= \frac{1}{8} \left[\prb{\ordv_t = (2,1) \mid \env^{t-1}_{1,2} > 1} - \prb{\ordv_t = (1,2) \mid \env^{t-1}_{1,2} > 1} \right],  
\end{align}  
}%
where we have used the fact that $\E{\rt{(2)} - \rt{(1)}} = \frac{1}{8}$ as Subsection~\ref{subsec:information} suggests and the fact that the rewards $r_{(i)}^t$ for $i\in \{1,2 \}$ are independent of the arriving agent's identity.

Under nudged arrival, the ideal permutation is $\sigma_t = (1,2)$, prioritizing agent $1$'s arrival over agent $2$'s arrival;therefore, the permutation $(1,2)$ is more likely than $(2,1)$, resulting in
\[
\prb{\ordv_t = (2,1) \mid \env^{t-1}_{1,2} > 1} - \prb{\ordv_t = (1,2) \mid \env^{t-1}_{1,2} > 1} \leq -\delta.
\]  
Thus, rewriting Equation~\eqref{eq:delta_expression}, we have $\E{\Delta^t_{1,2} \mid \env^{t-1}_{1,2} > 1}  \leq -\frac{\delta}{8}$. To conclude this example, we plug this result into Equation~\eqref{eq:env_update} and obtain
\[
\mathbb{E}[\env^t \mid \env^{t-1}_{1,2} > 1] \leq \env^{t-1}_{1,2}-\frac{\delta}{8},
\]
suggesting that the cumulative envy $\env^t$ is likely to decrease in round $t$ by a non-negligible value.
\end{example}
To be able to analyze envy dynamics and show that it cannot grow too much, we need to have some regularity assumptions on the way algorithms we analyze operate. Indeed, the envy reduction in Example~\ref{example: envy with sugg} relies heavily on the fact that $\E{\rt{(2)} - \rt{(1)}} = \frac{1}{8} > 0$. This inequality ensures that, in expectation, arriving second leads to a higher reward. While we aim to analyze any arbitrary algorithm, we need to ensure that the ordering $\sigma_t$ in Line~\ref{line:mapping} reduces envy in expectation.
The following natural assumption generalizes the behavior of Algorithm~\ref{alguni} in Example~\ref{example: envy with sugg}.\begin{assumption}\label{assumption: nudge alg ref}
In every round $t \in [T]$, the algorithm picks arms so that
\[
    \E{r_{(1)}^t} \leq \E{r_{(2)}^t}\leq \cdots \leq \E{r_{(N)}^t}.
\]
\end{assumption}
To satisfy this assumption,\footnote{In fact, all of our results holds for the much broader case where there exists a permutation $\sigma:N \rightarrow N$ and the algorithm picks arms so that 
$\E{r_{(\sigma(1))}^t} \leq \E{r_{(\sigma(2))}^t}\leq \cdots \leq \E{r_{(\sigma(N))}^t}$. In such a case, we would pick $\sigma_t$ in Line~\ref{line:mapping} of Algorithm~\ref{alg: sugg arr} so that $\left(R^{t-1}_{\sigma_t(\sigma(i))}\right)_i$ is an increasing series.} the algorithm at hand must have some information about the expected rewards; Bayesian information is a sufficient condition, although not necessary.

Furthermore, without loss of generality, we shall assume that $\prb{\Delta^t_{(N),(1)} \neq 0}>0$ in every round $t$. This is indeed without loss of generality, as Assumption~\ref{assumption: nudge alg ref} already guarantees that $\E{\Delta^t_{(N),(1)}}=\E{r^t_{(N)}-r^t_{(1)}} \geq 0 $, and any round in which $\prb{\Delta^t_{(N),(1)} \neq 0}=0$ does not affect envy and can be disregarded. Additionally, to simplify our analysis, we introduce the new notation $\tilde{\dif}$, denoting
\begin{equation}\label{eq def tdif}
\tdif = \min_{1\leq i <j \leq N, t\in [T]:  \prb{\Delta^t_{(j),(i)} \neq 0}>0 }\left\{ \E{\Delta^t_{(j),(i)} \mid \Delta^t_{(j),(i)} \neq 0} \right\}    
\end{equation}
The quantity $\tdif >0$ 
represents a lower bound on our ability to decrease envy from round to round. Recall that Assumption~\ref{assumption: nudge alg ref} implies that $\E{\Delta^t_{(j),(i)}}=\E{r^t_{(j)}-r^t_{(i)}} \geq 0 $ for $j>i$. However, $\Delta^t_{(j),(i)}$ can take 0 sometimes,\footnote{For instance, if $N > K+1$, any explore-first algorithm would have $\Delta^t_{(N-1),(N)}=0$ as the algorithm will pick the same arm for both sessions.} which means there is no scope for nudged arrival to further reduce envy. As long as Assumption~\ref{assumption: nudge alg ref} holds and $ \prb{\Delta^t_{(j),(i)} \neq 0}>0$, we know that
\[
\E{\Delta^t_{(j),(i)} \mid \Delta^t_{(j),(i)} \neq 0}=\frac{ \E{\Delta^t_{(j),(i)}}}{\prb{\Delta^t_{(j),(i)} \neq 0}};
\] thus, we expect $\tdif$ to be significant. To illustrate, recall that in Example~\ref{example: envy with sugg} it holds that $\E{\Delta^t_{(2),(1)}}=\frac{1}{8}$, whereas $\E{\Delta^t_{(2),(1)} \mid \Delta^t_{(2),(1)} \neq 0}=\frac{1}{4}$.
\subsection{Envy Analysis}\label{Envy Analysis}
We are ready to present the main result of the section: Upper bounding the envy under nudged arrival.
\begin{theorem}\label{thm: sugg-envy}
    When executing any algorithm that satisfies Assumption~\ref{assumption: nudge alg ref} with nudged arrival, the expected envy is
    \[\E{\env^T({\sugord})}\leq (N-1)\left(2+\frac{128}{15\delta \tdif}\right) .\]
\end{theorem}
Notice that this upper bound does not depend on the horizon $T$. Intuitively, under nudged arrival, envy behaves like a random walk with a drift toward zero. Although each round may introduce a discrepancy (akin to a random fluctuation), the nudging mechanism consistently pushes the cumulative difference back toward zero. Furthermore, the bound is inversely proportional to $\delta$ and $\tdif$: As $\delta$ decreases, the nudging effect weakens and nudged arrival increasingly resembles uniform arrival. We suspect the terms in the bound are not tight; we discuss it in Section~\ref{sec:discussion}.
\begin{proof}[Proof of Theorem~\ref{thm: sugg-envy}]
Fix any arbitrary algorithm satisfying  Assumption~\ref{assumption: nudge alg ref} and any arbitrary~$t \in [T]$. The proof is outlined as follows:
\begin{enumerate}
    \item Step 1 introduces envy gaps as stochastic processes and the concept of envy excursions.
    \item Demonstrating that envy gaps are nontrivial to analyze, Step 2 presents a more friendly stochastic process that we prove to upper bound the envy gap almost surely.
    \item Step 3 leverages Property~\ref{prop:nudge} and concentration inequalities to upper bound large deviations of the friendly stochastic process.
    \item Lastly, Step 4 uses the tail formula and the concentration from Step 3 to bound to cumulative envy.
\end{enumerate}
\textbf{Step 1: Envy Gap and Excursion}
For every $t$, let $\sigma_t: [N] \rightarrow [N]$ be the ideal permutation from Line~\ref{line:mapping}, i.e.,  
\[
R^{t-1}_{\sigma_t(N)} \leq  R^{t-1}_{\sigma_t(N-1)} \leq \cdots \leq R^{t-1}_{\sigma_t(2)} \leq R^{t-1}_{\sigma_t(1)} .
\]
For every $i, 1\leq i \leq N-1$, we define the \emph{envy gap} $G_i^t = R^t_{\sigma_t(i)} - R^t_{\sigma_t(i+1)}$, representing the envy between the agent with the $i$-th highest reward and the agent with the $(i+1)$-th highest reward.
Consequently, we can define $\env^t$ using the envy gap sequence $(G^t_i)_i$ by
\begin{equation}\label{eq:jknbvfbjj}
\env^t=R^t_{\sigma_t(1)}-R^t_{\sigma_t(N)}=\sum_{i=1}^{N-1} R^t_{\sigma_t(i)}-R^t_{\sigma_t(i+1)} = \sum_{i=1}^{N-1} G_i^t.   
\end{equation}

Next, fix any arbitrary $i$ in the range. We continue by analyzing \emph{excursions} from low envy to high envy and showing they are relatively short, meaning that the expected envy $\E{G_i^t}$ is low. We define an excursion as a sequence of consecutive rounds during which the gap $G_i^\cdot$ exceeds $1$. 
Let $\underline{t} = \max{\left\{ \tau \mid 1 \leq \tau \leq t, G_i^{\tau} \leq 1 \right\}}$ denote last round $\tau$ before $t$ that $G_i^{\tau}$ was less than 1. Similarly,  let $\bar{t} = \min{\left\{ \tau \mid t\leq \tau \leq T, G_i^{\tau} \leq 1 \right\}}$ be the first round $\tau$ after $t$ where $G_i^{\tau}$ is less than 1. For the extreme case where $\bar{t}$ is undefined, we set $\bar{t} = T + 1$. Furthermore, let $D(t)$ denote $t$'s excursion, namely, the set of all the consecutive rounds $\tau$ that includes $t$ during which $G_i^\tau \geq 1$. Formally, $D(t) = \{\tau | \underline{t} < \tau < \bar{t}\}$. Notice that $D(t)$ is an empty set if and only if $G_i^t \leq 1$.

\textbf{Step 2: Auxiliary Stochastic Process}
The sequence $\{G_i^\tau\}_{\tau \in D(t)}$ is challenging to work with because the agents occupying the $i$-th and $(i+1)$-th highest reward position may change from round to round.  We refer to these changes as \emph{rank swaps}, which cause increments like  $G_i^{\tau+1} - G_i^\tau$ to lack a consistent structure. To address this complexity, we introduce the stochastic process $(M^\tau)_\tau$, which is easier to analyze.
\begin{align*}
    M^\tau =
    \begin{cases}
        2 & \text{if } \tau = \underline{t}+1 \\
        M^{\tau-1}+r^\tau_{(i)}-r^\tau_{(i+1)} & \text{else}
        \end{cases}.
\end{align*}
Unlike $G_i^{\tau+1}-G_i^{\tau}$, the increments $M^{\tau+1}-M^{\tau} = r^\tau_{(i)}-r^\tau_{(i+1)}$ are more straightforward and negative in expectation due to Assumption~\ref{assumption: nudge alg ref}. The next proposition demonstrates that $M^\tau_i$ can assist when analyzing $G^\tau_i$.
\begin{proposition}\label{prop G less than M}
For every $\tau \in D(t)$, it holds that $G^\tau_i \leq M^\tau_i$ almost surely.
\end{proposition}
The proof of Proposition~\ref{prop G less than M} appears in \ifapp{Section~\ref{appendix:nudge}}{the appendix}. The main argument enabling this statement is that rank swaps can only decrease the increments $G_i^{\tau+1}-G_i^{\tau}$, but do not affect the increments $M^{\tau+1}-M^{\tau}$.

\textbf{Step 3: Concentration}
The recursive definition of $M^\tau$ implies that for every $\tau \in D(t)$, $M^\tau = 2+ \sum_{n= \underline{t}+2 }^{\tau} r^n_{(i)}-r^n_{(i+1)}$. The next proposition bounds large deviations of $M^\tau$.
\begin{proposition}\label{prop: sugg-m concentration}
    For any $n \in \mathbb{N}$ and $\tau \in D(t)$, it holds that
    \begin{align*}
    \prb{M^\tau > n}\leq \exp\left\{-\frac{(n-2)(\delta \tdif)}{8}\right\}.
    \end{align*}
\end{proposition}
The proof of Proposition~\ref{prop: sugg-m concentration} appears in \ifapp{Section~\ref{appendix:nudge}}{the appendix}. It leverages the Azuma-Hoeffding inequality and several algebraic tricks to obtain the bound. As we expect, as $\delta$ and $\tdif$ increase, the right-hand side becomes more significant. Alternatively, if the term $\delta \tdif$ approach zero, this bound become irrelevant as $\exp\{0\}=1$.

\textbf{Step 4: Tail Sum}
To finalize the proof, we use the tail-sum formula. Since $G_i^t$ is non-negative,
\begin{align*}
\E{G_i^t} &=
    \int_{x=0}^{\infty} \prb{G_i^t > x}dx \leq
    \sum_{n=0}^{\infty}{\int_{x=n}^{n+1} \prb{G_i^t > n} \,dx}=
    \sum_{n=0}^{\infty}{\prb{G_i^t > n}}.
\end{align*}    
Next, by applying Propositions~\ref{prop G less than M} and~\ref{prop: sugg-m concentration}, we conclude that
\begin{align}\label{eq:tail formula}
    \sum_{n=0}^{\infty}{\prb{G_i^t > n}}& \stackrel{\textnormal{Prop. \ref{prop G less than M}}}{\leq} \sum_{n=0}^{\infty}{\prb{M_i^t > n}} \leq 2+ \sum_{n=2}^{\infty}{\prb{M_i^t > n}}   \stackrel{\textnormal{Prop. \ref{prop: sugg-m concentration}}}{\leq} 2+ \sum_{n=2}^{\infty}\exp{\left\{-\frac{(n-2)\delta \tdif}{8} \right\}}\nonumber \\
    &= 2+ \sum_{n=0}^{\infty}\left(e^{-\frac{\delta \tdif}{8}} \right)^n = 2+\frac{1}{1-\exp{\left(\frac{-\delta \tdif}{8} \right)}} \stackrel{e^{-x}\leq 1-x+\frac{x^2}{2}}{\leq}
    2+\frac{1}{\frac{\delta \tdif}{8}-\frac{(\delta \tdif)^2}{128}}
 \nonumber \\
    & \stackrel{\delta \tdif \leq 1}{\leq} 2+\frac{1}{\frac{\delta \tdif}{8}-\frac{\delta \tdif}{128}} = 2+\frac{128}{15\delta \tdif}.
\end{align}
Ultimately, recall that Inequality~\eqref{eq:tail formula} applies to $\E{G_i^t}$ for every $i$; therefore, Equation~\eqref{eq:jknbvfbjj} ensures that $\E{\env^t}=\E{\sum_{i=1}^{N-1} G_i^t} \leq (N-1)\left(2+\frac{128}{15\delta \tdif}\right) $. This concludes the proof of Theorem~\ref{thm: sugg-envy}.
\end{proof}

{\color{green}

}
\subsection{Adversarial Arrival}
\label{sec: advord}
We end this section by focusing on the adversarial arrival order $\advord$. Intuitively, an adversary seeking to maximize envy would reverse the ideal permutation, placing agents in descending order of their current cumulative rewards. That is, sets the order $\ordv_t$ in round $t$ such that
\[
R^{t-1}_{\ordv_t(N)} \geq R^{t-1}_{\ordv_t(N-1)} \geq \dots \geq R^{t-1}_{\ordv_t(2)} \geq R^{t-1}_{\ordv_t(1)}.
\]
Indeed, it is easy to see that:
\begin{proposition}\label{thm: adv-envy}
When executing any algorithm that satisfies Assumption~\ref{assumption: nudge alg ref} with $\advord$, the expected envy is
    \[\E{\env^T ({\advord})}\geq \tilde{\dif}T.\]
\end{proposition}
\begin{proof}[Proof of Proposition~\ref{thm: adv-envy}]
Assume that the adversary picks agent 1 to be the first and agent $N$ to be the last, i.e., $\ordv_t(1)=1$ and $\ordv_t(N)=N$ for all $t\in [T]$. In such a case,
\begin{align*}
\E{\env^T({\advord})} &=
\E{\max_{i,j\in [N]}{\left\{ \sum^T_{t=1}{\adift{i}{j}} \right\} }} \geq \E{\sum^T_{t=1}{\dif_{N,1}^t}} = \E{\sum^T_{t=1}{\dif^t_{(N),(1)}}}  \geq T\E{\min_{1 \leq t \leq T}{\left\{\sdift{N}{1} \right\} }} \\
&\geq \tdif T.    
\end{align*}
This concludes the proof of Proposition~\ref{thm: adv-envy}.
\end{proof}

%% file: sections/efc.tex
\section{Extension: Trading Envy and Welfare}
\label{sec:extensions}

In the previous section, we have shown that coordinating agents' arrival order alone can significantly reduce envy, without affecting the algorithm's core decision-making process. In this section, we take an initial step toward understanding the efficiency-fairness tradeoff, a well-established concept in the literature on fair allocation~\cite{varian1973equity,bertsimas2012efficiency}  and fair classification~\cite{menon2018cost,zafar2017fairness}. 
Specifically, we extend the definition of algorithms from Section~\ref{sec:model} to allow agent‐specific treatment. In other words, algorithms can now observe agent identities and maintain agents accumulated rewards in their memory. Formally, the relevant histories contain triples of the form (agent index, pulled arm, realized reward). We hope to leverage this additional capability to balance social welfare and envy.


We focus on the special case of our running example (Example~\ref{example 1}): $N=2$ agents, $K=2$ arms with rewards drawn from the uniform distribution, $X_1, X_2 \sim \uni{0,1}$, and the uniform arrival $\uniord$. Furthermore, we assume Bayesian information, i.e., the prior distributions are known. We stress that our results are preliminary, albeit non-trivial. In Subsection~\ref{subsec:ext welfare}, we analyze the socially optimal algorithm from a welfare and envy perspective. Later, in Subsection~\ref{subsec:ext efc}, we develop $\efc$, our welfare-envy balancing algorithm. 
\subsection{Optimal Welfare and Optimal Envy}\label{subsec:ext welfare}
We first analyze the maximal social welfare for this setting. As it turns out, Algorithm~\ref{alguni} is a special case of the optimal two-agent algorithm, as we prove in \ifapp{Section~\ref{appendix:sociallyopt}}{the appendix}. Along with our results from Section~\ref{sec:uniform}, we conclude that:
\begin{observation}\label{obs:opt for tradeoff}
When executing Algorithm~\ref{alguni} on the instance of Example~\ref{example 1} and $\uniord$, it achieves an expected social welfare of (1+$\frac{1}{8}) T$ and induces an expected envy of $\env^T(\uniord)=\tilde \Theta(\sqrt T)$. Furthermore, this is the optimal welfare.
\end{observation}
This observation indicates that Algorithm~\ref{alguni} occupies an extreme point on the Pareto frontier of the welfare-envy tradeoff: It achieves maximum welfare but also incurs high envy. Another point on the frontier is the algorithm we call $NE$ (\textbf{N}o \textbf{E}nvy), guaranteeing $\env^t=0$ in every round $t$ almost surely. $NE$ draws the same arm in both sessions of every round, as this is essential to maintain zero envy (since the rewards are uniformly distributed and stochastic by nature). Of course, $NE$ has an expected social welfare of $\sw= T$.

A compelling way to address the social welfare of any algorithm is by examining its ability to exploit the information obtained in the earlier sessions. For example, comparing the performance of $NE$ and Algorithm~\ref{alguni} highlights this difference: $NE$ does not utilize information from the first session, whereas Algorithm~\ref{alguni} leverages it to secure a better reward in the second round. This strategic use of information by Algorithm~\ref{alguni} results in a welfare increase of $\frac{1}{8}$ each round, but also generates envy. This is a key element in the algorithm we propose next.
\subsection{Envy-freeness up to $C$}\label{subsec:ext efc}
\begin{algorithm}[t]
\caption{Envy-freeness up to $C$ ($\efc$)}
\label{alg:efc}
\SetAlgoLined
\DontPrintSemicolon
\LinesNumbered
\KwIn{horizon $T$, envy bound $C$}
\For{round $t = 1$ to $T$\label{efcline:for}}{
    pull $a_{1}$\label{efclin:pull_a1}\\
    \lIf{$x^t_1 > \frac{1}{2}$}{pull $a_{1}$\label{efclin:pull_a1_again}}
    \Else{\label{efclin:else}
        \lIf{\textnormal{there exists $r\in [0,1]$ such that $\abs{R^{t-1}_{(1)} + r^t_{(1)} - R^{t-1}_{(2)}-r} > C$}}{pull $a_{1}$\label{efclin:pull_a1_cond}}
        \lElse{pull $a_{2}$\label{efclin:pull_a2_cond_else}}
    }
}
\end{algorithm}

In what follows, we introduce $\efc$, which is an abbreviation of \textbf{E}nvy-\textbf{F}reeness up to $\textbf{C}$, and is
implemented in Algorithm~\ref{alg:efc}. $\efc$ operates by selectively limiting the exploitation of information when the gap between agents' rewards could potentially exceed a predefined envy threshold $C$. This mechanism enforces envy-freeness up to $C$, allowing for better welfare compared to $NE$ while maintaining low envy.

We now describe how $\efc$ works. It interacts with agents for $T$ rounds (Line~\ref{efcline:for}). In every round~$t$, $\efc$ pulls arm $a_1$ for the agent arriving in the first session (Line~\ref{efclin:pull_a1}). The decision to pull $a_1$ first is arbitrary since both arms are identically distributed. If $a_1$ realizes a high reward, i.e., $r^t_{(1)}=x^t_1 > \frac{1}{2}$, $\efc$ pulls it again for the agent in the second session (Line~\ref{efclin:pull_a1_again}). Otherwise, we enter the ``else'' clause in Line~\ref{efclin:else}.

If $a_1$ yields a low reward, the welfare-wise correct action is to pull $a_2$; however, recall that $\efc$ aims to keep the envy lower than $C$. As a result, it ensures that the envy $\abs{R^t_{(1)}-R^t_{(2)}}$ at the end of round $t$ is lower or equal to $C$. Specifically, the ``if'' clause in Line~\ref{efclin:pull_a1_cond} asks whether there exists a realization $r^t_{(2)}=r$ for which the envy would exceed $C$ by the end of the round. If such a realization exists, it pulls $a_1$. Otherwise, it pulls $a_2$ in Line~\ref{efclin:pull_a2_cond_else}. We term $\ef$ the special case of $\efc$ for $C=1$. 
\begin{theorem}\label{thm:ef1evny+sw}
When executing $\efc$ on the instance of Example~\ref{example 1} and $\uniord$, the following hold:
\begin{enumerate}
    \item For all $t$, $\env^t \leq C$ almost surely.
    \item For $C=1$, the social welfare is $\sw \geq \left(1 + \frac{1}{16}\right)T$.
\end{enumerate}
\end{theorem}
Interestingly, Theorem~\ref{thm:ef1evny+sw} implies that $\ef$ recovers half of the social welfare increase due to exploiting information, $(1+\nicefrac{1}{16})T$ versus $(1+\nicefrac{1}{8})T$ for Algorithm~\ref{alguni} and $T$ for $NE$, while limiting the maximal envy to 1 almost surely. We provide a proof sketch below and defer the full proof to \ifapp{Section~\ref{appendix:sociallyopt}}{the appendix}.
\begin{proof}[Proof sketch of Theorem~\ref{thm:ef1evny+sw}]
The first part of the theorem follows directly, as Line~\ref{efclin:pull_a1_cond} allows the envy at $t$ to change only if
\[
\prb{\env^t> C\middle| R^{t-1}_{(1)}, R^{t-1}_{(2)},r^t_{(1)}}=
\prb{\abs{R^{t-1}_{(1)} + r^t_{(1)} - R^{t-1}_{(2)}-r^t_{(2)}} > C\middle| R^{t-1}_{(1)}, R^{t-1}_{(2)},r^t_{(1)}}=0.
\]
The second part of the theorem requires a more detailed argument. We need to understand how often $\efc$ can exploit the information of the first session and enter the ``if'' clause in Line~\ref{efclin:pull_a1_cond}. Since the probability of entering the clause depends on the current level of envy, we must first understand how envy behaves over time. The main technical ingredient we use is the following proposition.
\begin{proposition}\label{prop:ef1 uni dominance}
In every round $t$, the distribution of $\env^t$ is stochastically dominated by the uniform distribution over $[0,1]$. I.e., for any $x \in [0,1]$, it holds that $\prb{\env^t \leq x} \geq x$.
\end{proposition}
Despite its intuitive nature, proving this claim demands careful and thorough case analysis. Equipped with Proposition~\ref{prop:ef1 uni dominance}, we turn to analyze how often $\ef$ pulls $a_2$ after observing a low reward $r^t_{(1)} \leq \frac{1}{2}$ in the first session.
\begin{proposition}\label{prop:ef1 open arm}
        In every round $t$ with $\rt{(1)}\ \leq \frac{1}{2}$, it holds that $\prb{a^t_{(2)} \neq a^t_{(1)} \mid \rt{(1)} \leq \frac{1}{2}} \geq \frac{1}{2}$. 
\end{proposition}
We complete the proof by computing $\sw(\ef)$ via the law of total expectation, using Proposition~\ref{prop:ef1 open arm} to show that the welfare in every round is $1+\frac{1}{16}$, matching the statement of the theorem.
\end{proof}

\subsection{Beyond $C=1$}
The envy analysis in Theorem~\ref{thm:ef1evny+sw} concerns $\ef$, which is a special case of $\efc$ with $C=1$. Unfortunately, our techniques rely heavily on this fact, and extending it would require a different approach. Our preliminary investigation has led us to the following conjecture.
\begin{conjecture}\label{thm: efc sw}
When executing $\efc$ on the instance of Example~\ref{example 1} with any $C \geq 1$ and $\uniord$, the expected social welfare of at least $\sw \geq (1+\frac{1}{8}-\frac{1}{16C})T$.
\end{conjecture}
Simulations we conducted and appear in \ifapp{Section~\ref{sec:simulations}}{the appendix} suggest that Conjecture~\ref{thm: efc sw} holds, and we hope future work could formally prove it.

%% file: sections/discussion.tex
\section{Discussion and Future Work}
\label{sec:discussion}
In this work, we have advanced the understanding of envy dynamics in explore-and-exploit systems. Our stylized model, which assumes reward consistency and sequential agent interactions, assists in characterizing envy under uniform and nudged arrival mechanisms, revealing envy dynamics and accumulation. Under uniform arrival, our results highlight that besides pathologic cases, any algorithm generates an unavoidable expected envy of roughly $\sqrt T$. In contrast, if agent arrival could be nudged, the envy ceases to depend on the horizon $T$. Our results highlight that strategic manipulation of arrival orders through nudging can substantially mitigate envy without altering algorithmic decisions. Furthermore, our preliminary investigation into the welfare-envy tradeoff in Section~\ref{sec:extensions} suggests that algorithms can balance social welfare and envy. 

Due to space limitations, we deferred two important analyses to the appendix. First, we examine the \emph{average envy}, defined as $\envavg^T = \frac{1}{\binom{N}{2}}\sum_{1\leq i<j \leq N}{\abs{\env_{i,j}^T}}$. We show how to leverage our results to obtain meaningful bounds for this measure of envy. Second, we conduct extensive simulations, aiming to empirically validate our theoretical results and perform sensitivity analysis.

We see considerable scope for future work. First, from a technical standpoint, our results contain several gaps. For instance, in nudged arrival, Theorem~\ref{thm: sugg-envy} provides an upper bound of $O\left(\frac{N}{\delta \tdif} \right)$. However, we conjecture that the dependence on the number of agents $N$ should only be logarithmic, if any. Indeed, our simulations support this view. Furthermore, our welfare-envy analysis is preliminary, and future work could extend it to multiple agents, non-uniform distributions, other arrival functions, etc. Another aspect that this paper did not fully address is social welfare maximization under Bayesian information and consistent rewards. Indeed, our goal was to analyze envy in algorithms without heavy assumptions about the way they operate. We consider this problem in \ifapp{Section~\ref{appendix:sociallyopt}}{the appendix}. We draw some interesting connections to Pandora's box and prophet inequalities. We design a socially optimal algorithm for the two-agent case, develop efficient and optimal algorithms for special cases with $N>2$ agents, and propose an inefficient but approximately optimal algorithm for general instances with $N>2$. We suspect that finding a socially optimal algorithm under Bayesian information is NP-hard, as it involves an optimal ordering of arms, which was proved hard for prophet inequalities~\cite{agrawal2020optimal}. Finding optimal algorithms, efficient approximations, or proving hardness remains an open problem.

From a conceptual standpoint, we see several exciting opportunities to deepen our understanding of envy in explore-and-exploit systems. Recall that we assume reward consistency: Rewards are realized only once per round and remain constant across the sessions within that round. This property is crucial, as envy primarily arises when the algorithm exploits information obtained from one agent to benefit another. However, this assumption may be overly restrictive, as it imposes an abrupt shift in reward dynamics across rounds while ignoring potential variability within rounds. A more general modeling approach could introduce structured reward dynamics such as those found in Markov Decision Processes (MDPs), allowing rewards to change gradually from session to session. This would still capture the essence of envy---since agents engaged in exploration would continue to be disadvantaged---but would reflect more realistic environments. 

Another intriguing conceptual direction is to relax the assumption that every agent arrives precisely once in every round. Indeed, if arrival is more chaotic, envy dynamics change remarkably. Agents might arrive multiple times in a round or skip rounds entirely, introducing new forms of informational asymmetry and potentially amplifying or mitigating envy in unexpected ways. Future work could pursue this challenge, exploring how arrival patterns interact with exploration strategies and fairness considerations. 

%% file: sections/appendix-uniform.tex
\section{Omitted Proofs from Section~\ref{sec:uniform}}\label{appendix:uni}
\subsection{Upper Bound}
\begin{proof}[Proof of Proposition~\ref{prop:envy is good SG}]
First, we note that the envy is a martingale.
\begin{observation}\label{envy is martingale}
For every $i,j\in[N]$, the sequence $\left(\env_{i,j}^t \right)_{t=1}^T$ is a martingale. 
\end{observation}
Furthermore, since $\dift$ is symmetric as Remark~\ref{remark: symmetric dif} hints, we can use the following result to connect its tail behavior and its variance. 
\begin{proposition}\label{thm: symmetric bounded sg}
Let $Y$ be a bounded random variable symmetric around $0$, i.e., for any $y \in \mathbb R$ it holds that $\prb{Y \geq y} = \prb{Y \leq -y}$. Then, $Y$ is $\sqrt{\var{Y}}$-subgaussian.
\end{proposition}
We prove Proposition~\ref{thm: symmetric bounded sg} after the end of this proof. Proposition~\ref{thm: symmetric bounded sg} suggests that $\dift$ is $\sqrt{\var{\dift}}$-subgaussian. Ultimately, Lemma~\ref{sg of sup of martingale} analyzes the subgaussianity parameter of the maximum of the martingale $(\env_{i,j}^t )_t$.
\begin{lemma}\label{sg of sup of martingale}
    Let $M^1, M^2,\dots M^T$ be a martingale with increments $Y^1,Y^2,\dots Y^T$, such that $Y^t \mid M^{t-1}$ is $\sigma_t$-SG. Then $\max_t M^t$ is $\left(\sqrt{\sum_{t=1}^T \sigma_t^2}\right)$-SG.
\end{lemma}
Lemma~\ref{sg of sup of martingale} suggests that $\max_{1\leq t\leq T} \env_{i,j}^t $ is 
$\left(\sqrt{\sum_{l=1}^t \var{\dift}}\right)$-SG, thereby concluding the proof of  Proposition~\ref{prop:envy is good SG}.
\end{proof}

\begin{proof}[Proof of Observation~\ref{envy is martingale}]
Recall that $\env_{i,j}^t = \sum^{t}_{l=1}{\Del{l}{i,j}}$. Since the order of selection at time $t$ is independent of $\ordr_t$, it holds that $\E{\env_{i,j}^{t+1}\mid \env_{i,j}^t} =\env_{i,j}^t$. Moreover, since $\E{\abs{\env_{i,j}^t}} \leq  T < \infty$, the stochastic process $\left(\env_{i,j}^t \right)_{t=1}^T$ is a martingale. 
\end{proof}

\begin{proof}[Proposition \ref{thm: symmetric bounded sg}]
    We begin by examining the Taylor polynomial of the function $f(z)=e^z$ around $0$.
    \[
    e^z=  1 + z + \frac{z^2}{2!} + \frac{z^3}{3!} e^{\xi_z},
    \]
    where the last term is the  Lagrange form of the remainder for some $\xi_z \in [-\abs{z}, \abs{z}]$.
    Using this expansion with $z=\lambda y$ and $\xi_{\lambda y} \in [-\abs{\lambda y}, \abs{\lambda y}]$ gets
    \[
    e^{\lambda y} = 1 + \lambda y + \frac{(\lambda y)^2}{2!} + \frac{(\lambda y)^3}{3!} e^{\xi_{\lambda y}}.
    \]
    If $-a<y<a$ then,
    \[
    e^{\lambda y} = 1 + \lambda y + \frac{(\lambda y)^2}{2!} + \frac{(\lambda y)^3}{3!} e^{\abs{\lambda a}}.
    \]
    Thus,
    \[
    \E{e^{\lambda Y}} \leq \E{1 + \lambda Y + \frac{(\lambda Y)^2}{2} + \frac{(\lambda Y)^3}{3} e^{|\lambda a|}} = 
    1 + \lambda \E{Y} + \frac{\lambda^2}{2}\E{Y^2} +\frac{\lambda^3 e^{|\lambda a|}}{3!} \E{Y^3}.
    \]
    Recall that $Y$ is symmetric around $0$ and hence $\E{Y} = \E{Y^3} = 0$, $\var{Y} = \E{Y^2}$.
    Combining these observations with the above we get
    \[
    \E{e^{\lambda Y}} \leq 1 + \frac{\lambda^2}{2}\var{Y} \leq e^{\frac{\lambda^2 \var{Y}}{2}},
    \]
    where the last inequality is due to that $1+z \leq e^z$ for all $z$.
    That sums up the proof of Lemma~\ref{thm: symmetric bounded sg}.
\end{proof}

\begin{proof}[Proof of Lemma~\ref{sg of sup of martingale}]
We begin by proving that $M^t$ is $\left(\sqrt{\sum_{l=1}^t \sigma_l^2}\right)$-SG, and then address $\max_t M^t$

The base case is $t=1$, where we have $M^1=Y^1$ and $Y^1$ is $\sigma_1$-SG by definition. Next, assume that the statement holds for $t=k-1$. It holds that 
\begin{align*}
\E{e^{\lambda M^{k}}}&=\E{\E{e^{\lambda M^{k}}\mid M^{k-1}}}=\E{\E{e^{\lambda (M^{k-1}+Y^k)}\mid M^{k-1}}} = \E{e^{\lambda M^{k-1}} \E{e^{\lambda Y^k}\mid M^{k-1}}}\\
& \leq \E{e^{\lambda M^{k-1}} e^{\frac{\lambda^2 \sigma_k^2}{2}}} \leq e^{\frac{\lambda^2 \sum_{l=1}^{k-1} \sigma_l^2}{2} }e^{\frac{\lambda^2 \sigma_k^2}{2}} = e^{\frac{\lambda^2 \sum_{l=1}^{k} \sigma_l^2}{2} },
\end{align*}
where we used total expectation, the fact that $Y^k\mid M^{k-1}$ is $\sigma_k$-SG and the inductive assumption.

Next, let $\tau$ denote the r.v. for which $M^\tau = \max_t M^t$. It holds that
\begin{align*}
    \E{e^{\lambda \max_t M^t}} &= \E{e^{\lambda M^\tau}} = \E{\E{e^{\lambda M^\tau} \mid \tau}} \leq \E{\E{e^{\lambda M^\tau} \mid \tau}} \leq \E{e^{\frac{\lambda^2 \sum_{l=1}^{\tau} \sigma_l^2}{2} }} \leq e^{\frac{\lambda^2 \sum_{l=1}^{T} \sigma_l^2}{2} },
\end{align*}
where the second to last step follows from the fact that $M^t$ is $\left(\sqrt{\sum_{l=1}^t \sigma_l^2}\right)$-SG. This completes the proof of Lemma~\ref{sg of sup of martingale}.
\end{proof}

\begin{proof}[Proof of Claim~\ref{claim: sg max}]
Fix $a \in \mathbb{R}$ and denote $Y_{\textnormal{max}} = \max_{i \in [n]}{\{Y_i\}}$ for convenience.
We begin by examining $e^{a\E{Y_{\textnormal{max}} }}$.
Using Jensen’s inequality, we have that
\begin{align}\label{eq:asdfrsdgvbsdfg}
e^{a\E{Y_{\textnormal{max}} }} \leq \E{e^{a Y_{\textnormal{max}}}}= \E{\max_{i \in [n]}{\left\{ e^{a Y_i } \right\}}} \leq
\sum_{i=1}^{n}{\E{e^{a Y_i}}} \leq \sum_{i=1}^{n}{ e^{\frac{a^2\sigma^2}{2}} } = ne^{\frac{a^2\sigma^2}{2}}.    
\end{align}
Taking $\ln$ on both sides of Inequality~\eqref{eq:asdfrsdgvbsdfg},
\[
a\E{Y_{\textnormal{max}}}\leq \ln(n) + \frac{a^2\sigma^2}{2}.
\]
Diving by $a$, we obtain
\begin{align}\label{eq:dasolgbdf}
\E{Y_{\textnormal{max}}}\leq \frac{\ln(n)}{a} + \frac{a\sigma^2}{2}.
\end{align}
Inequality~\eqref{eq:dasolgbdf} holds for all $a\in \mathbb{R}$ and specifically for the minimizer of $\frac{\ln{(n)}}{a}\!+\!\frac{a\sigma^2}{2}$, which is $a=\frac{\sqrt{2 \ln{(n)}}}{\sigma}$.
We complete the proof of the claim by substituting $a$ with this value in Inequality~\eqref{eq:dasolgbdf}.
\end{proof}
\begin{proof}[Proof of Proposition~\ref{thm: threshold var}]
    First, since $\dift \leq 1$ almost surely, $\var{\dift} \leq 1$ holds trivially. For the more challenging expression, the proof relies on the fact that when executing an explore-first algorithm, after at most $K$ sessions of any round, all remaining $N-K$ agents must receive the same reward. Hence, at the end of every round $t$, we can find at least $\binom{N-K+1}{2}$ pairs of agents $i,j$ that satisfy $\dif^t_{i,j} = 0$. In other words, for any $i,j$, it holds that 
    \[
    \prb{{\dif^t_{i,j}}=0}\geq \frac{\binom{N-K+1}{2}}{\binom{N}{2}},
    \]
    where the randomness is taken over the stochasticity of the rewards and the arrival order. From that,
    \begin{align*}
        \var{\dift} & = \E{{\dift}^2} - \E{\dift}^2 = \E{{\dift}^2 \mid {\dift}^2 \neq 0}\cdot \prb{{\dift}^2 \neq 0} - \E{\dift}^2
        \nonumber \\&\leq
        1\cdot \left( 1 - \frac{\binom{N-K+1}{2}}{\binom{N}{2}} \right) - 0 =
        1 - \frac{ \left( N-K+1 \right)\left( N-K \right) }{N\left( N-1 \right)}.
    \end{align*}
    With the help of some algebraic operations, we can simplify the expression.
    \begin{align*}
        \var{\dift} & \leq
        1 - \frac{ \left( N-K+1 \right)\left( N-K \right) }{N\left( N-1 \right)} =
        \frac{ N^2 - N - \left( N^2 - 2NK + K^2 + N - K \right) }{ N\left( N-1 \right)} 
        \\&=
        \frac{2NK -2N + K -K^2 }{ N\left( N-1 \right)} =
        \frac{2N(K-1) - K(K-1) }{N\left( N-1 \right)} =
        \frac{(2N - K)(K-1) }{N\left( N-1 \right)}.
    \end{align*}
    This concludes the proof of Proposition~\ref{thm: threshold var}.
\end{proof}
\begin{proof}[Proof of Corollary~\ref{thm: sqrt TK N}]
The corollary holds since 
\begin{align*}
\E{\max_{1\leq t \leq T} \env^t} &\leq
2\sqrt{\ln{(N)}\sum^T_{t=1}{\var{\dift}}}\leq 
2\sqrt{\ln{(N)}\sum^T_{t=1}{\min\left\{1,\frac{2(K-1)}{N-1} \right\}}}\\
&\leq 2\sqrt{T\ln{(N)}{\frac{2(K-1)}{N-1}}}.    
\end{align*}
\end{proof}

\subsection{Lower Bound}\label{appendix:lower bound}
\begin{claim}\label{claim: uni suff}
For the execution in Example~\ref{example: uni suff}, it holds that $\var{\dift} = \frac{1}{12}$.
\end{claim}
\begin{proof}[Proof of Claim~\ref{claim: uni suff}]
We prove the claim by using the definition of variance.
\begin{align}\label{eq:adfgsfhsaagh}
\var{\dift} & = \E{({\dift})^2}-\E{{\dift}}^2\overset{(1)}{=}\E{\left(\rt{(1)} - \rt{(2)}\right)^2} = \E{{\rt{(1)}}^2} - 2\E{\rt{(1)}\rt{(2)}} + \E{{\rt{(2)}}^2}
\end{align}
Next, we apply the threshold structure and properties of the uniform distribution to Equation~\eqref{eq:adfgsfhsaagh}.
{\thinmuskip=0mu
\medmuskip=0mu plus 0mu minus 0mu
\thickmuskip=1mu plus 1mu minus 1mu
\begin{align*}
\textnormal{Eq. }\eqref{eq:adfgsfhsaagh} &= 
\E{{X_1}^2} - \frac{1}{2}\cdot 2\E{X_1 X_1 \mid X_1 \geq \frac{1}{2}} - \frac{1}{2}\cdot 2\E{ X_1 X_2\mid X_1 < \frac{1}{2}}+ \frac{1}{2}\cdot \E{{X_1}^2 \mid X_1 \geq \frac{1}{2}} + \frac{1}{2}\cdot \E{{X_2}^2 \mid X_1 < \frac{1}{2}} \\
&=
\E{{X_1}^2} - \frac{1}{2} \E{{X_1}^2 \mid X_1 \geq \frac{1}{2}} - \E{ X_1\mid X_1 < \frac{1}{2}} \E{ X_2} + \frac{1}{2} \E{{X_2}^2}
\\&=
\frac{3}{2}\E{{\uni{0,1}}^2} - \frac{1}{2} \E{{\uni{\frac{1}{2},1}}^2} - \E{\uni{0,\frac{1}{2}}}= \frac{3}{2} \cdot \frac{1}{3}- \frac{1}{2} \cdot \frac{7}{12} - \frac{1}{4} \cdot \frac{1}{2} \\
&= \frac{1}{12},
\end{align*}}%
where we have used the independence of $X_1$ and $X_2$.
\end{proof}

\begin{claim}\label{claim:example ber suff}
For the execution in Example~\ref{example: ber suff}, it holds that 
\begin{enumerate}
    \item $\var{\Delta^t} \geq \frac{2p(1 - p)}{N}$.
    \item $\var{\Delta^t} \leq  2(1-p)$.
    \item $\var{\Delta^t} \leq  2pK$.
\end{enumerate}
\end{claim}
\begin{proof}[Proof of Claim~\ref{claim:example ber suff}]
Fix two arbitrary agents $i$ and $j$. Recall that Remark~\ref{remark: symmetric dif} ensures that $\Delta_{i,j}^t$ is identically distributed, regardless of the indexes $i$ and $j$. Let the event $B_q$ indicate that the number of arms that realize a value of 0 is $q$, for $q \in \{0, \ldots, K\}$. Further, let $E$ denote the event that $i$ and $j$ receive different rewards. Then we have:
\[
\var{\Delta_{i,j}^t} = \E{(\Delta_{i,j}^t)^2} 
= \sum_{q=0}^K \bigl( \E{(\Delta_{i,j}^t)^2 \mid B_q, E}\Pr(B_q, E) 
+ \E{(\Delta_{i,j}^t)^2 \mid B_q, \overline{E}}\Pr(B_q, \overline{E}) \bigr).
\]
Since $(\Delta_{i,j}^t)^2$ takes the value 1 under event $E$ and 0 otherwise, we get:
\begin{equation}\label{eq:gknmhmgf}
\var{\Delta_{i,j}^t} = \sum_{q=0}^K (1 \cdot \Pr(B_q, E) + 0 \cdot \Pr(B_q, \overline{E})) 
= \sum_{q=1}^K \Pr(B_q, E).    
\end{equation}
Fix any $q\in[K]$. It holds that 
\begin{align}\label{eq:develop for q}
\Pr(B_q, E) = \Pr(B_q)\Pr(E \mid B_q) =  p (1 - p)^{q}\cdot \frac{2\binom{N-2}{q-1}}{\binom{N}{q}} =  2 p (1 - p)^{q}\frac{q(N-q)}{N(N-1)}.
\end{align}
Combining Equations~\eqref{eq:gknmhmgf} and~\eqref{eq:develop for q}, we get
\begin{align}\label{eq:var def}
\var{\Delta_{i,j}^t} = \sum_{q=1}^K 2 p (1 - p)^{q}\frac{q(N-q)}{N(N-1)}.
\end{align}
Since all summands are positive, we obtain the first part of the claim by bounding from below using only the $q=1$ term. That is, we obtain $\var{\Delta_{i,j}^t} \geq  \frac{2p(1 - p)}{N}$. 

For the other parts of the claim, observe that for every $q \in [K]$, $ \frac{q(N-q)}{N(N-1)} \leq 1$; hence, Equation~\eqref{eq:var def} implies that 
\begin{align}\label{eq:gdhfghsdfg}
\var{\Delta_{i,j}^t} \leq \sum_{q=1}^K 2 p (1 - p)^{q}.
\end{align}
From here, we use Inequality~\eqref{eq:gdhfghsdfg} to obtain the second and third parts of the claim. First,
\begin{align*}
\sum_{q=1}^K 2 p (1 - p)^{q} \leq 2 p (1 - p) \sum_{q=0}^\infty (1 - p)^{q} = \frac{ 2 p (1 - p)}{1-(1-p)} = 2(1 - p);
\end{align*}
thus, $\var{\Delta_{i,j}^t} \leq  2(1-p)$ as the second part of the claim implies. Using a different approach to upper bound Inequality~\eqref{eq:gdhfghsdfg}, we get
\begin{align*}
\sum_{q=1}^K 2 p (1 - p)^{q} \leq 2p \sum_{q=1}^K 1^{q}=2pK
\end{align*}
As the third part of the claim asserts. This completes the proof of Claim~\ref{claim:example ber suff}.
\end{proof}

\begin{proof}[Proof of Proposition~\ref{thm: board}]
    The proof of Proposition~\ref{thm: board} relies on the following algebraic inequality, which we prove after this proof.
    \begin{observation}\label{obs: algebric}
        For any $a\geq 0$, it holds that $2a \geq 3a^2 - a^4$.
    \end{observation}
    
    Recall, $Y$ is a non-negative random variable and therefore $\sqrt{\frac{Y}{\E{Y}}}$ is non-negative as well.
    When setting $a=\sqrt{\frac{Y}{\E{Y}}}$ we have
    \begin{align*}
         \frac{2\sqrt{Y}}{\sqrt{\E{Y}}} \geq \frac{3Y}{\E{Y}} - \frac{Y^2}{\E{Y}^2},
    \end{align*}
    for any value $Y$ can take.
    
    Notice that $\E{Y}\geq 0$, hence, by multiplying each side of the inequality by $\frac{\sqrt{\E{Y}}}{2}$ we get
    \begin{align*}
        \sqrt{Y} \geq \sqrt{\E{Y}}\left(\frac{3Y\E{Y} - Y^2}{2\E{Y}^2}\right).
    \end{align*}
    Taking expectation on both sides yields
    \begin{align*}
        \E{\sqrt{Y}} & \geq
        \E{\sqrt{\E{Y}}\left(\frac{3Y\E{Y} - Y^2}{2\E{Y}^2}\right)} \overset{(1)}{=}
        \sqrt{\E{Y}}\left(\E{\frac{2Y\E{Y}}{2\E{Y}^2}} - \E{\frac{Y^2 - Y\E{Y}}{2\E{Y}^2}}\right)
        \\&\overset{(2)}{=}
        \sqrt{\E{Y}}\left(\frac{2\E{Y}\E{Y}}{2\E{Y}^2} - \frac{\E{Y^2} - \E{Y}\E{Y}}{2\E{Y}^2}\right) \overset{(3)}{=}
        \sqrt{\E{Y}}\left(1 - \frac{\var{Y}}{2\E{Y}^2}   \right) ,
    \end{align*}
    where $(1)$ and $(2)$ hold due to linearity of expectation and $(3)$ is by the definition of variance. This concludes the proof of Proposition~\ref{thm: board}.
\end{proof}

\begin{proof}[Proof of Observation~\ref{obs: algebric}]
    To prove the inequality, it is sufficient to prove that the function $f(a) = a^3 - 3a^2 -2$ is non-negative for all $a\geq 0$.
    It holds that $f'(a) = 3 a^2 - 3$; thus
    \begin{align*}
        f'(a) > 0, & \text{ when } 0 \leq a < 1 \\
        f'(a) = 0, & \text{ when } a = 1 \\
        f'(a) < 0, & \text{ when } 1 \leq a.
    \end{align*}
    I.e., $f(a)$ is monotonically decreasing for $0\leq a <1$ and monotonically increasing for $1 \leq a$. 
    Hence, for all $a \geq 0$ it holds that $f(1) \geq f(a) = 0$
\end{proof}

\begin{proof}[Proof of Proposition~\ref{prop:insufficient}]
We prove the proposition by reiterating the proof of Theorem~\ref{thm: uni lower-bound} while avoiding using the definition of sufficient execution. It suffices to show that the left-hand side of Inequality~\eqref{eq:m,bnhjikw} is constant. Starting with the numerator,
\begin{align}\label{eq:lpods}
\var{ \sum^T_{t=1}{ \dift^2} } &= \sum^T_{t=1} \var{ \dift^2} =  \sum^T_{t=1} \left(\E{(\dift)^4}-\E{(\dift)^2}^2 \right) = \sum^T_{t=1} \left( \var{\dift} - \var{\dift}^2 \right) \nonumber \\
&= T\left( \var{\dift} - \var{\dift}^2 \right),
\end{align}
where we have used the fact that the algorithm is stationary over rounds and that $\dift$ only takes values in the $\{0,1\}$ set. We keep using the superscript $t$ in Equation~\eqref{eq:m,bnhjikw} to ease readability, and it could be any arbitrary $t \in [T]$.

Next, we consider the denominator of the left-hand side of Inequality~\eqref{eq:m,bnhjikw}.
\begin{align}\label{eq:ydots}
2 \E{ \sum^T_{t=1}{(\dift)^2} }^2 =2  \left(\sum^T_{t=1}{\E{(\dift)^2} }\right)^2 = 2  \left(\sum^T_{t=1}{\var{\dift} }\right)^2 =  2 \left(T{\var{\dift} }\right)^2 = 2T^2 \var{\dift}^2
\end{align}
Combining Equations~\eqref{eq:lpods} and~\eqref{eq:ydots}, we get
\begin{align}\label{eq:prea}
\frac{T\left( \var{\dift} - \var{\dift}^2 \right)}{2T^2 \var{\dift}^2}=\frac{\left( 1 - \var{\dift} \right)}{2T \var{\dift}} \leq 
\frac{1}{2T \var{\dift}} \leq \frac{1}{2T}\frac{N}{2p(1-p)},
\end{align}
where the last inequality follows from Claim~\ref{claim:example ber suff}. 
Furthermore, recall that the proposition assumption guarantees that $p\in \left[\frac{N}{cT}, 1-\frac{N}{cT}\right]$ for $c\geq 2$, suggesting that
\[
p (1-p)\geq \frac{N}{cT}\left(1- \frac{N}{cT}\right) \geq \frac{N}{cT}\left(1- \frac{1}{c}\right) \geq \frac{N}{cT}\left(\frac{c-1}{c}\right) \geq \frac{N}{cT}\cdot \frac{1}{2}.
\]
Plugging this into the right-hand-side of Inequality~\eqref{eq:prea},
\begin{align}
 \frac{1}{2T}\frac{N}{2p(1-p)} \leq \frac{N}{4T} \cdot \frac{1}{p(1-p)} \leq \frac{N}{4T} \cdot \frac{2c T}{N}  = \frac{c}{2}.
\end{align}
Having observed that the left-hand-side of Inequality~\eqref{eq:m,bnhjikw} is bounded by a constant w.r.t. $T$, we complete the proof by plugging this constant into Inequality~\eqref{thm: mp uni lower-bound 2}.
\end{proof}

%% file: sections/appendix-nudge.tex
\section{Omitted Proofs from Section~\ref{sec:nudge}}\label{appendix:nudge}

\begin{proof}[Proof of Proposition~\ref{prop G less than M}]
We prove the claim by induction over $\tau$. The first round in the excursion $D(t)$ and our base case is $\tau=\underline{t}+1$. Since the rewards are in the $[0,1]$ interval and $G_i^{\underline{t}} \leq 1$, we know that $G_i^{\underline{t}+1} \leq 2 = M_i^{\underline{t}+1}$.

Next, assume the claim holds for $\tau$; thus, $G_i^\tau \leq M_i^\tau$. Recall that we are guaranteed that $\tau \in D(t)$. Without loss of generality, assume that at time $\tau$ agents are ordered lexicographically. Particularly,  ${\sigma^\tau(i)}=i,{\sigma^\tau(i+1)}=i+1$  and $G_i^\tau = R^\tau_{\sigma^\tau(i)}- R^\tau_{\sigma^\tau(i+1)}= R^\tau_i- R^\tau_{i+1}$. 
Next, observe that 
\begin{equation}\label{eq:asdgndsfjghm}
    R^{\tau+1}_{\sigma^{\tau+1}(i)} = \min_{j\in [i]}\left\{ 
    R^{\tau}_{j}+r^{\tau+1}_{j} 
    \right\}\leq R^{\tau}_{i}+r^{\tau+1}_{i}.
\end{equation}
Inequality~\eqref{eq:asdgndsfjghm} holds due to our assumption that the rewards are ordered according to agent indices at round $\tau$.  and since no agent in the set $[N]\setminus[i]$ could obtain a higher cumulative reward that agents $[i]$ at round $\tau + 1$ since all rewards are bounded by 1 and $G^\tau_i > 1$. Similarly,
\begin{equation}\label{eq:dsasdhhgtnt}
    R^{\tau+1}_{\sigma^{\tau+1}(i+1)} = \max_{j\in [N]\setminus[i]}\left\{ 
    R^{\tau}_{j}+r^{\tau+1}_{j} 
    \right\}\geq R^{\tau}_{i+1}+r^{\tau+1}_{i+1}.
\end{equation}
Combining Inequalities~\eqref{eq:asdgndsfjghm} and~\eqref{eq:dsasdhhgtnt}, we derive that
\begin{align*}
G_i^{\tau+1} &= R^{\tau+1}_{\sigma^{\tau+1}(i)} - R^{\tau+1}_{\sigma^{\tau+1}(i+1)} \leq R^{\tau}_{i}+r^{\tau+1}_{i} - R^{\tau}_{i+1}-r^{\tau+1}_{i+1} \\
& =G_i^\tau + r^{\tau+1}_{i} - r^{\tau+1}_{i+1} \leq M_i^\tau + r^{\tau+1}_{i} - r^{\tau+1}_{i+1} \\
&= M_i^\tau + r^{\tau+1}_{(i)} - r^{\tau+1}_{(i+1)} = M_i^{\tau+1},
\end{align*}
where we have used the inductive assumption and the fact that nudged arrival order sorts agents in a non-increasing order of rewards (Algorithm~\ref{alg: sugg arr}, Line~\ref{line:mapping}). This completes the proof of Proposition~\ref{prop G less than M}.
\end{proof}

\begin{proof}[Proof of Proposition~\ref{prop: sugg-m concentration}]
The recursive definition of $M^\tau$ implies that for every $\tau \in D(t)$, $M^\tau = 2+ \sum_{n= \underline{t}+2 }^{\tau} r^n_{(i)}-r^n_{(i+1)}$; thus, 
\begin{align}\label{eq:fghbdfgh}
\prb{M^t > n}  = \prb{\sum_{l= \underline{t}+2 }^{\tau} r^l_{(i)}-r^l_{(i+1)}> n-2}
\end{align}
Next, let $B^l$ denote the event that $r^l_{(i)}-r^l_{(i+1)} \neq 0$. Furthermore, let $B(\tau)$ denote the (random) set of rounds for which the event $B^l$ occurs between $\underline{t}+2$ and $\tau$. That is,
\[
B(\tau) = \{ l\mid \underline{t}+2 \leq l \leq \tau, \ind{B^l} \}
\]
As a result, due to Property~\ref{prop:nudge} and the definition of $\tdif$ in Equation~\eqref{eq def tdif},
\begin{equation}\label{eq:sgdjfndb}
\E{r^l_{(i)}-r^l_{(i+1)} \mid B^l} \leq - \delta \tdif.
\end{equation}
Rewriting Equation~\eqref{eq:fghbdfgh},
\begin{align}\label{eq:hbngaersd}
\prb{M^t > n}  &= \prb{\sum_{l \in B(\tau)} r^l_{(i)}-r^l_{(i+1)}> n-2} \nonumber \\
& =\sum_{b\subseteq \{\underline{t}+2,\dots \tau \}}\prb{\sum_{l \in b} r^l_{(i)}-r^l_{(i+1)}> n-2 \mid  B(\tau) = b}\prb{B(\tau) = b} \nonumber\\
& \stackrel{*}{=} \sum_{b\subseteq \{\underline{t}+2,\dots \tau \}, \abs{b}\geq n-2}\prb{\sum_{l \in b} r^l_{(i)}-r^l_{(i+1)}> n-2 \mid  B(\tau) = b}\prb{B(\tau) = b} \nonumber\\
& \leq \max_{b\subseteq \{\underline{t}+2,\dots \tau \}, \abs{b}\geq n-2} \prb{\sum_{l \in b} r^l_{(i)}-r^l_{(i+1)}> n-2 \mid  B(\tau) = b} \nonumber \\
& = \max_{b\subseteq \{\underline{t}+2,\dots \tau \}, \abs{b}\geq n-2} \prb{\sum_{l \in b} r^l_{(i)}-r^l_{(i+1)} + \delta \tdif \abs{b}> n-2+\delta \tdif \abs{b} \mid  B(\tau) = b},
\end{align}
where the change in the set over which we sum in $*$ follows since $\abs{r^l_{(i)}-r^l_{(i+1)}}\leq 1$ almost surely. Striving to bound the above, notice that, conditioned on $B(\tau) = b$, $\sum_{l \in b} \left( r^l_{(i)}-r^l_{(i+1)} + \delta \tdif \right)$ forms a super-martingale. Using Azuma-Hoeffding inequality,
\begin{align*}
\textnormal{Inequality }\eqref{eq:hbngaersd} &\leq \max_{b\subseteq \{\underline{t}+2,\dots \tau \}, \abs{b}\geq n-2} \exp\left\{-\frac{(n-2+\delta \tdif \abs{b})^2}{2 \sum_{l\in b} (1+\delta \tdif)^2 }\right\} \nonumber \\
& \stackrel{\delta \tdif \leq 1}{\leq}  \max_{b\subseteq \{\underline{t}+2,\dots \tau \}, \abs{b}\geq n-2} \exp\left\{-\frac{(n-2)^2+(n-2)(\delta \tdif \abs{b})+(\delta \tdif \abs{b})^2}{8\abs{b}}\right\} \nonumber
\\
& =  \max_{b\subseteq \{\underline{t}+2,\dots \tau \}, \abs{b}\geq n-2} \exp\Bigg\{-\frac{(n-2)(\delta \tdif)}{8}-\underbrace{\frac{(n-2)^2+(\delta \tdif \abs{b})^2}{8\abs{b}}}_{\geq 0}\Bigg\} \nonumber \\
& \leq  \max_{b\subseteq \{\underline{t}+2,\dots \tau \}, \abs{b}\geq n-2} \exp\left\{\frac{-(n-2)\delta \tdif}{8}\right\} \nonumber \\
& = \exp\left\{-\frac{(n-2)(\delta \tdif)}{8}\right\}.
\end{align*}
This completes the proof of Proposition~\ref{prop: sugg-m concentration}.
\end{proof}

%% file: sections/appendix-nudge-models.tex
\section{Models Captured by the Nudged Arrival Property}\label{appendix:nudge-models}
\subsection*{Mallows Model~\cite{mallows1957non}}

The Mallows model prioritizes rankings close to a reference order \( \sigma \). The probability of a sampled ranking \( \pi \) is proportional to \( e^{-\beta d(\pi, \sigma)} \), where \( d(\pi, \sigma) \) is the Kendall's tau distance between \( \pi \) and \( \sigma \), and \( \beta \geq 0 \) is the concentration parameter. For any pair of agents \( i, j \) such that \( \sigma^{-1}(i) < \sigma^{-1}(j) \), the probability that \( i \) precedes \( j \) satisfies (following the detailed argument of \cite[Section 2]{lu2014effective}):
\[
\Pr_{\pi \sim \text{Mallows}}(\pi^{-1}(i) < \pi^{-1}(j)) = \frac{e^{-\beta}}{1 + e^{-\beta}}.
\]
By arithmetic manipulation we get $\delta = \frac{1 - e^{-\beta}}{1 + e^{-\beta}}$.

\subsection*{Plackett-Luce Model~\cite{marden1996analyzing}}
In the Plackett-Luce model, each agent \( i \) is assigned a positive score \( w_i > 0 \), and the probability of observing a ranking \( \pi \) is given by:
\[
\Pr(\pi) = \prod_{k=1}^{N} \frac{w_{\pi(k)}}{\sum_{j=k}^{N} w_{\pi(j)}}.
\]
For any pair \( i, j \), the probability that \( i \) precedes \( j \) is:
\[
\Pr(\pi^{-1}(i) < \pi^{-1}(j)) = \frac{w_i}{w_i + w_j},
\]
which holds since the model satisfies Luce's choice axiom, which guarantees the independence of the pairwise ranking probabilities from the presence of other options \cite{luce1959individual}. Thus, by setting the scores such that \( \frac{w_i}{w_i + w_j} \geq \frac{1+\delta}{2} \) for every $i$ so that $i$ precedes $j$ in the optimal permutation, the nudged arrival property is satisfied. One way to do it is at each round $t$, recursively set $w_{\pi(1)}^{t} = 1, w_{\pi(i+1)}^{t} = \frac{1+\delta}{1 - \delta} w_{\pi(i)}^{t}$. We thus assume that the weights are not global across rounds but are round-dependent and are adjusted based on the accumulated rewards. This can be interpreted as either the designer nudging different agents more forcefully, or alternatively, in a behavioral approach, users that benefited more from the system in the past are willing to cooperate more with its nudges. 







\subsection*{Thurstone-Mosteller Model \cite{ThurstoneModel}}

In the Thurstone-Mosteller model, each agent \( i \) is assigned an independently drawn latent cardinal value \( v_i \sim \mathcal{N}(\mu_i, s^2) \). Thus, the probability that \( i \) is ranked above \( j \) is independent of all other draws besides the pairwise draws, and is exactly the probability that drawing from $i$'s normal variable exceeds drawing from $j$'s normal variable. The difference of two normal variables is normal by itself, with mean $\mu_i - \mu_j$ (the means difference) and variance $2s^2$ (the sum of variances), and we are interested in the probability that this variable is above $0$. We can normalize and shift the mean, and get that the probability that \( i \) is ranked above \( j \) is:
\[
\Pr(\pi^{-1}(i) < \pi^{-1}(j)) = \Phi\left(\frac{\mu_i - \mu_j}{\sqrt{2}s}\right),
\]
where \( \Phi \) is the CDF of the standard normal distribution $N(0,1)$. We can thus tune $\mu_i, \mu_j$ (with some globally-set $s$) so that for every $i,j$,  \( \Pr(\pi^{-1}(i) < \pi^{-1}(j)) \geq \frac{1+\delta}{2} \), and the nudged arrival property is satisfied. One way to do that is by calculating the constant shift of the mean $\delta\mu$ that satisfies \( \Pr(\pi^{-1}(i) < \pi^{-1}(j)) = \frac{1+\delta}{2}\), and have
\[
\mu_{\pi(i)}^{t} = \mu_{\pi(1)}^{t} + (i-1)\cdot \delta\mu. 
\]



%% file: sections/appendix-efc.tex
\section{Omitted Proofs from Section~\ref{sec:extensions}}
\newcommand{\at}[1]{a_{#1}^t}
\paragraph{Additional notation for this section} The analysis in this section makes extensive use of the notation $\at{(q)}$ for $q \in\{1,2\}$ and $t\in [T]$, denoting the arm pulled in session $i$ of round $t$. 
\begin{proof}[Proof of Theorem~\ref{thm:ef1evny+sw}]
The first part of the proof is given in the body of the paper; hence, we move to the second part, i.e., showing that $\sw = (1+\frac{1}{16})T$.

Fix any $t\in[T]$. We analyze the expected sum of rewards obtained in round $t$, $\E{\rt{(1)}+\rt{(2)}}$.
Notice that $\E{\rt{(1)}} = \E{a_1} = \frac{1}{2}$.
As for $\rt{(2)}$, we are uncertain about the arm the algorithm pulls, but can use total expectation:
\begin{align}\label{eq: ef11}
    \E{\rt{(2)}} & = \E{\rt{(2)} \mid \rt{(1)} > \frac{1}{2}}\cdot \prb{\rt{(1)} > \frac{1}{2}} +  \E{\rt{(2)} \mid \rt{(1)} \leq \frac{1}{2}}\cdot \prb{\rt{(1)} \leq \frac{1}{2}}
    \nonumber\\&= \frac{3}{8} + \frac{1}{2}\cdot \E{\rt{(2)} \mid \rt{(1)} \leq \frac{1}{2}},
\end{align}
where we have used Line~\ref{efclin:pull_a1_again} of $\efc$ for replacing $\E{\rt{(2)} \mid \rt{(1)} > \frac{1}{2}}$ with $\frac{3}{4}$, since arm $a_1$ is pulled for the second session as well. Simplifying the term $\E{\rt{(2)} \mid \rt{(1)} \leq \frac{1}{2}}$ and using the fact that $\E{\rt{(r)} \mid \rt{(1)} \leq \frac{1}{2}}=\frac{1}{4}$, we get,
\begin{align*}
    \E{\rt{(2)} \mid \rt{(1)} \leq \frac{1}{2}} &
    = \E{\rt{(2)} \mid \at{(2)} = \at{(1)}, \rt{(1)}} \cdot \prb{\at{(2)} = \at{(1)} \mid \rt{(1)} \leq \frac{1}{2}}
    \\& \qquad + \E{\rt{(2)} \mid \at{(2)} \neq \at{(1)}, \rt{(1)}} \cdot \prb{\at{(2)} \neq \at{(1)} \mid \rt{(1)} \leq \frac{1}{2}}
    \\&= \frac{1}{4} \cdot \prb{\at{(2)} = \at{(1)} \mid \rt{(1)} \leq \frac{1}{2}} + \frac{1}{2} \cdot \prb{\at{(2)} \neq \at{(1)} \mid \rt{(1)} \leq \frac{1}{2}}
    \\& = \frac{1}{4} + \frac{1}{4}\cdot \prb{\at{(2)} \neq \at{(1)} \mid \rt{(1)} \leq \frac{1}{2}}.
\end{align*}
Consequently, all that is left is to understand how often $\efc$ pulls the second arm when the first arm yields a low reward. Using Proposition~\ref{prop:ef1 open arm}, we obtain
\[
\E{\rt{(2)} \mid \rt{(1)} \leq \frac{1}{2}} =
\frac{1}{4} + \frac{1}{4} \cdot \prb{\at{(2)} \neq \at{(1)} \mid \rt{(1)} \leq \frac{1}{2}} \geq \frac{3}{8}.
\]
Using the above inequality and Equation~\eqref{eq: ef11}, we get
\[
\E{\rt{(2)}} \geq \frac{3}{8} + \frac{1}{2}\cdot \frac{3}{8}= \frac{9}{16}.
\]
Since this holds for any arbitrary $t$, by summing over all rounds, we get
\[
SW(EF1) = \E{\sum_{t=1}^{T}{ \rt{(1)} + \rt{(2)} }} = \sum_{t=1}^{T}{ \E{ \rt{(1)} + \rt{(2)} } } \geq \sum_{t=1}^{T}{ \frac{1}{2} + \frac{9}{16}} = \left( 1+ \frac{1}{16} \right)T.
\]
This concludes the proof of Theorem~\ref{thm:ef1evny+sw}.
\end{proof}

\begin{proof}[Proof of Proposition~\ref{prop:ef1 uni dominance}]
We prove the claim with induction over the round index $t$.
The base step, i.e., $t=0$, is straightforward. Fix any $x\in [0,1]$, and observe that
\[
\prb{\env^0 \leq x} = \prb{0 \leq x} = 1 \geq x.
\]
We move forward to the inductive step. Assume the claim holds for round $t-1$, and let us prove the claim for~$t$. First, notice that if $\at{(2)} = \at{(1)}$, then $\env^t=\env^{t-1}$.
Based on the inductive assumption, the distribution of $\env^{t-1}$ is stochastically dominated by $\uni{0,1}$, and thus so is the distribution of $\env^t$.

Otherwise, from here on we assume $\at{(2)} \neq \at{(1)}$. We continue with an extensive case analysis. We define the following six events $A_1,\dots, A_6$. Each event consists of the conditions that cause the algorithm to pull a different arm in the second session and the outcome of that round:
\begin{align*}
    &A_1 := \left( \rt{(1)} \leq \frac{1}{2} \right) \wedge \left(R^{t-1}_{(1)} = R^{t-1}_{(2)}\right) \wedge  \left( R^{t-1}_{(1)} + \rt{(1)} \geq R^{t-1}_{(2)} + \rt{(2)} \right), \\
    &A_2 := \left( \rt{(1)} \leq \frac{1}{2} \right) \wedge \left(R^{t-1}_{(1)} = R^{t-1}_{(2)}\right) \wedge  \left( R^{t-1}_{(2)} + \rt{(2)} > R^{t-1}_{(1)} + \rt{(1)} \right), \\
    &A_3 := \left( \rt{(1)} \leq \frac{1}{2} \right) \wedge \left(R^{t-1}_{(1)} > R^{t-1}_{(2)} \right) \wedge \left( R^{t-1}_{(1)} - R^{t-1}_{(2)} \leq 1 - \rt{(1)} \right) \wedge \left( R^{t-1}_{(1)} + \rt{(1)} \geq R^{t-1}_{(2)} + \rt{(2)} \right), \\
    &A_4 := \left( \rt{(1)} \leq \frac{1}{2} \right) \wedge \left(R^{t-1}_{(1)} > R^{t-1}_{(2)} \right) \wedge \left( R^{t-1}_{(1)} - R^{t-1}_{(2)} \leq 1 - \rt{(1)} \right) \wedge \left( R^{t-1}_{(2)} + \rt{(2)} > R^{t-1}_{(1)} + \rt{(1)} \right), \\
    &A_5 := \left( \rt{(1)} \leq \frac{1}{2} \right) \wedge \left( R^{t-1}_{(2)} > R^{t-1}_{(1)} \right) \wedge  \left( R^{t-1}_{(2)} - R^{t-1}_{(1)} \leq \rt{(1)} \right) \wedge \left( R^{t-1}_{(1)} + \rt{(1)} \geq R^{t-1}_{(2)} + \rt{(2)} \right), \\
    &A_6 := \left( \rt{(1)} \leq \frac{1}{2} \right) \wedge \left( R^{t-1}_{(2)} > R^{t-1}_{(1)} \right) \wedge  \left( R^{t-1}_{(2)} - R^{t-1}_{(1)} \leq \rt{(1)} \right) \wedge \left( R^{t-1}_{(2)} + \rt{(2)} > R^{t-1}_{(1)} + \rt{(1)} \right).
\end{align*}
Notice that 
\begin{observation}\label{obs:partition}
Given $\at{(2)} \neq \at{(1)}$, the events $A_1, \dots, A_6$ partition the space of all options for $R^{t-1}_{(1)}, R^{t-1}_{(2)}, r^t_{(1)}, r^t_{(2)}$.
\end{observation}
Equipped with Observation~\ref{obs:partition}, we turn to analyze $\env^t$ under $A_1, \dots, A_6$. Fix any arbitrary $x \in [0,1]$.
\begin{itemize}
    \item Case $A_1$.
    Under the conditions of event $A_1$ we have that $\rt{(1)} \geq  \rt{(2)}$.
    Recall that $\rt{(2)} \sim \uni{0,1}$, but considering the latter we know that $ \rt{(2)}\mid A_1 \sim \uni{0,\rt{(1)}}$.
    Given that $R^{t-1}_{(1)} + \rt{(1)} \geq R^{t-1}_{(2)} + \rt{(2)}$, we know the envy at the end of round $t$ is exactly $\env^t = R^{t-1}_{(1)} + \rt{(1)} - R^{t-1}_{(2)} - \rt{(2)} = \rt{(1)} -\rt{(2)}$;
    hence, $\env^{t}\mid A_1 \sim \uni{\rt{(1)} - \rt{(1)},\rt{(1)} - 0}$, i.e., $\env^{t}\mid A_1 \sim \uni{0,\rt{(1)}}$. Therefore,
    \begin{align*}
        \prb{\env^t \leq x \mid A_1} & = \prb{\uni{0, r^t_{(1)}}\leq x \mid r^t_{(1)} \leq \frac{1}{2}} \geq \prb{\uni{0, \frac{1}{2}} \leq x}
        \\& \geq \prb{\uni{0, 1} \leq x} = x.
    \end{align*}

    \item Case $A_2$.
    This case is similar to that of $A_1$, only now $ \rt{(2)}\mid A_2 \sim \uni{\rt{(1)}, 1}$ and $\env^t = \rt{(2)} -\rt{(1)}$, resulting with $\env^{t}\mid A_2 \sim \uni{\rt{(1)} - \rt{(1)},1 -\rt{(1)}}$, i.e., $\env^{t}\mid A_2 \sim \uni{0,1 -\rt{(1)}}$.
    Hence,
    \begin{align*}
        \prb{\env^t \leq x \mid A_2 } \geq \prb{\uni{0, 1} \leq x} = x.
    \end{align*}

    \item Case $A_3$.
    Under the conditions of $A_3$ we have that the envy after $t$ rounds is exactly $\env^t = R^{t-1}_{(1)} + \rt{(1)} -\left(R^{t-1}_{(2)} + \rt{(2)}\right)$.
    Since $R^{t-1}_{(1)} + \rt{(1)} \geq R^{t-1}_{(2)} + \rt{(2)}$, $\rt{(2)}$ is now a uniform random variable between $0$ and the minimum between $\left\{1,  R^{t-1}_{(1)} + \rt{(1)} - R^{t-1}_{(2)}\right\}$.
    Due to the guarantee $R^{t-1}_{(1)} - R^{t-1}_{(2)} \leq 1 - \rt{(1)}$ we can finally see that
    \[
    - \rt{(2)} \mid A_3 \sim \uni{-\left(R^{t-1}_{(1)} + \rt{(1)} - R^{t-1}_{(2)}\right), 0};
    \]
    thus,
    \[
    \env^t \mid A_3 \sim \uni{R^{t-1}_{(1)} + \rt{(1)} -R^{t-1}_{(2)} -\left(R^{t-1}_{(1)} + \rt{(1)} - R^{t-1}_{(2)}\right),  R^{t-1}_{(1)} + \rt{(1)} -R^{t-1}_{(2)}}.
    \]
    Finally,
    \begin{align*}
        \prb{\env^t \leq x | A_3} & = \prb{\uni{0, R^{t-1}_{(1)} - R^{t-1}_{(2)}+ \rt{(1)}} \leq x | A_3}
        \\& \geq \prb{\uni{0, 1} \leq x} = x.
    \end{align*}

    \item Case $A_4$.
    Under the conditions of $A_4$ we have that $\rt{(2)}$ is a uniform random variable distributed between $R^{t-1}_{(1)} - R^{t-1}_{(2)} + \rt{(1)}$ and $1$.
    The envy after round $t$ is exactly $R^{t-1}_{(2)} +\rt{(2)}- R^{t-1}_{(1)}-\rt{(1)}$ and thus it is a uniform random variable between $R^{t-1}_{(2)} - R^{t-1}_{(1)}-\rt{(1)} +\left( R^{t-1}_{(1)} - R^{t-1}_{(2)} + \rt{(1)} \right)$ and $R^{t-1}_{(2)} - R^{t-1}_{(1)}-\rt{(1)} +1$.
    I.e., $\env^t \mid A_4 \sim \uni{0,R^{t-1}_{(2)} - R^{t-1}_{(1)}-\rt{(1)} +1}$. Recall $A_4$ suggests $0 > R^{t-1}_{(2)} -R^{t-1}_{(1)} $; finally,
    \begin{align*}
        \prb{\env^t \leq x | A_4} & = \prb{\uni{0, R^{t-1}_{(2)} - R^{t-1}_{(1)}-\rt{(1)} +1} \leq x | A_4}  \\& \geq
        \prb{\uni{0, 0 - \rt{(1)} +1} \leq x } \geq \prb{\uni{0, 1} \leq x} = x,
    \end{align*}
    as $\rt{(1)}\geq 0$.
    Note that under $A_4$ it holds that $R^{t-1}_{(1)} - R^{t-1}_{(2)} \leq 1 - \rt{(1)} $ and thus $0 \leq R^{t-1}_{(2)} - R^{t-1}_{(1)}-\rt{(1)} +1$ almost surely.
    
    \item Case $A_5$.
    Under the conditions of $A_5$, similarly to $A_3$, the envy at the end of round $t$ is exactly $\env^t = R^{t-1}_{(1)} + \rt{(1)} - R^{t-1}_{(2)} + \rt{(2)}$ and
    \[
    - \rt{(2)} \mid A_5 \sim \uni{-\left(R^{t-1}_{(1)} + \rt{(1)} - R^{t-1}_{(2)}\right), 0};
    \]
    thus,
    \[
    \env^t \mid A_5 \sim \uni{R^{t-1}_{(1)} + \rt{(1)} -R^{t-1}_{(2)} -\left(R^{t-1}_{(1)} + \rt{(1)} - R^{t-1}_{(2)}\right),  R^{t-1}_{(1)} + \rt{(1)} -R^{t-1}_{(2)}}.
    \]
    Finally,
    \begin{align*}
        \prb{\env^t \leq x | A_5} & = \prb{\uni{0, R^{t-1}_{(1)} - R^{t-1}_{(2)}+ \rt{(1)}} \leq x | A_5}
        \\& \geq \prb{\uni{0, 0 + \rt{(1)}} \leq x | A_5}
        \geq \prb{\uni{0, \frac{1}{2}} \leq x} 
        \\& \geq \prb{\uni{0, 1} \leq x} = x.
    \end{align*}

    \item Case $A_6$.
    Under the conditions of $A_6$, similarly to $A_4$, the envy at the end of round $t$ is exactly $\env^t = R^{t-1}_{(2)} + \rt{(2)} - R^{t-1}_{(1)} - \rt{(1)}$ and
    \[
    \rt{(2)} \mid A_6 \sim \uni{R^{t-1}_{(1)}+\rt{(1)} -R^{t-1}_{(2)}, 1};
    \]
    thus,
    \[
    \env^t \mid A_6 \sim \uni{R^{t-1}_{(2)} - R^{t-1}_{(1)} - \rt{(1)} + R^{t-1}_{(1)}+\rt{(1)} -R^{t-1}_{(2)}, R^{t-1}_{(2)} - R^{t-1}_{(1)} - \rt{(1)} + 1}.
    \]
    Finally,
    \begin{align*}
        \prb{\env^t \leq x | A_6} & = \prb{\uni{0,R^{t-1}_{(2)} - R^{t-1}_{(1)} - \rt{(1)} + 1} \leq x | A_6}
        \\& \geq \prb{\uni{0,1} \leq x} = x,
    \end{align*}
    where the inequality holds due to $R^{t-1}_{(2)} - R^{t-1}_{(1)} \leq \rt{(1)}$.    
\end{itemize}
We have shown that the inductive step holds under all cases; thereby, the proof of Proposition~\ref{prop:ef1 uni dominance} is complete.
\end{proof}
\begin{proof}[Proof of Proposition~\ref{prop:ef1 open arm}]
We prove the statement using case analysis. We partition the space of events $a^t_{(2)} \neq a^t_{(1)}$ conditioning on $\rt{(1)} \leq \frac{1}{2}$:
\begin{align*}
    &B_1 := R^{t-1}_{(1)} = R^{t-1}_{(2)} \\
    &B_2 := \left(R^{t-1}_{(1)} > R^{t-1}_{(2)} \right) \wedge \left( R^{t-1}_{(1)} - R^{t-1}_{(2)} \leq 1 - \rt{(1)} \right) \\
    &B_3 := \left( R^{t-1}_{(2)} > R^{t-1}_{(1)} \right) \wedge  \left( R^{t-1}_{(2)} - R^{t-1}_{(1)} \leq \rt{(1)} \right).
\end{align*}
Therefore
\begin{align*}
    \prb{a^t_{(2)} \neq a^t_{(1)} \mid \rt{(1)} \leq \frac{1}{2}} & = \prb{ B_1\vee B_2 \vee B_3  \mid \rt{(1)} \leq \frac{1}{2}}
    \\& = \prb{B_1 \mid \rt{(1)} \leq \frac{1}{2} } + \prb{B_2 \mid \rt{(1)} \leq \frac{1}{2} } + \prb{B_3 \mid \rt{(1)} \leq \frac{1}{2} }.
\end{align*}
We prove each part separately, beginning with the event $B_1$.

Since the distributions of the rewards are continuous, event $B_1$ occurs if and only if until round $t$ both agents receive the same rewards from the same arm.
In this case, the algorithm pulls the same arm in both sessions if it yields a reward greater than $\frac{1}{2}$.
Therefore, we must have
\[
    \prb{r_{(1)}^\tau  = r_{(2)}^\tau} = \prb{r_{(1)}^\tau  > \frac{1}{2}}
\]
for all $\tau < t$; hence,
\begin{align}\label{B1}
\prb{B_1 \mid \rt{(1)} > \frac{1}{2}} =\prb{B_1} = \prb{\forall \tau<t : r_{(2)}^\tau > \frac{1}{2}} = \left(\frac{1}{2}\right)^{t-1}.
\end{align}

Next, we examine event $B_2$. Using Bayes formula,
\begin{align*}
    &\prb{B_2 \mid \rt{(1)} \leq \frac{1}{2} } =
    \prb{ \left(R^{t-1}_{(1)} > R^{t-1}_{(2)} \right) \wedge \left( R^{t-1}_{(1)} - R^{t-1}_{(2)} \leq 1 - \rt{(1)} \right) \mid \rt{(1)} \leq \frac{1}{2} }
    \\& = \prb{R^{t-1}_{(1)} - R^{t-1}_{(2)} \leq 1 - \rt{(1)} \mid R^{t-1}_{(1)} > R^{t-1}_{(2)}, \rt{(1)} \leq \frac{1}{2} }\cdot \prb{R^{t-1}_{(1)} > R^{t-1}_{(2)} \mid \rt{(1)} \leq \frac{1}{2}}.
\end{align*}
Notice that given $\rt{(1)} \leq \frac{1}{2}$, the random variable $1- \rt{(1)}$ is $\uni{\frac{1}{2}, 1}$ distributed. Similar arguments holds for $\rt{(1)} \mid \rt{(1)} > \frac{1}{2}$; thus, 
\begin{align*}
&\prb{R^{t-1}_{(1)} - R^{t-1}_{(2)} \leq 1 - \rt{(1)} \mid R^{t-1}_{(1)} > R^{t-1}_{(2)}, \rt{(1)} \leq \frac{1}{2} }
\\&=
\prb{R^{t-1}_{(1)} - R^{t-1}_{(2)} \leq \rt{(1)} \mid R^{t-1}_{(1)} > R^{t-1}_{(2)}, \rt{(1)} > \frac{1}{2} }.
\end{align*}
Simplifying the above,
\begin{align}\label{B2}
    & \prb{B_2 \mid \rt{(1)} \leq \frac{1}{2} }
    \nonumber\\&
    = \prb{R^{t-1}_{(1)} - R^{t-1}_{(2)} \leq \rt{(1)} \mid R^{t-1}_{(1)} > R^{t-1}_{(2)}, \rt{(1)} > \frac{1}{2} }\cdot \prb{R^{t-1}_{(1)} > R^{t-1}_{(2)} \mid \rt{(1)} \leq \frac{1}{2}}\nonumber \\
    & = \prb{R^{t-1}_{(1)} - R^{t-1}_{(2)} \leq \rt{(1)} \mid R^{t-1}_{(1)} > R^{t-1}_{(2)}, \rt{(1)} > \frac{1}{2} }\cdot \prb{R^{t-1}_{(1)} > R^{t-1}_{(2)}}
    .
\end{align}

As for event $B_3$, recall that the arrival order is uniform. As a result, $R^{t-1}_{(1)}, R^{t-1}_{(2)}$ are independent in $\rt{(1)}$. Leveraging this fact,
\begin{align}\label{B3}
    & \prb{B_3 \mid \rt{(1)} \leq \frac{1}{2} } =
    \prb{\left( R^{t-1}_{(2)} > R^{t-1}_{(1)} \right) \wedge  \left( R^{t-1}_{(2)} - R^{t-1}_{(1)} \leq \rt{(1)} \right) \mid \rt{(1)}> \frac{1}{2}}\nonumber \\
    & =  \prb{\left( R^{t-1}_{(1)} > R^{t-1}_{(2)} \right) \wedge  \left( R^{t-1}_{(1)} - R^{t-1}_{(2)} \leq \rt{(1)} \right) \mid \rt{(1)}> \frac{1}{2}}\nonumber \\
    & = \prb{R^{t-1}_{(1)} - R^{t-1}_{(2)} \leq \rt{(1)} \mid R^{t-1}_{(1)} > R^{t-1}_{(2)}, \rt{(1)} \leq \frac{1}{2} }\cdot \prb{R^{t-1}_{(1)} > R^{t-1}_{(2)} \mid \rt{(1)} \leq \frac{1}{2}}\nonumber\\
    & = \prb{R^{t-1}_{(1)} - R^{t-1}_{(2)} \leq \rt{(1)} \mid R^{t-1}_{(1)} > R^{t-1}_{(2)}, \rt{(1)} \leq \frac{1}{2} }\cdot \prb{R^{t-1}_{(1)} > R^{t-1}_{(2)}},
\end{align}
where the last two equalities hold from the same arguments as in the analysis of event $B_2$. Combining Equalities \eqref{B2} and \eqref{B3}, we get
\begin{align*}
 \prb{B_2 \mid \rt{(1)} \leq \frac{1}{2} } & + \prb{B_3 \mid \rt{(1)} \leq \frac{1}{2} }\\ &= \prb{R^{t-1}_{(1)} > R^{t-1}_{(2)}} \cdot\left(
 \prb{R^{t-1}_{(1)} - R^{t-1}_{(2)}   \leq \rt{(1)} \mid R^{t-1}_{(1)} > R^{t-1}_{(2)}, \rt{(1)} > \frac{1}{2} } \right. \\& + \left.
  \prb{R^{t-1}_{(1)} - R^{t-1}_{(2)} \leq \rt{(1)} \mid R^{t-1}_{(1)} > R^{t-1}_{(2)}, \rt{(1)} \leq \frac{1}{2} }\right) \\ & = 
  \frac{1}{2} \left(1 - \frac{1}{2^{t-1}} \right) \cdot\left(
 \prb{R^{t-1}_{(1)} - R^{t-1}_{(2)}   \leq \rt{(1)} \mid R^{t-1}_{(1)} > R^{t-1}_{(2)}, \rt{(1)} > \frac{1}{2} } \right. \\& + \left.
  \prb{R^{t-1}_{(1)} - R^{t-1}_{(2)} \leq \rt{(1)} \mid R^{t-1}_{(1)} > R^{t-1}_{(2)}, \rt{(1)} \leq \frac{1}{2} }\right),
\end{align*}
where the last equality is due to Equation~\eqref{B1}. Next, notice that 
\begin{align*}
\frac{1}{2} \left(1 - \frac{1}{2^{t-1}} \right) &\left(
\prb{R^{t-1}_{(1)} - R^{t-1}_{(2)}   \leq \rt{(1)} \mid R^{t-1}_{(1)} > R^{t-1}_{(2)}, \rt{(1)} > \frac{1}{2} } \right. \\+ & \left.
\prb{R^{t-1}_{(1)} - R^{t-1}_{(2)} \leq \rt{(1)} \mid R^{t-1}_{(1)} > R^{t-1}_{(2)}, \rt{(1)} \leq \frac{1}{2} }\right)\\=
\frac{1}{2} \left(1 - \frac{1}{2^{t-1}} \right) &\left(
\prb{R^{t-1}_{(1)} - R^{t-1}_{(2)}   \leq \rt{(1)} \mid R^{t-1}_{(1)} > R^{t-1}_{(2)}, \rt{(1)} > \frac{1}{2} } \frac{\Pr\left(r_{(1)}^t > \frac{1}{2}\Big| R_1^{t-1} > R_2^{t-1} \right)}{\Pr\left(r_{(1)}^t > \frac{1}{2}\Big| R_1^{t-1} > R_2^{t-1} \right)} \right. \\+ & \left.
\prb{R^{t-1}_{(1)} - R^{t-1}_{(2)} \leq \rt{(1)} \mid R^{t-1}_{(1)} > R^{t-1}_{(2)}, \rt{(1)} \leq \frac{1}{2} } \frac{\Pr\left(r_{(1)}^t\leq \frac{1}{2}\Big| R_1^{t-1} > R_2^{t-1} \right)}{\Pr\left(r_{(1)}^t\leq \frac{1}{2}\Big| R_1^{t-1} > R_2^{t-1} \right)} \right)
\\ = \frac{1}{2}  \left(1 - \frac{1}{2^{t-1}} \right) &\cdot \frac{1}{\frac{1}{2}} \cdot 
\prb{R^{t-1}_{(1)} - R^{t-1}_{(2)} \leq \rt{(1)} \mid R^{t-1}_{(1)} > R^{t-1}_{(2)} },
\end{align*}
where the last equation is based on the law of total probability and the fact that 
\[
\prb{\rt{(1)} > \frac{1}{2} \mid R^{t-1}_{(1)} > R^{t-1}_{(2)}} = \prb{\rt{(1)} \leq \frac{1}{2} \mid R^{t-1}_{(1)} > R^{t-1}_{(2)}} = \frac{1}{2}.
\]
Combining the latter with Equation~\eqref{B1}, we have
\begin{align*}
\prb{a^t_{(2)} \neq a^t_{(1)} \mid \rt{(1)} \leq \frac{1}{2}} & = \left(\frac{1}{2}\right)^{t-1}+ \left(1 - \frac{1}{2^{t-1}} \right)  
\prb{R^{t-1}_{(1)} - R^{t-1}_{(2)} \leq \rt{(1)} \mid R^{t-1}_{(1)} > R^{t-1}_{(2)} }\\
& = \left(\frac{1}{2}\right)^{t-1}+ \left(1 - \frac{1}{2^{t-1}} \right)  
\prb{\abs{R^{t-1}_{(1)} - R^{t-1}_{(2)}} \leq \rt{(1)} \mid R^{t-1}_{(1)} > R^{t-1}_{(2)} }\\
& = \left(\frac{1}{2}\right)^{t-1}+ \left(1 - \frac{1}{2^{t-1}} \right) 
\prb{\env^{t-1} \leq \rt{(1)} }\geq \prb{\env^{t-1} \leq \rt{(1)} }.
\end{align*}
To finish the proof we use Proposition~\ref{prop:ef1 uni dominance}, which implies that
\begin{align*}
    \prb{\env^{t-1} \leq \rt{(1)} } & = \int_{0}^{1} \prb{\env^{t-1} \leq u}\cdot f_{\uni{0,1}} \,du
    \geq \int_{0}^{1} u\cdot f_{\uni{0,1}} \,du
    \\&=\E{\uni{0,1}} = \frac{1}{2};
\end{align*}
thus,
\begin{align*}
\prb{a^t_{(2)} \neq a^t_{(1)} \mid \rt{(1)} \leq \frac{1}{2}} \geq 
\prb{\env^{t-1} \leq \rt{(1)} } \geq
\frac{1}{2}.
\end{align*}
This concludes the proof of Proposition~\ref{prop:ef1 open arm}.
\end{proof}

%% file: sections/appendix-average.tex
\section{Average Envy}
\label{sec: avg envy}
In this section, we examine another way to define envy: The average reward disparity between the agents. We define the \emph{average envy}, denoted $\envavg^T$, as 
\[
\envavg^T = \frac{1}{\binom{N}{2}}\sum_{1\leq i<j\leq N}{\abs{\env_{i,j}^T}}.
\]
For the special cases of $N=2$ agents, the definition of maximal envy $\env$ and average envy $\envavg$ coincide.

Since $\envavg^t \leq \env^t$ for all $t$  almost surely, any upper bound on the maximal envy can be applied to the average envy. Particularly, Theorem~\ref{thm: uni upper-bound} provides an immediate upper bound on $\E{\envavg^T}$ of $O\left(\sqrt{\ln (N) \sum^T_{t=1} \var{\dift}} \right)$. Using a slightly more careful analysis, we can eliminate the $\sqrt{\ln{(N)}}$ factor.
\begin{proposition}\label{prop: avg upper-bound}
When executing any algorithm, it holds that
\[\E{\max_{1\leq t \leq T} \envavg^t (\uniord)} \leq 2\sqrt{\ln{(N)} \sum^{T}_{t=1}{\var{\dift}} }.\]
\end{proposition}
\begin{proof}[Proof of Proposition~\ref{prop: avg upper-bound}]
Much like the proof of Theorem~\ref{thm: uni upper-bound}, this proof is based on the fact that $\dift$ are subgaussian random variables. From the linearity of expectation, we get
\begin{align}\label{avg 1}
\E{\max_{1\leq t \leq T} \envavg^t } &=
\E{\max_{1\leq t \leq T} \frac{1}{\binom{N}{2}}\sum_{1\leq i<j\leq N}{\abs{\env_{i,j}^t}}} 
\leq \E{\frac{1}{\binom{N}{2}}\sum_{1\leq i<j\leq N}{\max_{1\leq t \leq T} 
 \abs{\env_{i,j}^t}}} \nonumber\\
&=\frac{1}{\binom{N}{2}}\sum_{1\leq i<j\leq N}\E{\max_{1\leq t \leq T}  \abs{\sum^t_{\tau=1}{\adif{i}{j}^\tau}} }.
\end{align}
For every pair of agents $i,j \in [N]$, it holds that
\begin{align}\label{avg 2}
\E{\max_{1\leq t \leq T}  \abs{\sum^t_{\tau=1}{\adif{i}{j}^\tau}} } =
\E{\max_{1\leq t \leq T, \sigma\in \{-1,1\}}\left\{\sigma\sum^t_{\tau=1}{\adif{i}{j}^\tau}\right\}} \leq
\sqrt{2\ln{(2)}\sum^T_{t=1}{\var{\dift}}},
\end{align}
where the last inequality is due to Proposition~\ref{prop:envy is good SG} and Claim~\ref{claim: sg max}. By combining Inequalities~\eqref{avg 1} and \eqref{avg 2} we can conclude that
\begin{align*}
\E{\max_{1\leq t \leq T} \envavg^t } \leq
\frac{1}{\binom{N}{2}}\sum_{1\leq i<j\leq N}{\sqrt{2\ln{(2)}\sum^T_{t=1}{\var{\dift}}}} =
\sqrt{2\ln{(2)}\sum^T_{t=1}{\var{\dift}}},
\end{align*}
which concludes the proof of Proposition~\ref{prop: avg upper-bound}. 
\end{proof}
Next, we craft a lower bound for the average envy.
\begin{proposition}\label{prop: mp avg lower-bound}
    For a large enough $T$, a sufficiently random execution with a symmetric, memory-free algorithm yields
    \[\E{\envavg^T} \geq c\sqrt{ \sum^{T}_{t=1}{\var{\dif^t}}},\]
    where $c> 0$ is a global constant.    
\end{proposition}

\begin{proof}[Proof of Proposition~\ref{prop: mp avg lower-bound}]
The proof of Proposition~\ref{prop: mp avg lower-bound} is almost identical to the proof of Theorem~\ref{thm: uni lower-bound}. In that proof, we bounded the (maximal) envy from below using the envy between a specific couple of agents.
Since, the algorithm is symmetric, the bound we showed is valid for every two agents;
thus, for every $i,j$ it holds that
\begin{align*}
\E{\env_{i,j}^T} \geq \frac{A_1}{2}\sqrt{ \sum^{T}_{t=1}{\var{\dif^t}} },
\end{align*}
where $A_1$ is the constant from Theorem~\ref{thm:BDG}. Using linearity of expectation, we get
\begin{align*}
\E{\envavg^T} & =
\E{\frac{1}{\binom{N}{2}}\sum_{i,j \in [N]^2}{\env_{i,j}^T}} =
\frac{1}{\binom{N}{2}}\sum_{i,j \in [N]^2}{\E{\env_{i,j}^T}}
\geq
\frac{1}{\binom{N}{2}}\sum_{i,j \in [N]^2} {\frac{A_1}{2}\sqrt{ \sum^{T}_{t=1}{\var{\dif^t}} }}
\\&=
\frac{A_1}{2}\sqrt{ \sum^{T}_{t=1}{\var{\dif^t}} }.
\end{align*}
This concludes the proof of Proposition~\ref{prop: mp avg lower-bound}.
\end{proof}
We finalize this section by mentioning that the upper bound on nudged arrival and maximal envy holds trivially for the average envy due to the fact that $\envavg^t \leq \env^t$ for all $t$  almost surely. Future work could seek a tighter bound for the average envy.

%% file: sections/appendix-socially-opt.tex
\section{Socially Optimal Algorithms}\label{appendix:sociallyopt}
In this section, we consider the task of devising socially optimal algorithms. First, in Subsection~\ref{subsec:sw N=2}, we address the two-agent case. Then, in Subsection~\ref{subsec:sopt for N>2}, we develop algorithms for the $N>2$ case.

To ease readability, we make the assumption that reward distributions are stationary, i.e., $\mathcal{D}^1_i, \mathcal{D}^2_i, \dots \mathcal{D}^T_i$ are identical for every arm $a_i \in A$. Consequently,  $X^1_i, X^2_i, \dots X^T_i$ are i.i.d. and we drop the superscript. We stress that our results can also be easily extended to the non-stationary case. Furthermore, we let $\mu_i = \mathbb{E}[X_i^t]$ denote the expected reward of arm $a_i$.

\subsection{Social Welfare for $N=2$}
\label{subsec:sw N=2}

\begin{algorithm}[t]
\caption{Two-agents Socially Optimal Algorithm ($\sopt$)}
\label{alg: sopt}
\LinesNumbered
\DontPrintSemicolon 
\KwIn{horizon $T$, reward distributions $\mathcal{D}_1, \ldots, \mathcal{D}_K$}
Compute $(\is, \js)$ such that
\label{alg: sopt compute}
\begin{align}\label{eq: picking i,j star}
(\is, \js) \in \argmax_{(i,j) \in A^2} \left\{ \mu_i  + \prb{X_i < \mu_j}\mu_j+\prb{X_i \geq \mu_j}\E{X_i \mid X_i \geq \mu_j}  \right\}.
\end{align}\\
\For{round $t = 1$ to $T$}{\label{alg: sopt for}
    Pull $a_{\is}$ \label{alg: sopt 3}\\
    \lIf{$x^t_\is \geq \mu_\js$}{
        Pull $a_{\is}$ \label{alg: sopt if}
    }
    \lElse{
        Pull $a_{\js}$ \label{alg: sopt else}
    }
}
\end{algorithm}
In this section, we design $\sopt$, a socially optimal algorithm for the two-agent case, which we implement in Algorithm~\ref{alg: sopt}. $\sopt$ has Bayesian information, as it receives the reward distributions as inputs. In Line~\ref{alg: sopt compute}, it selects two arms $a_{\is}, a_{\js}$ according to Equation~\eqref{eq: picking i,j star}. As we prove formally, this selection maximizes $\E{{r^t_{(1)}}+{r^t_{(2)}}}$ for any $t$. Arms $a_{\is}, a_{\js}$ are the only arms the algorithm pulls during its execution.

In Line~\ref{alg: sopt for}, $\sopt$ begins interacting with the agents for $T$ rounds.
Notice $\sopt$ does not address the arrival of the agents at all: While the arrival function is crucial for measuring envy, it does not influence the SW.

$\sopt$ pulls arm $a_{\is}$ for the agent that arrives in the first session, 
and observes the realized reward $x^t_{\is}$ (Line~\ref{alg: sopt 3}). In Lines~\ref{alg: sopt if}--\ref{alg: sopt else}, $\sopt$ decides whether to pull the same arm in the second session or $a_{\js}$ instead, based on the realized $x^t_{\is}$:
If $x^t_{\is} \geq \mu_{\js}$ (Line~\ref{alg: sopt if}), i.e., we expect the reward of $a_{\js}$ to be less than or equal to the observed reward, $\sopt$ pulls $a_{\is}$ again.
Otherwise (Line~\ref{alg: sopt else}), we expect the reward of $a_{\js}$  to be greater than that of $a_{\is}$, so the algorithm pulls arm $a_{\js}$. Next, we prove the optimality of $\sopt$.

Before we prove the optimality of $\sopt$, present two propositions that assist in understanding its crux.
\begin{proposition}\label{prop:is}
$\sopt$ does not always pull the arm with the highest expected reward in the first session. That is, $\is$ is not necessarily $\argmax_{i\in \{1, \ldots, K \}}{\mu_i}$.
\end{proposition}

\begin{proof}[Proof of Proposition~\ref{prop:is}]
We show an example satisfying $\is \neq \argmax_{i\in \{1, \ldots, K \}}{\mu_i}$.
Consider $K=2$ arms, with the following distributions.
\begin{align*}
    X_1 \sim
    \begin{cases}
    0.75 & w.p. \ \  0.5 \\
    0.55 & w.p.\  \ 0.5
    \end{cases},
    X_2 \sim
    \begin{cases}
    1 & w.p. \ \  0.6 \\
    0 & w.p.\  \ 0.4
    \end{cases}.
\end{align*}
Observe that $\mu_1 = 0.65 > \mu_2=0.6$. Yet,
\begin{align*}
\mu_1  +\mu_2 \prb{X_1 < \mu_2}+\E{X_1 \mid X_1 \geq \mu_2}  \prb{X_1 \geq \mu_2}=
0.65  + 0.6 \cdot 0.5+ 0.75 \cdot 0.5 = 1.325,
\end{align*}
whereas,
\begin{align*}
\mu_2  +\mu_1 \prb{X_2 < \mu_1}+\E{X_2 \mid X_2 \geq \mu_1}  \prb{X_2 \geq \mu_1}=
0.6  + 0.65 \cdot 0.4 + 1 \cdot 0.6 = 1.46.
\end{align*}
Consequently, $\sopt$ chooses $(\is,\js) = (2,1)$.
This concludes the proof of Proposition~\ref{prop:is}.
\end{proof}
To get the intuition behind Proposition~\ref{prop:is}, recall that with the help of the information exposed by the first agent, the algorithm can make a better decision in the second session.
Therefore, we must find the perfect trade-off between the first agent's welfare (exploitation) and the leverage of the information they provide (exploration).
In contrast, we expect nothing but pure exploitation in the second session.
\begin{proposition}\label{prop:js}
Arm $a_\js$ has the highest expected reward among the remaining arms. I.e., $\js = \argmax_{j\in K\setminus\{\is\}}{\mu_j}$.
\end{proposition}
In other words, in the second session, the algorithm makes the choice that is the most rewarding.
\begin{proof}[Proof of Proposition~\ref{prop:js}]
Let $(\is, \js)$ be a pair of arms that maximizes Equation~\eqref{eq: picking i,j star}.
Suppose, for the sake of contradiction, that there exists ${j'}\notin \{\is,\js\}$ such that $\mu_{j'} > \mu_{\js}$;
thus, $\E{\max{\{X_{\is}, \mu_{{j'}}\}}} \geq \E{\max{\{X_{\is}, \mu_{\js}\}}}$.
I.e.,
\begin{align*}
    \E{X_{\is}} + \E{\max{\{X_{\is}, \mu_{{j'}}\}}} \geq \E{X_{\is}} + \E{\max{\{X_{\is}, \mu_{\js}\}}},
\end{align*}
where equality can occur if and only if $\prb{\max{\{X_{\is}, \mu_{\js}\}} = \mu_{\js}} = 0$.
In this case, the algorithm always pulls arm $a_{\is}$ in the second session, and arm $a_{\js}$ is irrelevant.
Hence, we assume strong inequality.
Notice that the left-hand side is exactly Equation~\eqref{eq: picking i,j star} with $(i,j) = (\is, {j'})$, and the right-hand side is exactly Equation~\eqref{eq: picking i,j star} with $(i,j) = (\is, \js)$.
Hence, we have obtained a contradiction to $(\is, \js)$ being a pair that maximizes Equation~\eqref{eq: picking i,j star}. This concludes the proof of Proposition~\ref{prop:js}.
\end{proof}
We are ready to prove the optimality of $\sopt$.
\begin{theorem}\label{thm:sopt is opt}
Fix any instance with $N=2$ and arbitrary reward distributions. For any algorithm $ALG$ with Bayesian information, it holds that
\[\sw\left(ALG\right) \leq \sw\left(\sopt\right).\]    
\end{theorem}
\begin{proof}[Proof of Theorem~\ref{thm:sopt is opt}]
Fix any instance with $N=2$ and arbitrary reward distributions. The social welfare of an algorithm $ALG$ over $T$ rounds is
\[
\sw(ALG) 
=\E{ \sum^2_{i=1}{R_i^T}}= \E{\sum_{t=1}^T \bigl(r^t_{(1)} + r^t_{(2)}\bigr)}
= T \, \E{r^1_{(1)} + r^1_{(2)}},
\]
where the last equality uses the fact that the rewards in each round are independent.

Thus, to show that $\sopt$ maximizes social welfare, it suffices to show that no algorithm can exceed its expected reward \emph{within a single round}. By the revelation principle~\cite{peters2001common}, there is an optimal \emph{threshold} algorithm: After observing the first-session reward, it decides in the second session by comparing the observed reward against the expected reward of any other arm.

Since $\sopt$ enumerates all pairs of arms $(i,j)$ in Equation~\eqref{eq: picking i,j star} and applies a greedy rule for the second session (choosing either the same arm $i$ or the other arm $j$ based on which is expected to yield a higher reward), it achieves the maximum expected reward per round. Multiplying by $T$ completes the proof.
\end{proof}

\subsection{Socially Optimal Algorithm for $N >2$}\label{subsec:sopt for N>2}
Next, we consider the problem of finding a socially optimal algorithm for $N >2$ agents. First, we present a socially optimal algorithm for the special case of Bernoulli rewards.

\begin{proposition}\label{prop:bernoulli sw optimal}
    Assume that $X_i \sim \ber{p_i}$ for every $i \in [K]$. An algorithm that pulls arms in descending order of $p_i$ until it realizes a reward of 1 is socially optimal for any number of agents $N$.
\end{proposition}
\begin{proof}[Proof of Proposition~\ref{prop:bernoulli sw optimal}]
    Since random algorithms are just a distribution over deterministic algorithms, we know there is an optimal algorithm that is deterministic. Additionally, as in the proof of Theorem~\ref{thm:sopt is opt}, it suffices to focus on a single and arbitrary round $t$.
    
    Notice that after we observe an arm such that $x^t_i =1$, it is optimal to pull it for all remaining agents.
    Similarly,  if we observe arm $a_i$ with $x^t_i =0$ it is strictly sub-optimal to pull it again, unless all arms yielded a reward of 0.
    Thus, the only thing left to prove is the optimality of the order in which the algorithm pulls arms.

    In this special case, we can use a reduction to a Pandora's Box (PB) problem~\cite{weitzman1978optimal}. We first describe the reduction, then characterize the optimal solution for the PB instance, and finally show the equivalence. 
       
    Let the PB instance include $K$ Bernoulli arms with success probability $p_i$ for every arm $i$ and costs $c_i = \frac{1}{B}$, where $B$ is a constant such that $B > \max_{i \in [K]} \frac{1}{p_i}$. Due to Weitzman's seminal result~\cite{weitzman1978optimal}, there exist indices $(\theta_i)_{i=1}^K$ such that the optimal sequence is descending in the index. Each index  $\theta_i$ is the solution to $\E{\max\left\{X_i - \theta_i,0 \right\}}  =c_i =\frac{1}{B}$. Thus, $\theta_i = 1- \frac{1}{B\cdot p_i}$. The optimal solution maximizes $\E{1-\frac{S}{B}}$, where $S$ is a r.v. that counts the number of useless arms pulled (arms with a realized reward of 0). Note that $S$ depends solely on the pulling order.

    Similarly, for our original problem, maximizing the social welfare amounts to maximizing $\E{N-S}$. Since $S$ is distributed identically in both problems, the optimal order for PB is optimal for the original problem as well.
\end{proof}


Next, we move beyond Bernoulli rewards. Fix an arbitrary round $t$, and assume that all rewards are supported in a finite set $V$ (we later explain how to relax this assumption). We now describe a dynamic programming procedure that finds the optimal algorithm with a computational complexity of $O(\abs{V} \cdot N \cdot 2^K)$. 

We consider the following parameters: $B \subseteq A$, representing a subset of available arms; $n \in \{0,1,\dots,N\}$, denoting the number of remaining agents; and $v \in V$, an arbitrary current reward that models the maximal reward of all observed arms in $A \setminus B$.

We define the function $f(n, B, v)$, representing the maximum expected social welfare achievable given the parameters $(n,B,v)$. To apply this dynamic programming approach, we first establish the base cases for the function $f$:
\begin{itemize}
    \item \textbf{No agents remaining ($n = 0$):} If there are no agents left to assign rewards, the maximum achievable reward is zero regardless of the subset of arms $B$ and the current reward $v$. Formally, 
    \[
    f(0, B, v) = 0 \quad \forall \, B \subseteq A, \, v \in V.
    \]
    
    \item \textbf{No available arms ($B = \emptyset$):} If there are no arms left to pull, the only option is to assign the current reward $v$ to all remaining agents. Thus, the maximum reward in this scenario is the product of $v$ and the number of remaining agents $n$. Formally, 
    \[
    f(n, \emptyset, v) = v \cdot n \quad \forall \, n \in \{0,1,\dots,N\}, \, v \in V.
    \]
\end{itemize}

We move on to the recursive definition of $f$:
\begin{equation}\label{eq:dp_f}
f(n, B, v) = \max \left\{ v \cdot n, \max_{a \in B} \mathbb{E}\left[ X_a^t + f\left(n - 1, B \setminus \{a\}, \max\{v, X_a^t\} \right) \right] \right\}
\end{equation}
The recursive relation in Equation~\eqref{eq:dp_f} considers two scenarios at each step:
\begin{enumerate}
    \item \textbf{Assigning the Current Reward ($v \cdot n$):} In this scenario, the algorithm assigns the current reward $v$ to all remaining $n$ agents without pulling any additional arms.
    
    \item \textbf{Pulling an Arm ($\max_{a \in B} \mathbb{E}\left[ X_a^t + f\left(n - 1, B \setminus \{a\}, \max\{v, X_a^t\} \right) \right]$):} Here, the algorithm selects an arm $a \in B$, observes the stochastic reward $X_a^t$, and then recursively assigns rewards to the remaining $n - 1$ agents. The subset of available arms is updated to $B \setminus \{a\}$, and the current reward is updated to $\max\{v, X_a^t\}$.
\end{enumerate}

The optimal value for round $t$ is obtained by evaluating the function $f$ at the initial conditions where all agents and arms are available, and the starting reward is zero:
\[
f(n = N, B = A, v = 0).
\]
This value represents the maximum expected social welfare achievable by the optimal algorithm for round $t$. 

\begin{theorem}\label{thm:optimality_runtime}
    The dynamic programming procedure defined by the function $f(n, B, v)$ in Equation~\eqref{eq:dp_f} correctly computes the maximum expected total social welfare achievable for round $t$. Furthermore, the procedure operates with a runtime complexity of $O(\abs{V} \cdot N\cdot K\cdot 2^K)$.
\end{theorem}

\begin{proof}
    The optimality of the procedure follows from standard dynamic programming arguments, relying on the fact that the overall problem can be constructed from optimal solutions to its subproblems. The function $f(n, B, v)$ is designed to represent the maximum expected total social welfare achievable given the state $(n, B, v)$. By taking the maximum of the two options (recall Equation~\eqref{eq:dp_f}, $f(n, B, v)$ ensures that the optimal decision is made at each state, thereby maximizing the expected total social welfare.

To determine the runtime complexity, we analyze the number of possible states and the computation required for each state. The function $f(n, B, v)$ is parameterized by the number of remaining agents $n$, which can take $N+1$ possible values; the subset of available arms $B$, which has $2^K$ possible subsets; and the current reward $v$, which can take $\abs{V}$ possible values. Therefore, the total number of states is $O(\abs{V} \cdot N \cdot 2^K)$. For each state $(n, B, v)$, the dynamic programming procedure performs a constant-time computation to calculate $v \cdot n$ and iterates over all arms in $B$, performing constant-time computations for each arm. Given that there are up to $K$ arms in $B$, the computation per state is $O(K)$. Consequently, the overall runtime complexity is $O(\abs{V} \cdot N \cdot K \cdot 2^K )$. 
\end{proof}

Finally, we can relax the finite-support assumption. In scenarios where rewards are arbitrary within the continuous interval $[0,1]$, we discretize the reward space by selecting a finite set $V$ that approximates the continuous range of rewards. By choosing an appropriate granularity for the discretization, we can control the trade-off between the accuracy of the approximation and the computational complexity of the algorithm. While a finer discretization yields a closer approximation to the true optimal solution, it simultaneously increases the size of the state space, thereby enhancing the computational burden. Conversely, a coarser discretization reduces computational requirements at the expense of approximation precision. As this approach is standard, we omit the details.

Unfortunately, the procedure above is inefficient in the number of arms $K$. We leave the tasks of finding an optimal algorithm, proving hardness, and finding approximately optimal algorithms as open problems.

%% file: sections/appendix-simulations.tex
\section{Simulations}
\label{sec:simulations}
\newcommand{\uins}{I_U}
\newcommand{\bins}{I_B}
\newcommand{\sins}{I_S}
In this section, we report the results of extensive simulations to empirically validate our theoretical results and test our conjectures. Specifically, we devote Subsection~\ref{subsec:dep t n} to verify the behavior of uniform, nudged, and adversarial arrival as functions of the horizon $T$ and number of agents $N$. In Subsection~\ref{subsec:nudge sensitive}, we provide a sensitivity analysis of nudged arrival. Subsection~\ref{subsec:sim efc} validates our results from Section~\ref{sec:extensions}, as well as test Conjecture~\ref{thm: efc sw}.

\paragraph{Simulation setup}
For analyses where the time horizon $T$ and the number of agents $N$ are not explicitly specified, we set $T = 10{,}000$ rounds and $N=2$ agents by default. We report the average results over $1{,}000$ independent runs, with the shaded areas indicating three standard deviations from the mean. All simulations were conducted on a standard Lenovo laptop, with the total execution time amounting to a couple of hours.

We used two instances in most of the experiments, 
\begin{itemize}
    \item Uniform instance ($\uins$): $K=4$ arms, all with $\uni{0,1}$ reward distributions. The algorithm explores arms until it finds one with a value greater than $\frac{3}{4}$. 
    \item Bernoulli instance ($\bins$): $K=3$ Bernoulli arms with success parameters of $p_1 = 0.6, p_2=0.4, p_3=0.2$. The algorithm explores arms in a descending order of $p_i$, until a reward of 1 is materialized.
\end{itemize}

\subsection{Dependence on $T$ and $N$}\label{subsec:dep t n}
\begin{figure}
    \centering
    \includegraphics[width=\linewidth]{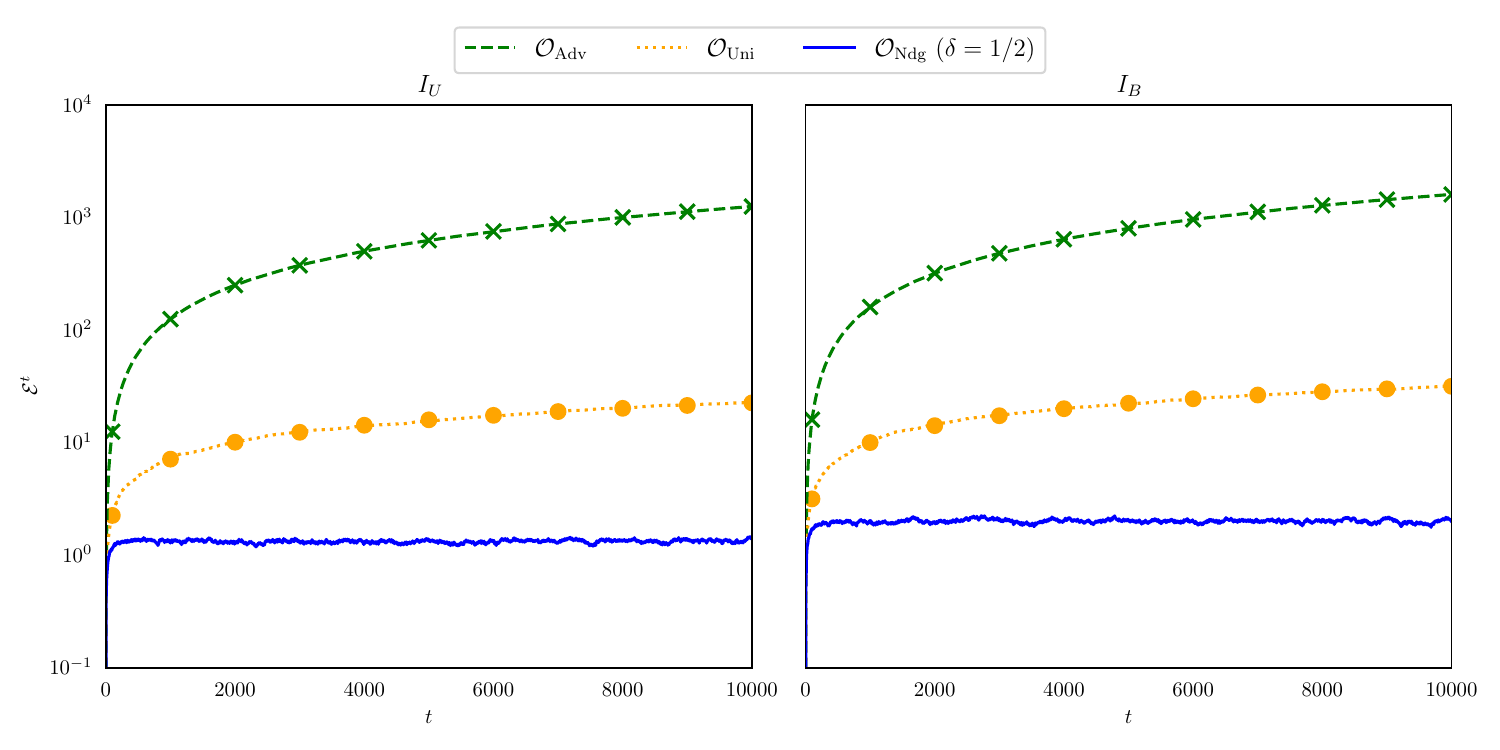}
    \caption{
    $\env^t$ as a function of $t$ for both the uniform instance ($\uins$, left panel) and the Bernoulli instance ($\bins$, right panel), each with $N=2$. 
    The three arrival functions shown are $\uniord$, $\sugord$ (with $\delta=\frac 1 2$), and $\advord$. Green `X' markers represent the maximum likelihood estimates (MLE) for the linear model $y = c \cdot x$, while orange circles indicate the MLE for the square-root model $y = c \cdot \sqrt{x}$. The perfect alignment of the simulated data with these curves confirms our theoretical predictions.}
    \label{fig: envy-func-t}
\end{figure}
In this subsection, we validate our results regarding the dependency on $T$ and $N$. Figure~\ref{fig: envy-func-t} shows the cumulative envy as a function of time for all three arrival functions: $\advord$ (green dashed), $\uniord$ (orange dotted), and $\sugord$ (blue smooth). Each plot shows the cumulative envy over time on a logarithmic vertical axis, reflecting the distinct asymptotic behaviors of the arrival functions. The left panel presents the uniform instance $\uins$, and the right panel the Bernoulli instance $\bins$.  Due to the logarithmic scale, the shaded regions indicating three standard deviations are barely distinguishable; we provide further details in Table~\ref{table} in Subsection~\ref{appendix: simulations}.

For both instances, we see that the cumulative envy of $\advord$ and $\uniord$ increases over time, whereas the cumulative envy of $\sugord$ remains nearly constant (subject to some noise). As time progresses, we observe substantially different growth rates in envy across the three arrival functions, consistent with our theoretical analysis. 

The green `X' markers represent the maximum likelihood estimates (MLE) for the linear function $y = c \cdot x$, which closely match the green (dashed) curve for the envy under $\advord$, thereby confirming the linear growth predicted by Proposition~\ref{thm: adv-envy}. The orange (dotted) curve corresponds to $\uniord$, and the orange circles depict the MLE for the square-root function $y = c \cdot \sqrt{x}$. Their close alignment supports Corollary~\ref{cor: uni-envy}. Finally, Theorem~\ref{thm: sugg-envy} asserts that envy under $\sugord$ remains bounded when both the instance parameters and the number of agents are fixed. This is precisely what we observe in the blue (smooth) curves of both panels in Figure~\ref{fig: envy-func-t}.
\begin{figure}
    \centering
    \includegraphics[width=\linewidth]{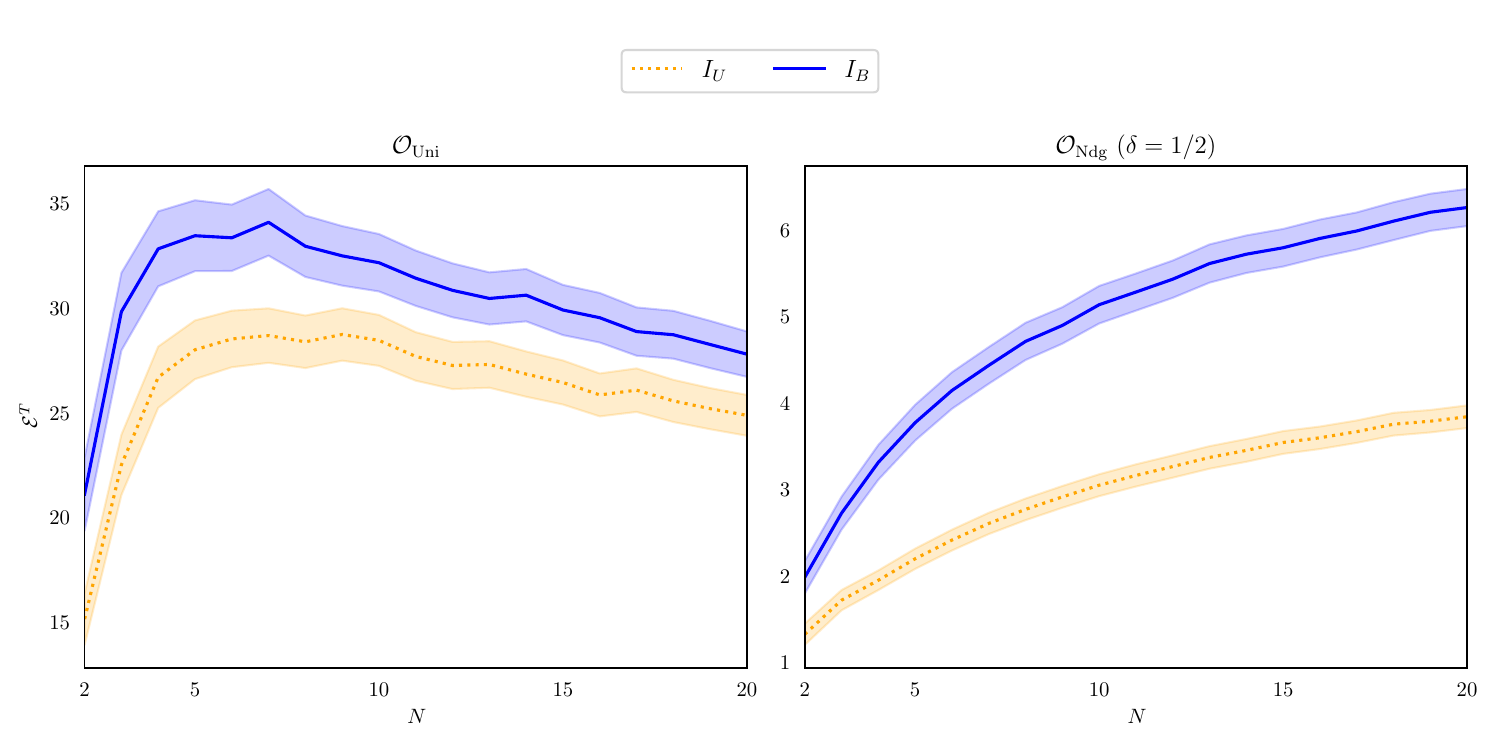}
    \caption{$\env^T$ as a function of $N$ for both $\uins$ and $\bins$, under $\uniord$ (left panel) and $\sugord$ with $\delta=\nicefrac{1}{2}$ (right panel).}
    \label{fig: envy-T-func-N}
\end{figure}

We proceed to examine how envy depends on the number of agents, focusing on $\uniord$ and $\sugord$ with $\delta = \frac{1}{2}$. In Figure~\ref{fig: envy-T-func-N}, we plot the cumulative envy after $T = 10^4$ rounds as a function of the number of agents $N$, with $N$ ranging from $2$ to $20$. The left panel depicts the envy under $\uniord$ for both instances, $\bins$ (blue smooth) and $\uins$ (orange dotted). In each instance, envy initially rises for small $N$ and then declines, matching the intuition from Corollary~\ref{thm: sqrt TK N}: as $N$ grows, there are increasingly more agents exploiting rather than exploring (given that these instances have $K\in\{3,4\}$ arms). The right panel in Figure~\ref{fig: envy-T-func-N} considers nudged arrival $\sugord$. Here, envy increases with the number of agents, as Theorem~\ref{thm: sugg-envy} hints. However, this increase is not linear in $N$, but rather milder. We conjecture that the dependence of $\sugord$ is essentially sub-linear in $N$, leaving a precise characterization for future work.
\subsection{Sensitivity Analysis for Nudged Arrival}\label{subsec:nudge sensitive}
\begin{figure}
    \centering
    \begin{subfigure}{0.49\textwidth}
        \includegraphics[width=\linewidth]{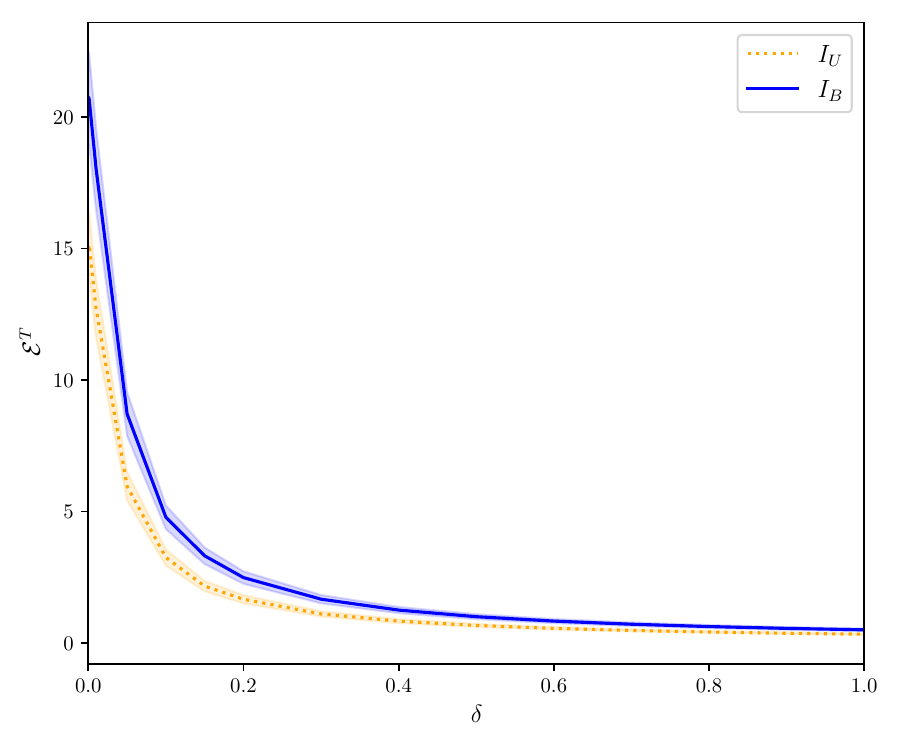}
        \caption{Envy as a function of $\delta$ for $\uins$ and $\bins$.
        }\label{fig: envy_at_T_as_func_delta}
    \end{subfigure}
    \hfill
    \begin{subfigure}{0.49\textwidth}
        \includegraphics[width=\linewidth]{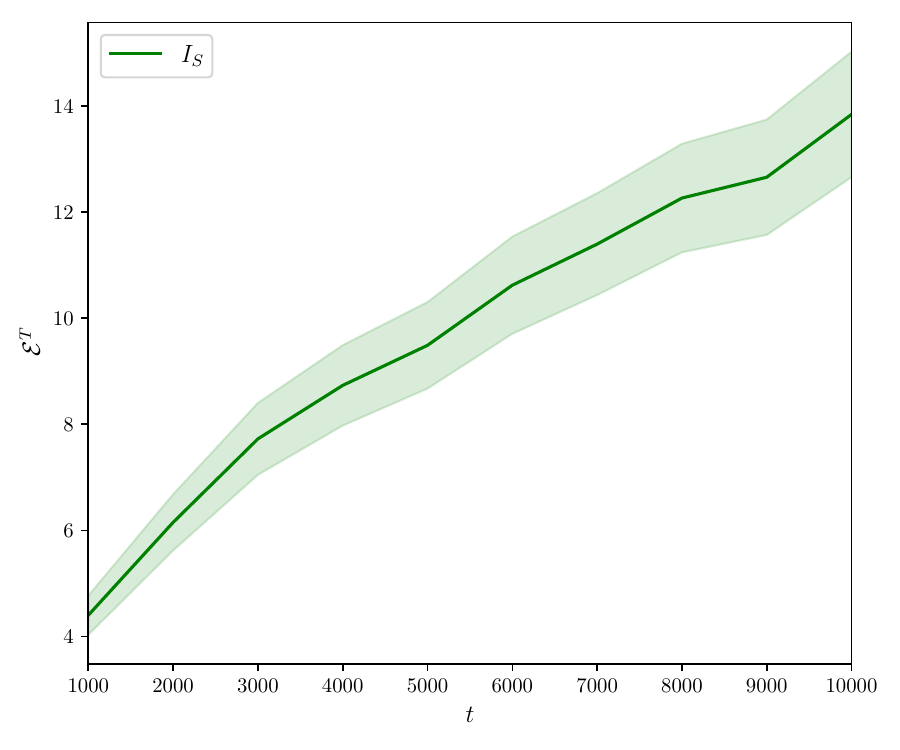}
        \caption{Envy as a function of $\tdif$ for $\sins$.
        }\label{fig: envy_at_T_delta_func_of_T}
    \end{subfigure}
    \caption{Sensitivity Analysis for $\sugord$.}
\end{figure}
In this subsection, we provide a sensitivity analysis for $\sugord$. Figure~\ref{fig: envy_at_T_as_func_delta} shows the cumulative envy as a function of the parameter $\delta$, which reflects how strongly the system can influence the agents' arrival. 
For both $\bins$ (blue smooth) and $\uins$ (orange dotted), we observe that envy decreases as $\delta$ increases. 
Moreover, this reduction appears consistent with $\nicefrac{1}{\delta}$, aligning with the prediction from Theorem~\ref{thm: sugg-envy}.

To examine the dependence of envy on $\tdif$, we introduce reward distributions that explicitly involve $T$, unlike $\uins$ and $\bins$. 
Specifically, we define a new instance, $\sins$, with $K = 2$. 
In this instance,
\begin{align*}
    X_1 \sim
    \begin{cases}
    1 & w.p. \ \  \frac{1}{2} \\
    \frac{1}{4} & w.p.\  \ \frac{1}{2}
    \end{cases},\quad 
    X_2 \sim
    \begin{cases}
    1 & w.p. \ \  \frac{1}{4} + \frac{2}{\sqrt{T}} \\
    0 & w.p.\  \ \frac{3}{4} - \frac{2}{\sqrt{T}}
    \end{cases}.
\end{align*}
Furthermore, we consider the following algorithm: In the first session of every round, it pulls $a_1$. If the observed reward is $1$, it pulls $a_1$ again in the second session; otherwise, it pulls $a_2$. Note that 
\[
\E{\rt{(1)}}=\frac{1+0.25}{2}=\frac{5}{8},
\]
while
\[
\E{\rt{(2)}}= \prb{X_1 = 1}\cdot 1 + \prb{X_1 = 0.25}\cdot \E{X_2} = \frac{1}{2} + \frac{1}{2} \cdot \left( \frac{1}{4} + \frac{2}{\sqrt{T}} \right) = \frac{5}{8} +\frac{1}{\sqrt{T}};
\]
thus, $\tilde{\dif}= \E{\rt{(2)}-\rt{(1)}}= \frac{1}{\sqrt{T}}$.

We now examine how envy evolves in the instance $\sins$. 
Figure~\ref{fig: envy_at_T_delta_func_of_T} displays the cumulative envy $\env^T$ at the final round $T$ for various values of $T$. 
As $T$ grows, $\tilde{\dif}$ decreases, causing envy to increase accordingly. 
Moreover, since $\tfrac{1}{\tilde{\dif}} = \sqrt{T}$, Theorem~\ref{thm: sugg-envy} implies that envy scales proportionally to $\sqrt{T}$, which is consistent with the empirical results.

\subsection{Analysis of EFC}\label{subsec:sim efc}
In this subsection, we examine the theoretical results from Section~\ref{sec:extensions}. Naturally, as our results in that section apply to the special case of Example~\ref{example 1}, all simulations were performed with $N=2$ agents, $K=2$ arms $X_1, X_2 \sim \uni{0,1}$ under uniform arrival and the $\efc$ algorithm. 

\begin{figure}
    \centering
        \includegraphics[scale=0.6]{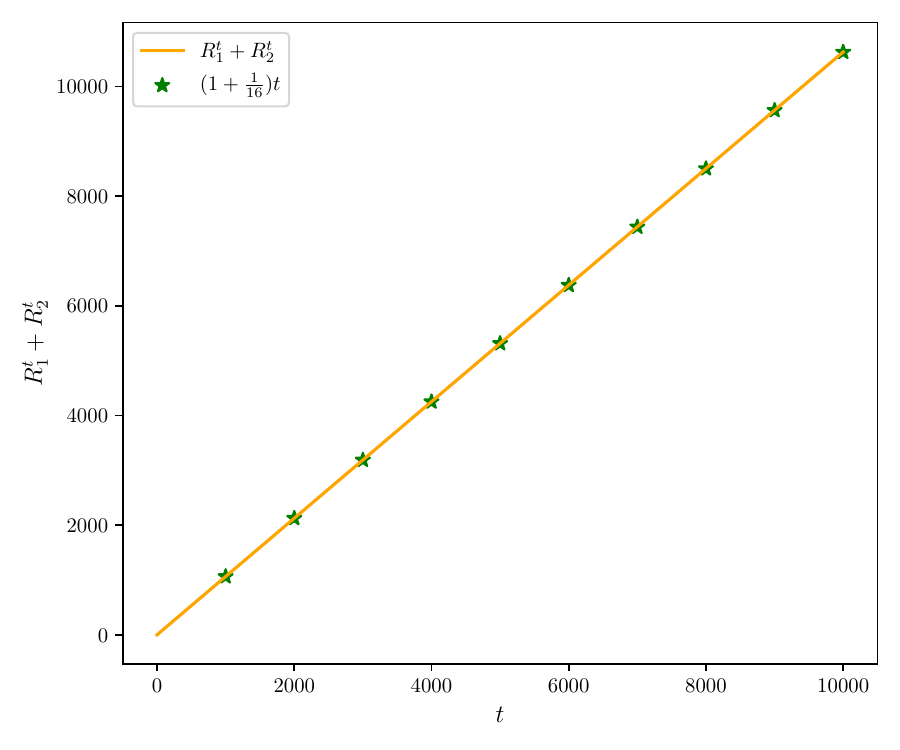}
        \caption{$\Rt{1}+\Rt{2}$ under the $\efc$ algorithm with $C=1$, as a function of $t$.
        }\label{fig: sw ef1}
\end{figure}

\begin{figure}
    \centering
        \includegraphics[width=\linewidth]{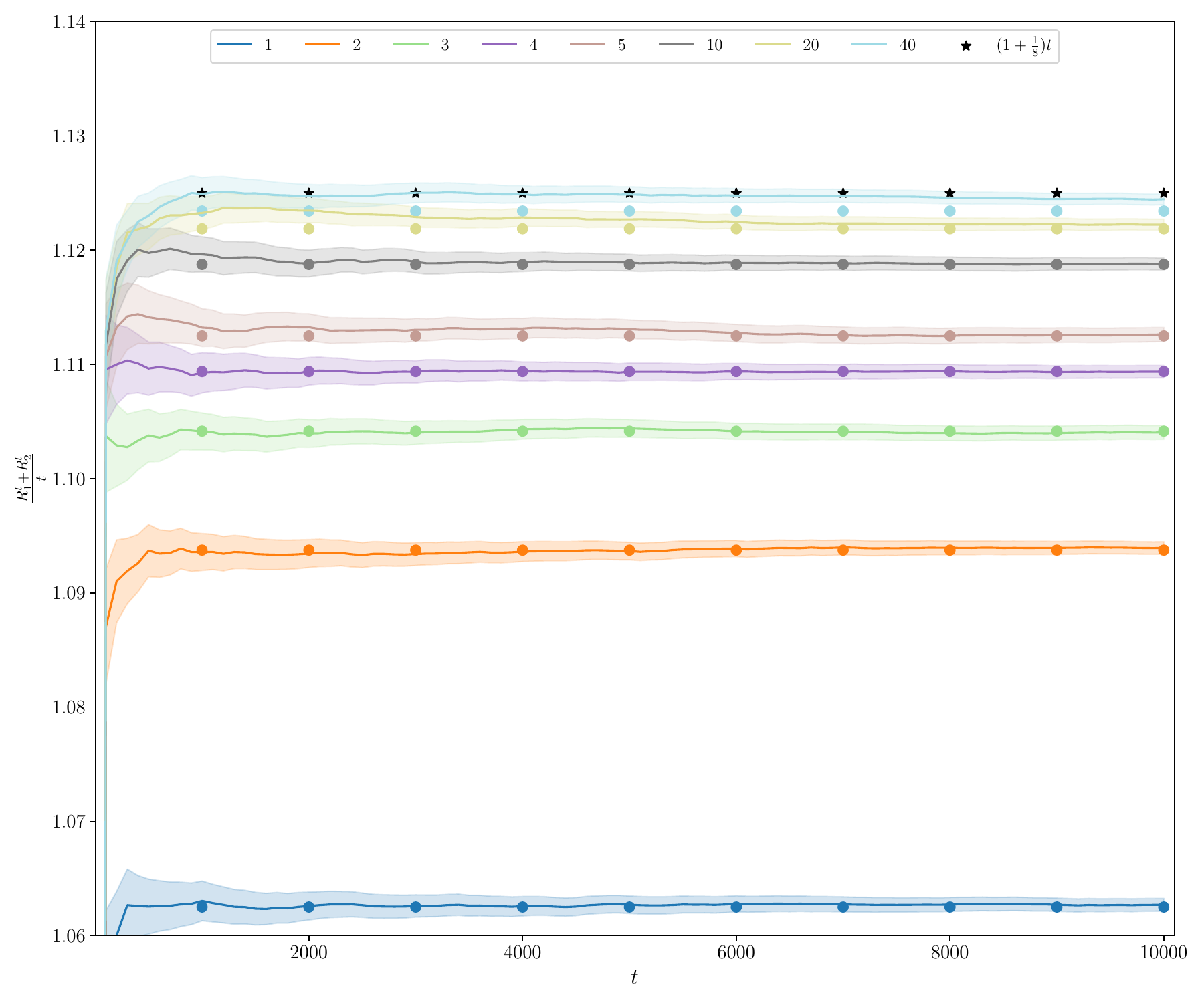}
        \caption{$\frac{\Rt{1}+\Rt{2}}{t}$ under the $\efc$ algorithm, as a function of $t$ compared to $1 + \frac{1}{8}\cdot \frac{2C-1}{2C}$. and $1+\frac{1}{8}$}\label{fig: sw efc}
\end{figure}

Figure~\ref{fig: sw ef1} illustrates the social welfare for $\efc$ with $C=1$, plotting 
$\sw_t := R_1^t + R_2^t$ as the orange (smooth) curve. 
The green `X' markers correspond to the line $y=(1+\frac{1}{16})x$, 
the social welfare guaranteed by Theorem~\ref{thm:ef1evny+sw}; 
their perfect alignment confirms the theoretical prediction.

Figure~\ref{fig: sw efc} evaluates Conjecture~\ref{thm: efc sw}, asserting that for any $C \ge 1$, 
we have $\sw_t \geq (1+\frac{1}{8}-\frac{1}{16C})t$. 
To facilitate comparison, the vertical axis in Figure~\ref{fig: sw efc} depicts the average welfare, $\tfrac{\sw_t}{t}$, 
versus the round number $t$ on the horizontal axis. 
We ran $\efc$ for $C \in \{1,2,3,4,5,10,20,40\}$, adding circle markers in corresponding colors to highlight 
the values of $1 + \tfrac{2C - 1}{2C} \cdot \tfrac{1}{8}$. 
Additionally, the star marker shows the horizontal line $1 + \tfrac{1}{8}$, 
representing the maximum achievable welfare in this setting, as noted in Observation~\ref{obs:opt for tradeoff}. 
For each $C$, the corresponding curve lies above the dotted line, 
consistent with our conjectured lower bound. 
When $1 \leq C \leq 10$, the average welfare nearly coincides with the conjectured bound, 
whereas for $C = 20$ and $C = 40$, the welfare is strictly higher than the markers. 
The reason for the latter is that such high envy states occur rarely when $T = 10^4$, 
so the constraint in Line~\ref{efclin:pull_a1_cond} in $\efc$ is seldom activated in practice.

\subsection{Standard Deviations for Figure~\ref{fig: envy-func-t}}\label{appendix: simulations}
\begin{table}[!ht]
    \centering
    \caption{Three standard deviations for Figure~\ref{fig: envy-func-t}.}\label{table}
    \begin{tabular}{|l|l|l|l|l|l|l|}
    \hline
        $t$ & $\uins, \advord$ & $\uins, \uniord$ & $\uins, \sugord$ & $\bins, \advord$ & $\bins, \uniord$ &  $\bins, \sugord$ \\ \hline
        1000 & 0.79 & 0.54 & 0.13 & 1.12 & 0.72 & 0.18 \\ \hline
        2000 & 1.11 & 0.74 & 0.12 & 1.59 & 1.02 & 0.18 \\ \hline
        3000 & 1.34 & 0.87 & 0.12 & 1.92 & 1.27 & 0.19 \\ \hline
        4000 & 1.52 & 0.99 & 0.12 & 2.22 & 1.38 & 0.19 \\ \hline
        5000 & 1.72 & 1.09 & 0.13 & 2.50 & 1.59 & 0.18 \\ \hline
        6000 & 1.91 & 1.23 & 0.12 & 2.72 & 1.73 & 0.19 \\ \hline
        7000 & 2.11 & 1.34 & 0.12 & 2.91 & 1.89 & 0.19 \\ \hline
        8000 & 2.30 & 1.42 & 0.12 & 3.16 & 2.08 & 0.18 \\ \hline
        9000 & 2.43 & 1.52 & 0.13 & 3.37 & 2.22 & 0.21 \\ \hline
        10000 & 2.58 & 1.59 & 0.13 & 3.55 & 2.35 & 0.18 \\ \hline
    \end{tabular}
\end{table}

%% file: main.bbl
\begin{thebibliography}{66}
\providecommand{\natexlab}[1]{#1}
\providecommand{\url}[1]{\texttt{#1}}
\expandafter\ifx\csname urlstyle\endcsname\relax
  \providecommand{\doi}[1]{doi: #1}\else
  \providecommand{\doi}{doi: \begingroup \urlstyle{rm}\Url}\fi

\bibitem[Agrawal and Goyal(2012)]{agrawal2012analysis}
Shipra Agrawal and Navin Goyal.
\newblock Analysis of thompson sampling for the multi-armed bandit problem.
\newblock In \emph{Conference on learning theory}, pages 39--1. JMLR Workshop and Conference Proceedings, 2012.

\bibitem[Agrawal et~al.(2020)Agrawal, Sethuraman, and Zhang]{agrawal2020optimal}
Shipra Agrawal, Jay Sethuraman, and Xingyu Zhang.
\newblock On optimal ordering in the optimal stopping problem.
\newblock In \emph{Proceedings of the 21st ACM Conference on Economics and Computation}, pages 187--188, 2020.

\bibitem[Amanatidis et~al.(2024)Amanatidis, Fusco, Reiffenh{\"a}user, and Tsikiridis]{amanatidis2024pandora}
Georgios Amanatidis, Federico Fusco, Rebecca Reiffenh{\"a}user, and Artem Tsikiridis.
\newblock Pandora's box problem over time.
\newblock \emph{arXiv preprint arXiv:2407.15261}, 2024.

\bibitem[Auer(2002)]{auer2002finite}
P~Auer.
\newblock Finite-time analysis of the multiarmed bandit problem, 2002.

\bibitem[Aziz and Mackenzie(2016)]{aziz2016discrete}
Haris Aziz and Simon Mackenzie.
\newblock A discrete and bounded envy-free cake cutting protocol for any number of agents.
\newblock In \emph{2016 IEEE 57th Annual Symposium on Foundations of Computer Science (FOCS)}, pages 416--427. IEEE, 2016.

\bibitem[Bahar et~al.(2020)Bahar, Ben-Porat, Leyton-Brown, and Tennenholtz]{bahar2020fiduciary}
Gal Bahar, Omer Ben-Porat, Kevin Leyton-Brown, and Moshe Tennenholtz.
\newblock Fiduciary bandits.
\newblock In \emph{International Conference on Machine Learning}, pages 518--527. PMLR, 2020.

\bibitem[Bargiacchi et~al.(2018)Bargiacchi, Verstraeten, Roijers, Now{\'e}, and Hasselt]{bargiacchi2018learning}
Eugenio Bargiacchi, Timothy Verstraeten, Diederik Roijers, Ann Now{\'e}, and Hado Hasselt.
\newblock Learning to coordinate with coordination graphs in repeated single-stage multi-agent decision problems.
\newblock In \emph{International conference on machine learning}, pages 482--490. PMLR, 2018.

\bibitem[Ben-Porat et~al.(2021)Ben-Porat, Sandomirskiy, and Tennenholtz]{ben2021protecting}
Omer Ben-Porat, Fedor Sandomirskiy, and Moshe Tennenholtz.
\newblock Protecting the protected group: Circumventing harmful fairness.
\newblock In \emph{Proceedings of the AAAI Conference on Artificial Intelligence}, volume 35(6), pages 5176--5184, 2021.

\bibitem[Benade et~al.(2018)Benade, Kazachkov, Procaccia, and Psomas]{benade2018make}
Gerdus Benade, Aleksandr~M Kazachkov, Ariel~D Procaccia, and Christos-Alexandros Psomas.
\newblock How to make envy vanish over time.
\newblock In \emph{Proceedings of the 2018 ACM Conference on Economics and Computation}, pages 593--610, 2018.

\bibitem[Berger et~al.(2023)Berger, Ezra, Feldman, and Fusco]{berger2023pandora}
Ben Berger, Tomer Ezra, Michal Feldman, and Federico Fusco.
\newblock Pandora's problem with combinatorial cost.
\newblock In \emph{Proceedings of the 24th ACM Conference on Economics and Computation}, pages 273--292, 2023.

\bibitem[Bertsimas et~al.(2012)Bertsimas, Farias, and Trichakis]{bertsimas2012efficiency}
Dimitris Bertsimas, Vivek~F Farias, and Nikolaos Trichakis.
\newblock On the efficiency-fairness trade-off.
\newblock \emph{Management Science}, 58\penalty0 (12):\penalty0 2234--2250, 2012.

\bibitem[Besbes et~al.(2014)Besbes, Gur, and Zeevi]{besbes2014stochastic}
Omar Besbes, Yonatan Gur, and Assaf Zeevi.
\newblock Stochastic multi-armed-bandit problem with non-stationary rewards.
\newblock \emph{Advances in neural information processing systems}, 27, 2014.

\bibitem[Boodaghians et~al.(2020)Boodaghians, Fusco, Lazos, and Leonardi]{boodaghians2020pandora}
Shant Boodaghians, Federico Fusco, Philip Lazos, and Stefano Leonardi.
\newblock Pandora's box problem with order constraints.
\newblock In \emph{Proceedings of the 21st ACM Conference on Economics and Computation}, pages 439--458, 2020.

\bibitem[Brams and Taylor(1995)]{Steven95}
Steven~J. Brams and Alan~D. Taylor.
\newblock An envy-free cake division protocol.
\newblock \emph{The American Mathematical Monthly}, 102\penalty0 (1):\penalty0 9--18, 1995.
\newblock ISSN 00029890, 19300972.
\newblock URL \url{http://www.jstor.org/stable/2974850}.

\bibitem[Braverman and Mossel(2007)]{braverman2007noisy}
Mark Braverman and Elchanan Mossel.
\newblock Noisy sorting without resampling.
\newblock \emph{arXiv preprint arXiv:0707.1051}, 2007.

\bibitem[Brustle et~al.(2024)Brustle, Correa, D{\"u}tting, Ezra, Feldman, and Verdugo]{brustle2024competition}
Johannes Brustle, Jos{\'e} Correa, Paul D{\"u}tting, Tomer Ezra, Michal Feldman, and Victor Verdugo.
\newblock The competition complexity of prophet inequalities.
\newblock In \emph{Proceedings of the 25th ACM Conference on Economics and Computation}, pages 807--830, 2024.

\bibitem[Caragiannis et~al.(2019)Caragiannis, Kurokawa, Moulin, Procaccia, Shah, and Wang]{caragiannis2019unreasonable}
Ioannis Caragiannis, David Kurokawa, Herv{\'e} Moulin, Ariel~D Procaccia, Nisarg Shah, and Junxing Wang.
\newblock The unreasonable fairness of maximum {N}ash welfare.
\newblock \emph{ACM Transactions on Economics and Computation (TEAC)}, 7\penalty0 (3):\penalty0 1--32, 2019.

\bibitem[Caton and Haas(2020)]{caton2020fairness}
Simon Caton and Christian Haas.
\newblock Fairness in machine learning: A survey.
\newblock \emph{arXiv preprint arXiv:2010.04053}, 2020.

\bibitem[Chakraborty et~al.(2017)Chakraborty, Chua, Das, and Juba]{chakraborty2017coordinated}
Mithun Chakraborty, Kai Yee~Phoebe Chua, Sanmay Das, and Brendan Juba.
\newblock Coordinated versus decentralized exploration in multi-agent multi-armed bandits.
\newblock In \emph{IJCAI}, pages 164--170, 2017.

\bibitem[Chu et~al.(2011)Chu, Li, Reyzin, and Schapire]{chu2011contextual}
Wei Chu, Lihong Li, Lev Reyzin, and Robert Schapire.
\newblock Contextual bandits with linear payoff functions.
\newblock In \emph{Proceedings of the Fourteenth International Conference on Artificial Intelligence and Statistics}, pages 208--214. JMLR Workshop and Conference Proceedings, 2011.

\bibitem[Cohen et~al.(2021)Cohen, Schmidt-Kraepelin, and Mansour]{cohen2021dueling}
Lee Cohen, Ulrike Schmidt-Kraepelin, and Yishay Mansour.
\newblock Dueling bandits with team comparisons.
\newblock \emph{Advances in Neural Information Processing Systems}, 34:\penalty0 20633--20644, 2021.

\bibitem[Cohler et~al.(2011)Cohler, Lai, Parkes, and Procaccia]{cohler2011optimal}
Yuga Cohler, John Lai, David Parkes, and Ariel Procaccia.
\newblock Optimal envy-free cake cutting.
\newblock In \emph{Proceedings of the AAAI Conference on Artificial Intelligence}, volume~25, pages 626--631, 2011.

\bibitem[Conitzer et~al.(2017)Conitzer, Freeman, and Shah]{conitzer2017fair}
Vincent Conitzer, Rupert Freeman, and Nisarg Shah.
\newblock Fair public decision making.
\newblock In \emph{Proceedings of the 2017 ACM Conference on Economics and Computation}, pages 629--646, 2017.

\bibitem[Davis(1970)]{davis1970intergrability}
Burgess Davis.
\newblock On the intergrability of the martingale square function.
\newblock \emph{Israel Journal of Mathematics}, 8:\penalty0 187--190, 1970.

\bibitem[Dubey(1975)]{dubey1975uniqueness}
Pradeep Dubey.
\newblock On the uniqueness of the shapley value.
\newblock \emph{International Journal of Game Theory}, 4\penalty0 (3):\penalty0 131--139, 1975.

\bibitem[Frazier et~al.(2014)Frazier, Kempe, Kleinberg, and Kleinberg]{frazier2014incentivizing}
Peter Frazier, David Kempe, Jon Kleinberg, and Robert Kleinberg.
\newblock Incentivizing exploration.
\newblock In \emph{Proceedings of the fifteenth ACM conference on Economics and computation}, pages 5--22, 2014.

\bibitem[Hassanzadeh et~al.(2023)Hassanzadeh, Kreacic, Zeng, Xiao, and Ganesh]{hassanzadeh2023sequential}
Parisa Hassanzadeh, Eleonora Kreacic, Sihan Zeng, Yuchen Xiao, and Sumitra Ganesh.
\newblock Sequential fair resource allocation under a markov decision process framework.
\newblock In \emph{Proceedings of the Fourth ACM International Conference on AI in Finance}, pages 673--680, 2023.

\bibitem[Hossain et~al.(2020)Hossain, Mladenovic, and Shah]{Hossain20-2}
Safwan Hossain, Andjela Mladenovic, and Nisarg Shah.
\newblock Designing fairly fair classifiers via economic fairness notions.
\newblock In \emph{Proceedings of The Web Conference 2020}, WWW '20, page 1559–1569, New York, NY, USA, 2020. Association for Computing Machinery.
\newblock ISBN 9781450370233.
\newblock \doi{10.1145/3366423.3380228}.
\newblock URL \url{https://doi.org/10.1145/3366423.3380228}.

\bibitem[Hossain et~al.(2021)Hossain, Micha, and Shah]{Hossain20}
Safwan Hossain, Evi Micha, and Nisarg Shah.
\newblock Fair algorithms for multi-agent multi-armed bandits.
\newblock \emph{Advances in Neural Information Processing Systems}, 34:\penalty0 24005--24017, 2021.

\bibitem[Jones et~al.(2023)Jones, Nguyen, and Nguyen]{jones2023efficient}
Matthew Jones, Huy Nguyen, and Thy Nguyen.
\newblock An efficient algorithm for fair multi-agent multi-armed bandit with low regret.
\newblock In \emph{Proceedings of the AAAI Conference on Artificial Intelligence}, volume 37(7), pages 8159--8167, 2023.

\bibitem[Joseph et~al.(2016)Joseph, Kearns, Morgenstern, and Roth]{joseph2016fairness}
Matthew Joseph, Michael Kearns, Jamie~H Morgenstern, and Aaron Roth.
\newblock Fairness in learning: Classic and contextual bandits.
\newblock \emph{Advances in neural information processing systems}, 29, 2016.

\bibitem[Kahneman et~al.(1986)Kahneman, Knetsch, and Thaler]{Kahneman86}
Daniel Kahneman, Jack~L. Knetsch, and Richard~H. Thaler.
\newblock Fairness and the assumptions of economics.
\newblock \emph{The Journal of Business}, 59\penalty0 (4):\penalty0 S285--S300, 1986.
\newblock ISSN 00219398, 15375374.
\newblock URL \url{http://www.jstor.org/stable/2352761}.

\bibitem[Kaneko and Nakamura(1979)]{kaneko1979nash}
Mamoru Kaneko and Kenjiro Nakamura.
\newblock The nash social welfare function.
\newblock \emph{Econometrica: Journal of the Econometric Society}, pages 423--435, 1979.

\bibitem[Kash et~al.(2014)Kash, Procaccia, and Shah]{kash2014no}
Ian Kash, Ariel~D Procaccia, and Nisarg Shah.
\newblock No agent left behind: Dynamic fair division of multiple resources.
\newblock \emph{Journal of Artificial Intelligence Research}, 51:\penalty0 579--603, 2014.

\bibitem[Kleinberg et~al.(2016)Kleinberg, Mullainathan, and Raghavan]{Kleinberg16}
Jon Kleinberg, Sendhil Mullainathan, and Manish Raghavan.
\newblock Inherent trade-offs in the fair determination of risk scores, 2016.
\newblock URL \url{https://arxiv.org/abs/1609.05807}.

\bibitem[Kleinberg et~al.(2008)Kleinberg, Slivkins, and Upfal]{kleinberg2008multi}
Robert Kleinberg, Aleksandrs Slivkins, and Eli Upfal.
\newblock Multi-armed bandits in metric spaces.
\newblock In \emph{Proceedings of the fortieth annual ACM symposium on Theory of computing}, pages 681--690, 2008.

\bibitem[Kremer et~al.(2014)Kremer, Mansour, and Perry]{Kremer14}
Ilan Kremer, Yishay Mansour, and Motty Perry.
\newblock Implementing the “wisdom of the crowd”.
\newblock \emph{Journal of Political Economy}, 122\penalty0 (5):\penalty0 988--1012, 2014.
\newblock ISSN 00223808, 1537534X.
\newblock URL \url{http://www.jstor.org/stable/10.1086/676597}.

\bibitem[Levine et~al.(2017)Levine, Crammer, and Mannor]{levine2017rotting}
Nir Levine, Koby Crammer, and Shie Mannor.
\newblock Rotting bandits.
\newblock \emph{Advances in neural information processing systems}, 30, 2017.

\bibitem[Liu et~al.(2017)Liu, Radanovic, Dimitrakakis, Mandal, and Parkes]{liu2017calibrated}
Yang Liu, Goran Radanovic, Christos Dimitrakakis, Debmalya Mandal, and David~C Parkes.
\newblock Calibrated fairness in bandits.
\newblock \emph{arXiv preprint arXiv:1707.01875}, 2017.

\bibitem[Lu and Boutilier(2014)]{lu2014effective}
Tyler Lu and Craig Boutilier.
\newblock Effective sampling and learning for mallows models with pairwise-preference data.
\newblock \emph{J. Mach. Learn. Res.}, 15\penalty0 (1):\penalty0 3783--3829, 2014.

\bibitem[Luce(1959)]{luce1959individual}
R~Duncan Luce.
\newblock \emph{Individual choice behavior}, volume~4.
\newblock Wiley New York, 1959.

\bibitem[Mallows(1957)]{mallows1957non}
Colin~L Mallows.
\newblock Non-null ranking models. i.
\newblock \emph{Biometrika}, 44\penalty0 (1/2):\penalty0 114--130, 1957.

\bibitem[Mansour et~al.(2015)Mansour, Slivkins, and Syrgkanis]{Mansour15}
Yishay Mansour, Aleksandrs Slivkins, and Vasilis Syrgkanis.
\newblock Bayesian incentive-compatible bandit exploration.
\newblock In \emph{Proceedings of the Sixteenth ACM Conference on Economics and Computation}, EC '15, page 565–582, New York, NY, USA, 2015. Association for Computing Machinery.
\newblock ISBN 9781450334105.
\newblock \doi{10.1145/2764468.2764508}.
\newblock URL \url{https://doi.org/10.1145/2764468.2764508}.

\bibitem[Marden(1996)]{marden1996analyzing}
John~I Marden.
\newblock \emph{Analyzing and modeling rank data}.
\newblock CRC Press, 1996.

\bibitem[Mehrabi et~al.(2019)Mehrabi, Morstatter, Saxena, Lerman, and Galstyan]{Mehrabi19}
Ninareh Mehrabi, Fred Morstatter, Nripsuta Saxena, Kristina Lerman, and Aram Galstyan.
\newblock A survey on bias and fairness in machine learning, 2019.
\newblock URL \url{https://arxiv.org/abs/1908.09635}.

\bibitem[Menon and Williamson(2018)]{menon2018cost}
Aditya~Krishna Menon and Robert~C Williamson.
\newblock The cost of fairness in binary classification.
\newblock In \emph{Conference on Fairness, accountability and transparency}, pages 107--118. PMLR, 2018.

\bibitem[Moulin(2004)]{moulin2004fair}
Herv{\'e} Moulin.
\newblock \emph{Fair division and collective welfare}.
\newblock MIT press, 2004.

\bibitem[Patil et~al.(2021)Patil, Ghalme, Nair, and Narahari]{patil2021achieving}
Vishakha Patil, Ganesh Ghalme, Vineet Nair, and Yadati Narahari.
\newblock Achieving fairness in the stochastic multi-armed bandit problem.
\newblock \emph{The Journal of Machine Learning Research}, 22\penalty0 (1):\penalty0 7885--7915, 2021.

\bibitem[Pessach and Shmueli(2022)]{pessach2022review}
Dana Pessach and Erez Shmueli.
\newblock A review on fairness in machine learning.
\newblock \emph{ACM Computing Surveys (CSUR)}, 55\penalty0 (3):\penalty0 1--44, 2022.

\bibitem[Peters(2001)]{peters2001common}
Michael Peters.
\newblock Common agency and the revelation principle.
\newblock \emph{Econometrica}, 69\penalty0 (5):\penalty0 1349--1372, 2001.

\bibitem[Procaccia(2013)]{procaccia2013cake}
Ariel~D Procaccia.
\newblock Cake cutting: Not just child's play.
\newblock \emph{Communications of the ACM}, 56\penalty0 (7):\penalty0 78--87, 2013.

\bibitem[Ron et~al.(2021)Ron, Ben-Porat, and Shalit]{ron2021corporate}
Tom Ron, Omer Ben-Porat, and Uri Shalit.
\newblock Corporate social responsibility via multi-armed bandits.
\newblock In \emph{Proceedings of the 2021 ACM Conference on Fairness, Accountability, and Transparency}, pages 26--40, 2021.

\bibitem[Roughgarden(2010)]{Roughgarden10}
Tim Roughgarden.
\newblock Algorithmic game theory.
\newblock \emph{Commun. ACM}, 53\penalty0 (7):\penalty0 78–86, jul 2010.
\newblock ISSN 0001-0782.
\newblock \doi{10.1145/1785414.1785439}.
\newblock URL \url{https://doi.org/10.1145/1785414.1785439}.

\bibitem[Shapley and Shubik(1954)]{shapley_shubik_1954}
L.~S. Shapley and Martin Shubik.
\newblock A method for evaluating the distribution of power in a committee system.
\newblock \emph{American Political Science Review}, 48\penalty0 (3):\penalty0 787–792, 1954.
\newblock \doi{10.2307/1951053}.

\bibitem[Shapley and Scarf(1974)]{shapley1974cores}
Lloyd Shapley and Herbert Scarf.
\newblock On cores and indivisibility.
\newblock \emph{Journal of mathematical economics}, 1\penalty0 (1):\penalty0 23--37, 1974.

\bibitem[Shapley(1971)]{shapley1971cores}
Lloyd~S Shapley.
\newblock Cores of convex games.
\newblock \emph{International journal of game theory}, 1:\penalty0 11--26, 1971.

\bibitem[Simchowitz and Slivkins(2024)]{simchowitz2024exploration}
Max Simchowitz and Aleksandrs Slivkins.
\newblock Exploration and incentives in reinforcement learning.
\newblock \emph{Operations Research}, 72\penalty0 (3):\penalty0 983--998, 2024.

\bibitem[Sinclair et~al.(2022)Sinclair, Banerjee, and Yu]{sinclair2022sequential}
Sean~R. Sinclair, Siddhartha Banerjee, and Christina~Lee Yu.
\newblock Sequential fair allocation: Achieving the optimal envy-efficiency tradeoff curve.
\newblock In \emph{Abstract Proceedings of the 2022 ACM SIGMETRICS/IFIP PERFORMANCE Joint International Conference on Measurement and Modeling of Computer Systems}, SIGMETRICS/PERFORMANCE '22, page 95–96, New York, NY, USA, 2022. Association for Computing Machinery.
\newblock ISBN 9781450391412.
\newblock \doi{10.1145/3489048.3526951}.
\newblock URL \url{https://doi.org/10.1145/3489048.3526951}.

\bibitem[Slivkins(2011)]{slivkins2011contextual}
Aleksandrs Slivkins.
\newblock Contextual bandits with similarity information.
\newblock In \emph{Proceedings of the 24th annual Conference On Learning Theory}, pages 679--702. JMLR Workshop and Conference Proceedings, 2011.

\bibitem[Thurstone(1994)]{ThurstoneModel}
Louis~Leon Thurstone.
\newblock A law of comparative judgment.
\newblock \emph{Psychological Review}, 34:\penalty0 273--286, 1994.
\newblock URL \url{https://api.semanticscholar.org/CorpusID:144782881}.

\bibitem[Varian(1973)]{varian1973equity}
Hal~R Varian.
\newblock Equity, envy, and efficiency.
\newblock 1973.

\bibitem[Wang et~al.(2021)Wang, Bai, Sun, and Joachims]{wang2021fairness}
Lequn Wang, Yiwei Bai, Wen Sun, and Thorsten Joachims.
\newblock Fairness of exposure in stochastic bandits.
\newblock In \emph{International Conference on Machine Learning}, pages 10686--10696. PMLR, 2021.

\bibitem[Weitzman(1978)]{weitzman1978optimal}
Martin Weitzman.
\newblock \emph{Optimal search for the best alternative}, volume 78 (8).
\newblock Department of Energy, 1978.

\bibitem[Yue et~al.(2012)Yue, Broder, Kleinberg, and Joachims]{yue2012k}
Yisong Yue, Josef Broder, Robert Kleinberg, and Thorsten Joachims.
\newblock The k-armed dueling bandits problem.
\newblock \emph{Journal of Computer and System Sciences}, 78\penalty0 (5):\penalty0 1538--1556, 2012.

\bibitem[Zafar et~al.(2017)Zafar, Valera, Rogriguez, and Gummadi]{zafar2017fairness}
Muhammad~Bilal Zafar, Isabel Valera, Manuel~Gomez Rogriguez, and Krishna~P Gummadi.
\newblock Fairness constraints: Mechanisms for fair classification.
\newblock In \emph{Artificial intelligence and statistics}, pages 962--970. PMLR, 2017.

\bibitem[Zeng and Psomas(2020)]{zeng2020tradeoffs}
David Zeng and Alexandros Psomas.
\newblock Fairness-efficiency tradeoffs in dynamic fair division.
\newblock In \emph{Proceedings of the 21st ACM Conference on Economics and Computation}, EC '20, page 911–912, New York, NY, USA, 2020. Association for Computing Machinery.
\newblock ISBN 9781450379755.
\newblock \doi{10.1145/3391403.3399467}.
\newblock URL \url{https://doi.org/10.1145/3391403.3399467}.

\end{thebibliography}
